\documentclass[reprint,
nofootinbib,onecolumn,notitlepage,
amsmath,amssymb,
aps,longbibliography
]{revtex4-1}
\usepackage[english]{babel}
\usepackage[utf8x]{inputenc}
\usepackage[T1]{fontenc}

\AtBeginDocument{%
    \newwrite\bibnotes
    \def\bibnotesext{Notes.bib}
    \immediate\openout\bibnotes=\jobname\bibnotesext
    \immediate\write\bibnotes{@CONTROL{REVTEX41Control}}
    \immediate\write\bibnotes{@CONTROL{%
    apsrev41Control,author="08",editor="1",pages="1",title="0",year="0"}}
     \if@filesw
     \immediate\write\@auxout{\string\citation{apsrev41Control}}%
    \fi
  }%

\usepackage[nodayofweek]{datetime}
\usepackage{amsmath}
\usepackage{url}
\usepackage{graphicx}
\usepackage{amsthm}
\usepackage{dcolumn}
\usepackage{bm}
\usepackage{amsmath}
\usepackage{amssymb}
\usepackage{amsfonts}
\usepackage{graphicx}
\usepackage{lipsum}

\usepackage{hyperref}
\usepackage{xcolor}
\definecolor{mylinkcolor}{rgb}{0,0,0.8} 
\hypersetup{unicode=true,%
  bookmarksnumbered=false,bookmarksopen=false,bookmarksopenlevel=1, %
  breaklinks=true,pdfborder={0 0 0},colorlinks=true}%
\hypersetup{%
  anchorcolor=mylinkcolor,citecolor=mylinkcolor, %
  filecolor=mylinkcolor,linkcolor=mylinkcolor, %
  menucolor=mylinkcolor,runcolor=mylinkcolor, %
  urlcolor=mylinkcolor}

\newtheorem{lemma}{Lemma}
\newtheorem{theorem}{Theorem}
\newtheorem{conjecture}{Conjecture}
\newtheorem{corollary}{Corollary}

\theoremstyle{definition}
\newtheorem{definition}{Definition}

\newtheorem{protocol}{Protocol}

\newcommand{\tr}{\mathrm{tr}}
\newcommand{\ket}[1]{| #1 \rangle}
\newcommand{\bra}[1]{\langle #1 |}

\newcommand{\ketbra}[2]{\ket{#1}\!\bra{#2}}
\newcommand{\proj}[1]{\ket{#1}\!\bra{#1}}
\newcommand{\id}{\openone}
\newcommand{\cD}{\mathcal{D}}
\newcommand{\cC}{\mathcal{C}}
\newcommand{\cE}{\mathcal{E}}
\newcommand{\cF}{\mathcal{F}}
\newcommand{\cH}{\mathcal{H}}
\newcommand{\cI}{\mathcal{I}}
\newcommand{\cM}{\mathcal{M}}
\newcommand{\cN}{\mathcal{N}}
\newcommand{\cP}{\mathcal{P}}
\newcommand{\cQ}{\mathcal{Q}}
\newcommand{\cS}{\mathcal{S}}
\newcommand{\cT}{\mathcal{T}}
\newcommand{\cR}{\mathcal{R}}
\newcommand{\ot}{\otimes}
\newcommand{\bin}{\mathrm{bin}}
\newcommand{\scorefunction}{S}
\newcommand{\score}{\omega}
\newcommand{\Ext}{\mathrm{Ext}}
\newcommand{\Var}{\mathrm{Var}}
\newcommand{\Min}{\mathrm{Min}}
\newcommand{\Max}{\mathrm{Max}}
\newcommand{\Freq}{\mathrm{Freq}}
\newcommand{\rate}{\mathrm{rate}}
\DeclareMathOperator{\e}{e}
\newcommand{\dd}{\mathrm{d}} 
\newcommand{\ii}{\mathrm{i}} 

\newcommand{\comment}[1]{}
\def\M{\ensuremath\mathcal}
\def\B{\ensuremath\mathbf}

\begin{document}
\title{Improved device-independent randomness expansion rates using two sided randomness}
\author{Rutvij Bhavsar}
\affiliation{Department of Mathematics, University of York, Heslington, York, YO10 5DD, United Kingdom}
\author{Sammy Ragy}
\affiliation{Department of Mathematics, University of York, Heslington, York, YO10 5DD, United Kingdom}
\author{Roger Colbeck}
\email{roger.colbeck@york.ac.uk}
\affiliation{Department of Mathematics, University of York, Heslington, York, YO10 5DD, United Kingdom}
\date{$5^{\text{th}}$ January 2026}

\begin{abstract}
  A device-independent randomness expansion protocol aims to take an initial random string and generate a longer one, where the security of the protocol does not rely on knowing the inner workings of the devices used to run it.  In order to do so, the protocol tests that the devices violate a Bell inequality and one then needs to bound the amount of extractable randomness in terms of the observed violation. The entropy accumulation theorem lower bounds the extractable randomness of a protocol with many rounds in terms of the single-round von Neumann entropy of any strategy achieving the observed score. Tight bounds on the von Neumann entropy are known for the one-sided randomness (i.e., where the randomness from only one party is used) when using the Clauser-Horne-Shimony-Holt (CHSH) game. Here we investigate the possible improvement that could be gained using the two-sided randomness. We generate upper bounds on this randomness by attempting to find the optimal eavesdropping strategy, providing analytic formulae in two cases. We additionally compute lower bounds that outperform previous ones and can be made arbitrarily tight (at the expense of more computation time). These bounds get close to our upper bounds, and hence we conjecture that our upper bounds are tight. We also consider a modified protocol in which the input randomness is recycled. This modified protocol shows the possibility of rate gains of several orders of magnitude based on recent experimental parameters, making device-independent randomness expansion significantly more practical. It also enables the locality loophole to be closed while expanding randomness in a way that typical spot-checking protocols do not.
\end{abstract}

\maketitle
\section{Introduction}
Random numbers have a wide variety of uses. In some applications, only the distribution of the random numbers matters, while in others it is important that the generated numbers are also private, for instance when used for cryptography.  According to our current understanding of physics, generating fundamentally random numbers requires quantum processes. While it is easy to come up with a quantum process that can in theory generate random numbers, given a candidate quantum random number generator, verifying that it is indeed generating random numbers and at what rate is a difficult task that usually requires detailed understanding of the physical device. Any mismatch between the model of the device used in the security proof and the real device could in principle be exploited by an adversary.

Device-independent protocols aim to circumvent the mismatch problem by designing the protocol to abort unless the devices used are performing sufficiently well, and without needing to know the internal mechanism by which they operate.  In the context of device-independent randomness expansion (DIRE) the main idea is that the ability to violate a Bell inequality implies that the devices doing so must be generating randomness~\cite{ColbeckThesis,CK2}. Thus, in a sense, the protocol self-tests~\cite{MayersYao} the devices during its operation, leading to enhanced security.  Although challenging to accomplish, recently the first experimental demonstrations of DIRE were performed~\cite{LZL&,Shalm_rand,LLR&}, following earlier experiments considering randomness generation~\cite{PAMBMMOHLMM,BKGZM&,LZL&gen}. On the theoretical side, an increasingly sophisticated series of proofs~\cite{VV,MS1,MS2} led to the current state of the art~\cite{ZFK,ARV}. In this work we consider the proofs based on the entropy accumulation theorem~\cite{DFR,DF}, which establishes a lower bound on the amount of randomness after many Bell tests in terms of the von Neumann entropy of a single round achieving the observed score (see Appendix~\ref{app:EAT} for a precise statement). Bounds on the von Neumann entropy will hence be a focus of this work.

In a DIRE protocol, randomly chosen inputs are made to two separated devices in such a way that each device cannot learn the input chosen by the other. We use $X$ and $Y$ to label the inputs and $A$ and $B$ to label the outputs, and consider an adversary with side information $E$. This side information could be quantum, i.e., the general strategy has an adversary hold the $E$ part of a state $\rho_{A'B'E}$ where the $A'$ and $B'$ systems are held by the devices. Each input $X$ to the first device corresponds to a measurement on $A'$ giving outcome $A$ and, analogously, for each input $Y$ to the other device there is a corresponding measurement on $B'$ with outcome $B$. There are two quantities of interest, both of which depend on the post-measurement state: the first is the score in some non-local game, which is a function of the conditional distribution $p_{AB|XY}$; the second is the von Neumann entropy of either one or both of the outputs.  More precisely, we want to express the minimum von Neumann entropy as a function of the score. In this work we study six such von Neumann entropies: $H(AB|X=0,Y=0,E)$, $H(AB|XYE)$, $H(AB|E)$, $H(A|X=0,Y=0,E)$, $H(A|XYE)$ and $H(A|E)$\footnote{In the one sided cases, conditioning on $Y$ is irrelevant, but we keep this for notational symmetry.}.

The difficulty in bounding these stems from the fact that a priori there is no upper bound on the dimension of the systems $A'$, $B'$ and $E$ that need to be considered. For instance, for some Bell inequalities it is known that the maximum quantum violation cannot be achieved if $A'$ and $B'$ are finite dimensional~\cite{Slofstra}, and there is evidence that this is true even in the case where $X$ and $Y$ are binary and where $A$ and $B$ have three possible values~\cite{PV}.  However, in the case where $A$, $B$, $X$ and $Y$ are all binary, Jordan's lemma~\cite{Jordan} implies that there is no loss in generality in considering a convex combination of strategies in which $A'$ and $B'$ are two-dimensional.  This observation was used to give a tight analytic lower bound on $H(A|X=0,Y=0,E)$ in terms of the CHSH score~\cite{PABGMS} and is also crucial for the present work. This analytic bound was also extended to a family of CHSH-like inequalities in~\cite{WAP}. As well as Jordan's lemma, such bounds rely on a series of additional simplifications. In cases where such simplifications do not hold, alternative ways to lower bound the single-round von Neumann entropy are needed.  One way is to bound it using the single-round min entropy, which can be optimized at a particular level of the semi-definite hierarchy~\cite{NPA2}. Although this method can generate numerically certified lower bounds for arbitrary protocols~\cite{BRC}, there is a significant loss in tightness when moving to the single-round min entropy. More recent methods give tighter computational bounds on the von Neumann entropy~\cite{SBVTRS,BFF,BFF2022}. Forming more direct and tighter bounds is a key open problem in the field of device-independence and useful for increasing the practicality of device-independent tasks.

We discuss when each of the six entropic quantities is of most interest in Section~\ref{sec:signif}. In summary, the one-sided entropies are most-relevant in the context of device independent quantum key distribution (DIQKD), while using both outputs is useful for DIRE. The quantities that are conditioned on $X=0$ and $Y=0$ are useful for spot-checking protocols in which particular fixed measurements (taken here to be the $X=0$ and $Y=0$ measurements) are used to generate randomness/key, and where the other measurements are used only rarely to check that the devices are behaving honestly. The quantities conditioned on $X$ and $Y$ are useful for analysing protocols in which the generation rounds involve multiple settings, where the input randomness is recycled, and can also be used to close the locality loophole in randomness generation. The remaining quantities, $H(AB|E)$ and $H(A|E)$, tell us something about the fundamental randomness based on the Bell violation, and the second of these can be useful for DIQKD protocols in which the measurement settings in the generation rounds are chosen randomly and to which an adversary, Eve, does not have access.

In Section~\ref{sec:rates} we discuss the computation of the von Neumann entropy bounds before giving numerically generated upper bound curves for each of the six quantities when using the CHSH game as the Bell test in Section~\ref{sec:num}. For $H(A|XYE)$ and $H(AB|XYE)$ we additionally give conjectured analytic forms for the curves. We also compute lower bounds in the cases of $H(A|XYE)$ and $H(AB|X=0,Y=0,E)$ that are tighter than those previously known and in principle converge to these quantities by increasing the accuracy of the computation. We use our bounds with the EAT to show how the potential improvement carries over to finite statistics. To do so, in Section~\ref{sec:protocols} we consider not only the usual spot-checking type protocol, but also modified protocols. The first modification is to replace the spot-check with a biased random number generator, and the second is to remove the biasing while recycling the input randomness in order to still enable expansion. Taking the experimental conditions from~\cite{LLR&}, using the two sided randomness and randomness recycling gives a rate increase of several orders of magnitude with our new bounds.

\section{The significance of various entropic quantities}\label{sec:signif}
In this section we discuss the significance of the six entropic quantities given above in the context of DIRE, noting that the one-sided quantities are also useful for DIQKD. To do so we first describe the general structure of the raw randomness generation part of a spot-checking and non-spot-checking DIRE protocol. A more complete description of the protocols is in Section~\ref{sec:protocols}.  

In a protocol without spot-checking, there are two untrusted devices, and in every round their inputs $X_i$ and $Y_i$ are generated according to some distribution $p_{XY}$. Often two independent random number generators are used for this, so that $p_{XY}=p_Xp_Y$.  The generated numbers are used as inputs to the devices, which return two outputs $A_i$ and $B_i$ respectively. This is repeated for $n$ rounds generating the raw randomness ${\bf A}$, ${\bf B}$, where the bold font denotes the concatenation of all the outputs.

In a spot-checking protocol there is an additional step in which each round is declared to be either a test round ($T_i=1$) or a generation round ($T_i=0$), where test rounds occur with probability $\gamma$, which is typically small. On test rounds $X_i$ and $Y_i$ are generated according to some distributions $p_{XY}$. On generation rounds, $X_i$ and $Y_i$ are set according to some other distribution -- in this work we use the deterministic distribution $X_i=Y_i=0$. These are used as inputs to the devices, which return two outputs $A_i$ and $B_i$ respectively. The rationale behind using a spot-checking protocol is that randomness is required to perform a Bell test and it is desirable to be able to run the protocol with a smaller requirement on the amount of input randomness required. Choosing whether to test or not requires roughly $H_\bin(\gamma)$ bits of randomness per round\footnote{Here $H_\bin$ denotes the binary entropy.}, so choosing $\gamma$ small enough leads to an overall saving.  Furthermore, protocols often discard the input randomness, in which case for many Bell tests spot-checking is necessary in order to achieve expansion. In the CHSH game, for instance, if $p_{XY}$ is chosen uniformly, each test round requires $2$ bits of randomness, but the amount of two-sided randomness output by the quantum strategy with the highest possible winning probability is only $1+H_\bin(\frac{1}{2}(1+\frac{1}{\sqrt{2}}))\approx1.60$ bits. However, as we discuss later, the input randomness need not be discarded.

In the case of small $\gamma$, almost every round is a generation round and hence an eavesdropper wishes to guess the outputs for the inputs $X=0$ and $Y=0$. The entropy $H(AB|X=0,Y=0,E)$ is thus the relevant quantity for spot-checking DIRE protocols.  The one-sided quantity $H(A|X=0,Y=0,E)$ has often been used instead because of the existing analytic bound for this~\cite{PABGMS,WAP}, but, because this ignores one of the outputs, it is wasteful as an estimate of the generated randomness. For DIQKD protocols, on the other hand, the one-sided entropy is the relevant quantity, since to make key Alice's and Bob's strings should match. We also remark that these quantities can be useful bounds for protocols without spot-checking if the distribution $p_{XY}$ is heavily biased towards $X=Y=0$.

The quantities $H(AB|XYE)$ and $H(A|XYE)$ are useful for protocols without spot checking. One might imagine, for example, using a source of public randomness, such as a randomness beacon to choose the inputs to the protocol, in which case $X$ and $Y$ become known to the adversary (but are not known before the devices are prepared).  In this case, rather than being interested in randomness expansion, the task is to turn public randomness into private randomness in a device-independent way. One can also use $H(AB|XYE)$ and $H(A|XYE)$ in protocols when the input randomness is recycled. In this case we are really interested in $H(ABXY|E)$, but, because $X$ and $Y$ are chosen independently of $E$, this can be expanded as $H(XY)+H(AB|XYE)$. Hence $H(AB|XYE)$ is the relevant quantity in this case as well. The one-sided quantity $H(A|XYE)$ could also be used for DIQKD without spot-checking.

In addition we consider the quantities $H(AB|E)$ and $H(A|E)$.  The second of these could be useful for QKD protocols in which the key generation rounds do not have a fixed input and where Alice and Bob do not publicly reveal their measurement choices during the protocol.  For instance, the sharing of these choices could be encrypted using an initial key, analogously to a suggested defence against memory attacks~\cite{bckone}\footnote{Note that such protocols would only be useful if more key is generated than is required, so the protocol we are thinking of here is really quantum key expansion. Furthermore, the results presented in Figure~\ref{fig:rates} show that the use of $H(A|E)$ only gives a minor advantage over $H(A|XYE)$.}. $H(AB|E)$ would be a useful quantity for randomness generation in a protocol without spot-checking and in which $X$ and $Y$ are kept private after running the protocol and not used in the overall output.  When such protocols are based on the CHSH inequality, they cannot allow expansion.  These quantities can also be thought of as quantifying the fundamental amount of randomness obtainable from a given Bell violation. Although we have computed the graphs for $H(AB|E)$ and $H(A|E)$, existing versions of the EAT cannot be directly applied to them --- see Appendix~\ref{app:EAT}.

\section{Rates for CHSH-based protocols}\label{sec:rates}
We calculate various one sided and two sided rates for protocols based on the CHSH game, which involves trying to violate the CHSH Bell inequality~\cite{CHSH}. In this game, each party makes a binary input and receives a binary output and the game is won if $A\oplus B=XY$. We define the CHSH score by
$$\frac{1}{4}\left(\sum_ap_{AB|00}(a,a)+\sum_ap_{AB|01}(a,a)+\sum_ap_{AB|10}(a,a)+\sum_ap_{AB|11}(a,a\oplus1)\right)\,,$$
which is the probability of winning the CHSH game when the inputs are chosen at random\footnote{Note that even if nonuniform distributions of inputs are used when running protocols, in this work the CHSH score is always defined as here.}.
Classical strategies can win this game with probability at most $3/4$, while quantum strategies can get as high as $\frac{1}{2}\left(1+\frac{1}{\sqrt{2}}\right)\approx0.85$. For a fixed CHSH game score, $\score$, we wish to compute the minimum von Neumann entropy over all strategies achieving that score. In this context a strategy comprises three Hilbert spaces $\cH_{A'}$, $\cH_{B'}$ and $\cH_E$, POVMs $\{M_{a|x}\}_a$ on $\cH_{A'}$ for both $x=0$ and $x=1$, POVMs $\{N_{b|y}\}_b$ on $\cH_{B'}$ for both $y=0$ and $y=1$, and a state $\rho_{A'B'E}$ on $\cH_{A'}\ot\cH_{B'}\ot\cH_E$.  Given a strategy, and a distribution $p_{XY}$ there is an associated channel $\cN$ that acts on $A'B'$ and takes the state to the post-measurement state, i.e.,
$$\tau_{ABXYE}=(\cN\ot\cI_E)(\rho_{A'B'E})=\sum_{abxy}p_{XY}(x,y)\proj{a}_A\ot\proj{b}_B\ot\proj{x}_X\ot\proj{y}_Y\ot\tr_{A'B'}\left((M_{a|x}\ot N_{b|y}\ot \id_E)\rho_{A'B'E}\right)\,,$$
where $\cI_E$ is the identity channel on $E$. The entropic quantities we consider all pertain to this state\footnote{In the cases where we condition on $X=0$ and $Y=0$, we can project this state onto $\proj{0}_X\ot\proj{0}_Y$ and renormalize --- see Appendix~\ref{app:simp} for more detail.}. Note also that the score is a function of $\tau$, which we denote $\scorefunction(\tau)$ -- in particular, $p_{AB|xy}(a,b)=\tr\left((M_{a|x}\ot N_{b|y}\ot \id_E)\rho_{A'B'E}\right)$.

For each of the six entropic quantities previously discussed we consider the infimum over all strategies that achieve a given score.  We use this to define a set of curves.  We write $F_{AB|XYE}(\score,p_{XY})=\inf H(AB|XYE)_{\tau}$, where the infimum is over all strategies for which $\scorefunction(\tau)=\score$. In the same way we define $F_{AB|E}(\score,p_{XY})$, $F_{A|XYE}(\score,p_{XY})$ and $F_{A|E}(\score,p_{XY})$, replacing the objective function by the corresponding entropy. We also define $F_{AB|00E}(\score)$ and $F_{A|00E}(\score)$ analogously, noting that these are independent of $p_{XY}$. Furthermore, if we write $F_{AB|XYE}(\score)$ etc.\ (i.e., leaving out the $p_{XY}$), we refer to the case where $p_{XY}$ is uniform over $X$ and $Y$. For a more precise writing of these optimizations, see~\eqref{opt_1}.

We also consider a related set of functions $G_{AB|XYE}(\score,p_{XY})$, $G_{AB|E}(\score,p_{XY})$ etc.\ that are defined analogously, but while optimizing over a smaller set of allowed strategies. More precisely, the $G$ functions are defined by restricting $\cH_{A'}$ and $\cH_{B'}$ to be two dimensional and $\cH_E$ to have dimension 4, taking $\rho_{A'B'E}$ to be pure with $\rho_{A'B'}$ diagonal in the Bell basis, and taking the POVMs to be projective measurements onto states of the form $\cos(\alpha)\ket{0}+\sin(\alpha)\ket{1}$ (see~\eqref{opt_simp2}).

It turns out that $F_{AB|00E}(\score)=G_{AB|00E}(\score)$, $F_{A|00E}(\score)=G_{A|00E}(\score)$, and that in each of the other four cases $F$ is formed from $G$ by taking the convex lower bound. The underlying reason for this is that Jordan's lemma~\cite{Jordan} implies that in the case of Bell inequalities with two inputs and two outputs, any strategy is equivalent to a convex combination of strategies in which $A'$ and $B'$ are qubit systems. This means that if we solve the qubit case, the general case follows by taking the convex lower bound\footnote{The convex lower bound is not needed in the cases where we condition on $X=0$ and $Y=0$ because $G$ is already convex in these cases.}. Such an argument was made in~\cite{PABGMS} and we give the details in Appendix~\ref{app:Proof1}.

A note on notation: in this work we measure entropies in bits, taking $\log$ to represent the logarithm base 2, and $\ln$ for the natural logarithm where needed.

\section{Numerically computing rates}\label{sec:num}
The optimizations that define the $G$ functions can be expressed in terms of $7$ real parameters ($3$ to specify the state and $4$ to choose the measurements). They are hence amenable to numerical optimizations. We note also that except in the cases $G_{AB|E}(\score,p_{XY})$ and $G_{A|E}(\score,p_{XY})$ we can remove an additional parameter. We discuss simple ways to write the entropic expressions in Appendix~\ref{app:simp}.

\subsection{Upper bounds}
We obtain upper bounds by using numerical solvers that attempt to compute $G$ (these give upper bounds because the computations are not guaranteed to converge).  Our program for computing $G$ runs in $N$ iterations. In each iteration the program starts by making a random guess for the parameters from within the valid range. It then uses sequential quadratic programming to minimize $\bar{H}((\cN\ot\cI_E)(\rho_{A'B'E}))$ subject to the CHSH score being fixed [here $\bar{H}$ is a placeholder for one of the entropic quantities of interest]. On each iteration, the program arrives at a candidate for the minimum value, and we run $N\approx10^4$ iterations to arrive at the conjectured minimum value for $\bar{H}((\cN\ot\cI_E)(\rho_{A'B'E}))$. The numerical optimization is performed in Python using the sequential least squares programming (SLSQP) solver in SciPy. The curves obtained were found to match those generated by solving numerically in {\sc Mathematica}.

Since these optimizations are not guaranteed to converge, the generated curves are upper bounds on the infima. Some confidence of their tightness comes from the smoothness of the curves, the consistency across different numerical solvers, and that the generated points match the known analytic tight bound in the case $H(A|X=0,Y=0,E)$. They also closely match the numerical lower bounds we computed for $G_{A|XYE}$ and $G_{AB|X=0,Y=0,E}$ discussed in Section~\ref{sec:lower}.

$G_{AB|00E}$ and $G_{A|00E}$ are convex functions, and hence $F=G$ for these. For the other cases we generate the graphs in the case where $p_{XY}$ is uniform, observing that each of the $G$ curves starts with a concave part and switches to convex for larger CHSH scores. Since the minimum entropy is always zero for classical scores, each of the $G$ curves approach $0$ as $\score$ approaches $3/4$. Each of the $F$ curves can be found from $G$ by finding the tangent to $G$ that passes through $(3/4,0)$.  We call the score at which this tangent is taken $\score^*$, defined by $(\score^*-3/4)G'(\score^*)=G(\score^*)$. We then have
\begin{equation}
    F(\score) =     \begin{cases}
    G'(\score^*) (\score -\frac{3}{4}) & \text{if}\ \score \leq \score^* \\
    G(\score) & \text{otherwise} \\
    \end{cases}\,.
\end{equation}
We give estimates for $\score^*$ for each of the cases below. In essence, what this means is that for $\score<\score^*$ the optimal strategy for Eve is to either use a deterministic classical strategy with score $3/4$ or a strategy that achieves score $\score^*$, mixing these such that the average score is $\score$. Eve can remember which strategy she used, and hence the entropy from her perspective is also the convex mixture of the entropies of the endpoints.

Figure~\ref{fig:rates} shows the curves we obtained for the functions $F$ in each of the six cases. Note that, except in the cases where we condition on $X=0$ and $Y=0$, the graphs all have linear sections as a result of taking the convex lower bound. In Appendix~\ref{app:graphs} we show the graphs for $G$ together with those for $F$. The approximate coordinate of the top of the linear segment for $F_{AB|E}$ is $(0.8523,1.8735)$ and for $F_{A|E}$ it is $(0.8505,0.967)$. Note also that $1+H_\bin\left(\frac{1}{2}+\frac{1}{\sqrt{32}}\right)\approx1.908$ is the maximum value on the graph $F_{AB|E}(\score)$.

\begin{figure}
\includegraphics[width=\textwidth]{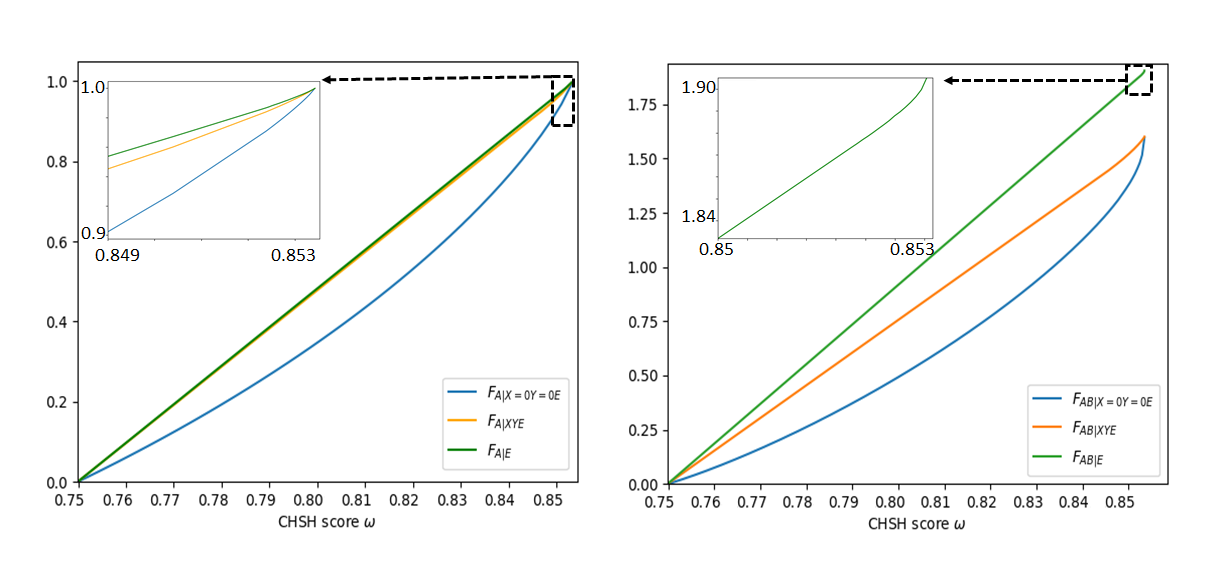}\\
\phantom{m}(a)\hspace{26em}(b)
\caption{Graphs of the rates for (a) the one-sided and (b) the two-sided randomness with uniformly chosen inputs. Each of these curves has a non-linear part and the blue curves do not have a linear part.}
\label{fig:rates}
\end{figure}

By examining the parameters that come out of the numerical optimizations we have the following.

\begin{lemma}\label{lem:ABgXYE}
Consider the curve $g_1(\score)=1+H_\bin(\score)-2H_\bin(\frac{1}{2}+\frac{2\score-1}{\sqrt{2}})$. $F_{AB|XYE}(\score)$ can be upper bounded in terms of $g_1$ as follows
\begin{equation}\label{eq:main_result}
F_{AB|XYE}(\score)\leq\begin{cases}g_1(\score)&\score_{AB|XYE}^*\leq\score\leq\frac{1}{2}\left(1+\frac{1}{\sqrt{2}}\right) \\
g_1'(\score_{AB|XYE}^*)(\score-3/4)&3/4\leq\score\leq\score_{AB|XYE}^*\end{cases}\,.
\end{equation}
where $\score_{AB|XYE}^*\approx0.84403$ is the solution to $g'_1(\score)(\score-3/4)=g_1(\score)$. Note that $g_1(\score_{AB|XYE}^*)\approx1.4186$ and the maximum value reached is $1+H_\bin(1/2+1/(2\sqrt{2}))\approx1.601$.
\end{lemma}
\begin{proof}
We first consider an upper bound on $G_{AB|XYE}(\score)$. In Appendix~\ref{app:simp} we give a parameterization of a two-qubit state (with parameters $R$, $\theta$ and $\delta$) and measurements (with parameters $\alpha_0$, $\alpha_1$, $\beta_0$ and $\beta_1$) before computing an expression for $H(AB|XYE)$ in terms of these (cf.~\eqref{eq:HABgXYE}).  We also obtain an expression for the CHSH score (cf.~\eqref{eq:score}).  Choosing $R=\sqrt{2}(2\score-1)$, $\theta=0$, $\delta=R^2/4$, $\alpha_0=0$, $\alpha_1=\pi/4$, $\beta_0=\pi/8$, $\beta_1=-\pi/8$ we find a score $\score$, and calculating $H(AB|XYE)$ we obtain $H(AB|XYE)=g_1(\score)$ and hence $G_{AB|XYE}(\score)\leq g_1(\score)$. Since, $G_{AB|XYE}(3/4)=0$, and $F_{AB|XYE}$ is formed from $G_{AB|XYE}$ by taking the convex lower bound, we establish the claim.
\end{proof}

\begin{lemma}\label{lem:AgXYE}
Consider the curve $g_2(\score)=1-H_\bin\left(\frac{1}{2}+\frac{2\score-1}{\sqrt{2}}\right)$.
$F_{A|XYE}(\score)$ is upper bounded by the convex lower bound of $g_2(\score)$. In other words,
\begin{equation}
F_{A|XYE}(\score)\leq\begin{cases}g_2(\score)&\score_{A|XYE}^*\leq\score\leq\frac{1}{2}\left(1+\frac{1}{\sqrt{2}}\right) \\
g_2'(\score_{A|XYE}^*)(\score-3/4)&3/4\leq\score\leq\score_{A|XYE}^*\end{cases}\,.
\end{equation}
where $\score_{A|XYE}^*\approx0.84698$ is the solution to $g'_2(\score)(\score-3/4)=g_2(\score)$. Note that $g_2(\score_{A|XYE}^*)\approx0.92394$.
\end{lemma}
\begin{proof}
The proof is the same as for Lemma~\ref{lem:ABgXYE}, except that $g_1$ is replaced by $g_2$ --- the choice of state and measurements remains the same.
\end{proof}

\subsection{Lower bounds}\label{sec:lower}
Lemma~\ref{lem:AgXYE} gives an upper bound on the one-sided randomness using an explicit strategy. However, for security proofs a lower bound is needed. We compute such lower bounds for $G_{A|XYE}$ and $G_{AB|X=0,Y=0,E}$. The relevant computations can be found in Appendices~\ref{app:AgXYE} and~\ref{app: two sided proof} (Corollary~\ref{cor:onesidelower} and Lemma~\ref{lem: twosided lower bound}) and are displayed in Fig.~\ref{fig:rates_conj} alongside the upper bounds and a lower bound from~\cite{BFF2022}.

The idea is behind our lower bounds is as follows. We first show that for every fixed value of $\score$ the functions $G_{A|XYE}(\score)$ and $G_{AB|X=0,Y=0,E}(\score)$ can each be expressed as a minimization over $3$ real parameters. For fixed $\score$ we compute the values of the objective function on a grid of points comprising these parameters. By bounding the derivative of the objective function within the cuboids generated by the grid we establish a lower bound on the function over the possible parameters. The lower bound we generate can in principle be made arbitrarily good by decreasing the grid spacing (at the expense of taking more time to evaluate).

Given lower bounds on $G_{A|XYE}(\score)$ and $G_{AB|X=0,Y=0,E}(\score)$ for a finite set of values of $\score$, we can get lower bounds for all values of $\score$ by using that the $G$ functions are monotonically increasing in $\score$, so we have $G(\score)\leq G(\score-\nu)$, where $\nu$ is the spacing between the finite set of values of $\score$. Hence, we consider forming lower bounds as above for a set of values $\M{W}=\{\score_1,\score_2,\ldots\}$ in the range $(3/4,(1/2)(1+(1/2)^{1/2})]$. We can then consider the points $\{(\score_1,0),(\score_2,G(\score_1)),(\score_3,G(\score_2)),\ldots\}$, i.e., where each is shifted one place. Taking the convex lower bound of these shifted points gives a convex lower bound for $G_{A|XYE}$ and $G_{AB|X=0,Y=0,E}$.  By taking more points in the set $\M{W}$ tighter lower bounds can be obtained. 

Lower bounds generated in this way are shown in Fig.~\ref{fig:rates_conj}, and can be seen to be close to the upper bounds. In Fig.~\ref{fig:rates_conj}(b) we also compare with a lower bound on $G_{AB|X=0,Y=0,E}$ from~\cite{BFF2022}. The lower bounds from our technique can be improved by refining the partition of the domain at the expense of increasing the computational time required. As seen in Fig.~\ref{fig:rates_conj}(b), refining the partition moves the lower bound closer to the upper bound, leading us to the following conjecture.
\begin{conjecture}\label{conj:ABgXYE}
The upper bounds in Lemmas~\ref{lem:ABgXYE} and~\ref{lem:AgXYE} are tight.
\end{conjecture}

One technique for generating finite key rates is to use so-called min-tradeoff functions, which are affine lower bounds on the von Neumann entropy (see Appendix~\ref{app:EAT} and~\cite{DFR,DF}). Based on an experimental setup, an affine lower bound on $G_{AB|X=0,Y=0,E}$ can be generated at the value of $\score$ corresponding to that obtainable in the experiment, hence a refined partition can be used, based only on values of $\score$ close to the experimental value.

\begin{figure}[h!]
\includegraphics[width=\textwidth]{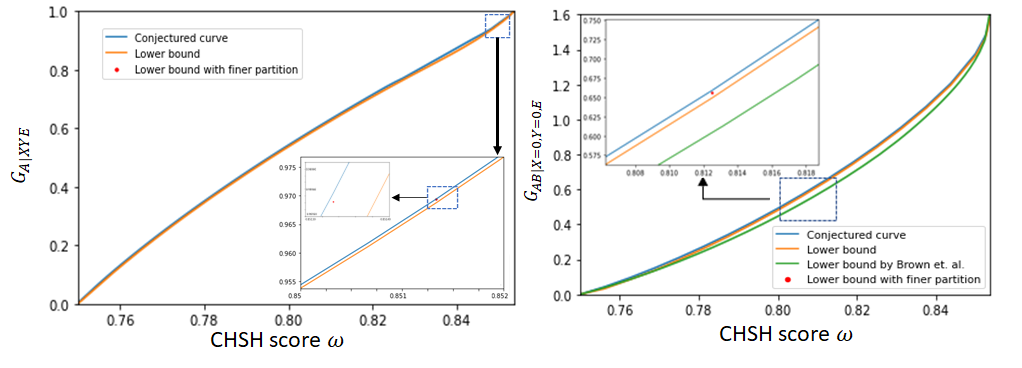}\\
\phantom{m}(a)\hspace{26em}(b)
\caption{Graphs of the conjectured rates and lower bounds for (a) $G_{A|XYE}$ (b) $G_{AB|X=0,Y=0,E}$  with uniformly chosen inputs. For $G_{AB|X=0,Y=0,E}$ we also show a lower bound from Brown et al.~\cite{BFF2022}. We also demonstrate that the lower bound for $G_{AB|X=0,Y=0,E}$ can be tightened by refining the partitioning of the domain for a specific point (due to the increased computation time, we did not do this throughout).}
\label{fig:rates_conj}
\end{figure}

\section{Protocols for randomness expansion}\label{sec:protocols}
In this section we discuss CHSH-based protocols for DIRE of both the spot-checking and non spot-checking types. We pick specific protocols for concreteness, but there are many possible variations. For instance, the protocols we discuss condense the observed statistics to a single score, but this is not necessary, and in some cases and for some sets of experimental conditions it can be advantageous to use multiple scores~\cite{BRC,TSGPL}.

Before getting to the protocols, we first describe the setup, assumptions and security definition. Although DIRE requires no assumptions on how the devices used operate, the setup for DIRE involves a user who performs the protocol within a secure laboratory, from which information cannot leak. Individual devices can also be isolated within their own sub-laboratory and the user can ensure that these devices only learn the information necessary for the protocol (in particular, they cannot learn any inputs given to other devices).  The user has access to a trusted classical computer and an initial source (or sources) of trusted randomness.

The quantum devices used for the protocol are only limited by the laws of quantum theory and may share arbitrary entanglement with each other and with an adversary. However, they cannot communicate with each other, or to the adversary after the protocol starts. Furthermore, we assume they are kept isolated after the protocol (cf.\ the discussion in Appendix~\ref{app:compos}).

For security of the protocols, we use a composable security definition. Consider a protocol with output $Z$ and use $\Omega$ to denote the event that it does not abort. The protocol is $(\epsilon_S,\epsilon_C)$-secure if
\begin{enumerate}
\item $\frac{1}{2}p_{\Omega} || \rho_{ZE|\Omega} - \frac{1}{d_Z}\id_Z \ot \rho_{E|\Omega} ||_1 \leq \epsilon_S$,
  where $E$ represents all the systems held by an adversary and $d_Z$ is the dimension of system $Z$; and
  \item There exists a quantum strategy such that $p_{\Omega}\geq1-\epsilon_C$.
\end{enumerate}
Here $\epsilon_S$ is called the soundness error, and $\epsilon_C$ is the completeness error.

\subsection{CHSH-based spot-checking protocol for randomness expansion}
We now describe a spot-checking protocol for randomness expansion. It uses a central biased random number generator $R_T$ and two other random number generators, $R_A$ and $R_B$ that are near each of the devices used to run the protocol.

\begin{protocol}\label{prot:spotcheck}{\bf (Spot-checking protocol)}
\phantom{a}\\

\noindent{\bf Parameters:}\\
$n$ -- number of rounds\\
$\gamma$ -- test probability\\
$\score_{\exp}$ -- expected CHSH score\\
$\delta$ -- confidence width for the score
\begin{enumerate}
    \item Set $i=1$ for the first round, or increase $i$ by 1.
    \item Use $R_T$ to choose $T_i\in\{0,1\}$ where $T_i=1$ occurs with probability $\gamma$.
    \item If $T_i=1$ (test round), $R_A$ is used to choose $X_i$ uniformly, which is input to one device giving output $A_i$. Likewise $R_B$ is used to choose $Y_i$ uniformly, which is input to the other device giving output $B_i$. Set $U_i=1$ if $A_i\oplus B_i=X_iY_i$ and $U_i=0$ otherwise.
    \item If $T_i=0$ (generation round), the devices are given inputs $X_i=Y_i=0$, and return the outputs $A_i$ and $B_i$. Set $U_i=\bot$.
    \item Return to Step~1 unless $i=n$.
    \item Calculate the number of rounds in which $U_i=0$ occurred, and abort the protocol if this is larger than $n\gamma(1-\score_{\exp}+\delta)$.
    \item \label{st:7} Process the concatenation of all the outputs with a quantum-proof strong extractor $\Ext$ to yield $\Ext({\bf AB},{\bf R})$, where ${\bf R}$ is a random seed for the extractor. Since a strong extractor is used, the final outcome can be taken to be the concatenation of ${\bf R}$ and $\Ext({\bf AB},{\bf R})$.
\end{enumerate}
\end{protocol}

There are a few important points to take into account when running the protocol. Firstly, it is crucial that each device only learns its own input and not the value of the other input, or of $T_i$. If this is not satisfied it is easy for devices to pass the protocol without generating randomness. Secondly, for implementations in which devices can fail to record outcomes when they should, it is important to close the detection loophole, which can be done by assigning an outcome, say $0$, when a device fails to make a detection and otherwise using the same protocol.

In order to run the protocol, some initial randomness is needed to choose which rounds are test rounds, to choose the inputs in the test rounds and to seed the extractor. Since the extractor randomness forms part of the final output, it is not consumed in the protocol, so for considering the rate at which the protocol consumes randomness we can work out the amount of uniform randomness needed to supply the inputs.  Using the rounded interval algorithm~\cite{HaoHoshi} to make the biased random number generator, $n(H_\bin(\gamma)+2\gamma)+3$ is the expected amount of input randomness required. To achieve expansion the number of output bits must be greater than this. We use the entropy accumulation theorem (EAT) to lower bound the amount of output randomness. Asymptotically the relevant quantity is $H(AB|X=0,Y=0,E)$. The quantity $H(A|X=0,Y=0,E)$ acts as a lower bound for this, and can be used in its place if convenient, for instance in analyses that are more straightforward with an analytic curve.

\subsection{CHSH-based protocols without spot-checking}
In this section we discuss two such protocols.  Protocol~\ref{prot:biased} uses two biased local random number generators to choose the inputs on each round.  Protocol~\ref{prot:nonspotcheck} eliminates the bias, but also recycles the input randomness.  Recycling the input randomness is necessary when unbiased random number generators are used, since otherwise more randomness is required to run the protocol than is generated. Protocol~\ref{prot:nonspotcheck} gives the highest randomness generation rate of all the protocols we discuss.

\begin{protocol}\label{prot:biased} {\bf (Protocol with biased local random number generators)}
\phantom{a}\\

\noindent\textbf{Parameters}:\\
$n$ -- number of rounds \\
$\zeta^A$ -- probability of 1 for random number generator $R_A$ (taken to be below $1/2$)\\ 
$\zeta^B$ -- probability of 1 for random number generator $R_B$ (taken to be below $1/2$)\\ 
$\omega_{\text{exp}}$ -- expected CHSH score. \\
$\delta$ -- confidence widths for each score. 
\begin{enumerate}
    \item Set $i=1$ for the first round, or increase $i$ by 1.
    \item Use $R_A$ to choose $X_i\in\{0,1\}$, which is input to one of the devices giving output $A_i\in\{0,1\}$. Likewise use $R_B$ to generate $Y_i\in\{0,1\}$, which is input to the other device giving output $B_i\in\{0,1\}$. Here $X_i=1$ occurs with probability $\zeta^A$ and $Y_i=1$ occurs with probability $\zeta^B$. Set $U_i=(X_i,Y_i,1)$ if $A_i\oplus B_i=X_iY_i$ and $U_i=(X_i,Y_i,0)$ otherwise.
    \item Return to Step~1 unless $i=n$.
    \item \label{st:omega} Compute the value
    \begin{eqnarray}\label{eq:omega}
    \score=\frac{1}{4}\sum_{x,y} \frac{|\{i:U_i=(x,y,1)\}|}{n p_{X}(x) p_Y(y)}
    \end{eqnarray}
    and abort the protocol if $\score<\score_{\exp}-\delta$.
    Here $p_{X}(1)=\zeta^A$, $p_{X}(0)=1-\zeta^A$, $p_{Y}(1)=\zeta^{B}$ and $p_{Y}(0)=1-\zeta^{B}$.
    \item Process the concatenation of all the outputs with a quantum-proof strong extractor $\Ext$ to yield $\Ext({\bf AB},{\bf R})$, where ${\bf R}$ is a random seed for the extractor. Since a strong extractor is used, the final outcome can be taken to be the concatenation of ${\bf R}$ and $\Ext({\bf AB},{\bf R})$.
\end{enumerate}
\end{protocol}

Note that the quantity $|\{i:U_i=(x,y,1)\}|/(n p_{X}(x) p_Y(y))$ in~\eqref{eq:omega} is an estimate of the probability of winning the CHSH game for inputs $X=x$ and $Y=y$, and hence the $\score$ computed in Step~\ref{st:omega} is an estimate of the CHSH value that would be observed if the same setup was used but with $X$ and $Y$ chosen uniformly.

The input randomness required per round in this protocol is roughly $H_\bin(\zeta^A)+H_\bin(\zeta^B)$.
To quantify the amount of output randomness (before randomness extraction is performed), in the asymptotic limit similar to the spot checking protocol, the relevant operational quantity is the von Neumann entropy $H(AB|XYE)$. Expansion hence cannot be achieved if $H(AB|XYE)-H_\bin(\zeta^A)-H_\bin(\zeta^B)<0$, which places constraints on the pairs of possible $(\zeta^A,\zeta^B)$. For $\zeta^A$ and $\zeta^B$ smaller than $1/2$, the quantity $H(AB|XYE)-H_\bin(\zeta^A)-H_\bin(\zeta^B)$ increases as $\zeta^A$ and $\zeta^B$ decrease, and hence we want to take these to be small. They only need to be large enough to ensure that $X=1,Y=1$ occurs often enough to give a good estimate of the empirical score.

Since
\begin{eqnarray}
H(AB|XYE) &=&  \sum_{xy} p_{XY}(x,y)H(AB|XYE) \nonumber\\
&\geq&\min_{x,y}H(AB|X=x,Y=y,E)   \label{Bakchodi} 
\end{eqnarray}
we can use the bounds formed for $H(AB|X=0,Y=0,E)$ instead\footnote{There is nothing special about the choice $X=0$ and $Y=0$ when computing the bounds for $H(AB|X=0,Y=0,E)$.}, albeit with a loss of entropy (this loss of entropy is small if $\zeta_A$ and $\zeta_B$ are small).

One reason for using Protocol~\ref{prot:biased} rather than Protocol~\ref{prot:spotcheck} is that the former enables the locality loophole to be closed while expanding randomness.  In order to perform the Bell tests as part of a device-independent protocol we need to make inputs to two devices in such a way that neither device knows the input of the other. One way to ensure this is by using independent random number generators on each side of the experiment, and ensuring the the outcome of each device is given at space-like separation from the production of the random input to the other. Although space-like separation can provide a guarantee (within the laws of physics) that each device does not know the input of the other\footnote{Provided we have a reasonable way to give a time before which the output of $R_A$ and $R_B$ did not exist.}, in a cryptographic setting it is necessary to assume a secure laboratory to prevents any unwanted information leaking from inside the lab to an eavesdropper. The same mechanism by which the lab is shielded from the outside world can be used to shield devices in the lab from one another and hence can prevent communication between the two devices during the protocol.  However, although unnecessary for cryptographic purposes, it is interesting to consider closing the locality loophole while expanding randomness.

This is not possible in a typical spot-checking protocol, where a central random number generator is used to decide whether a round is a test round or not. Considering Protocol~\ref{prot:spotcheck}, the locality loophole can be readily closed during the test rounds, but the use of the central random number generator means that, if one is worried that hidden communication channels are being exploited, there is a loophole that the devices could behave differently on test rounds and generation rounds. For instance, measurement devices that know whether a round is a test or generation round could supply pre-programmed outputs in generation rounds, while behaving honestly in test rounds. Thus, spot-checking protocols do not enable fully closing the locality loophole while expanding randomness.

When using Protocol~\ref{prot:biased} with $\zeta^A=\zeta^B=\zeta$, the main difference to Protocol~\ref{prot:spotcheck} is that the distribution of $X$ and $Y$ is $((1-\zeta)^2,\zeta(1-\zeta),\zeta(1-\zeta),\zeta^2)$ rather than $(1-3\gamma/4,\gamma/4,\gamma/4,\gamma/4)$. In the analysis this manifests itself in the statistics, and the much lower probability of $X=1,Y=1$ requires an adjustment of $\delta$ to achieve the same error parameters for the protocol. A comparison between the output rates for Protocols~\ref{prot:spotcheck} and~\ref{prot:biased} is shown in Figures~\ref{fig:EAT_n} and~\ref{fig:EAT_score}.\bigskip

\begin{protocol}\label{prot:nonspotcheck} {\bf (Protocol with recycled input randomness)}
\phantom{a}\\

\noindent\textbf{Parameters}:\\
$n$ -- number of rounds \\
$\score_{\exp}$ -- expected CHSH score. \\
$\delta$ -- confidence width.
\begin{enumerate}
    \item Set $i=1$ for the first round, or increase $i$ by 1.
    \item Use $R_A$ to choose $X_i\in\{0,1\}$ uniformly, serving as the input to one of the devices giving output $A_i\in\{0,1\}$. Likewise use $R_B$ to generate $Y_i\in\{0,1\}$ uniformly, which is input to the other device giving output $B_i\in\{0,1\}$. Set $U_i=1$ if $A_i\oplus B_i=X_iY_i$ and $U_i=0$ otherwise.
    \item Return to Step~1 unless $i=n$.
    \item Count the number of rounds for which $U_i=0$ occurred and abort the protocol if this is above $n(1-\score_{\exp}+\delta)$.
    \item Process the concatenation of all the inputs and outputs with a quantum-proof strong extractor $\Ext$ to yield $\Ext({\bf ABXY},{\bf R})$, where ${\bf R}$ is a random seed for the extractor. Since a strong extractor is used, the final outcome can be taken to be the concatenation of ${\bf R}$ and $\Ext({\bf ABXY},{\bf R})$.
\end{enumerate}
\end{protocol}

An important difference in this protocol compared to Protocols~\ref{prot:spotcheck} and~\ref{prot:biased} is in the extraction step, which now extracts randomness from the input strings ${\bf X}$ and ${\bf Y}$ as well as the outputs. Without recycling the inputs expansion would not be possible in Protocol~\ref{prot:nonspotcheck}. With the modification the relevant quantity to decide the length of the output is $H(ABXY|E)$, and so $H(AB|XYE)=H(ABXY|E)-H(XY)=H(ABXY|E)-2$ is the relevant quantity for calculating the rate of expansion. Note that in order to reuse the input in a composable way, it also needs to be run through an extractor~\cite{CK2} (for a discussion of why it is important to do so and a few more composability-related issues, see Appendix~\ref{app:compos}).

We could also consider an adaptation of Protocol~\ref{prot:spotcheck} in which the input randomness is recycled, forming Protocol~\ref*{prot:spotcheck}$'$ from Protocol~\ref{prot:spotcheck} by replacing Step~\ref{st:7} by
\begin{enumerate}
\item[\ref*{st:7}$'$.] Process the concatenation of all the inputs and outputs with a quantum-proof strong extractor $\Ext$ to yield $\Ext({\bf ABXY},{\bf R})$, where ${\bf R}$ is a random seed for the extractor. Since a strong extractor is used, the final outcome can be taken to be the concatenation of ${\bf R}$ and $\Ext({\bf ABXY},{\bf R})$.
\end{enumerate}
In this case, as the number of rounds, $n$, increases the advantage gained by this modification decreases, becoming negligible asymptotically. This is because as $n$ increases, the value of $\gamma$ required to give the same overall security tends to zero, and hence the amount of input randomness required becomes negligible. Note that recycling the input randomness in Protocol~\ref{prot:biased} in the case where $\zeta^A=\zeta^B=1/2$ is equivalent to Protocol~\ref{prot:nonspotcheck}. 

Like Protocol~\ref{prot:biased}, Protocol~\ref{prot:nonspotcheck} also allows the locality loophole to be closed if on each round $i$, the random choice $X_i$ is space-like separated from the output $B_i$ and the random choice $Y_i$ is space-like separated from the output $A_i$.

\begin{figure}
\includegraphics[width=\textwidth]{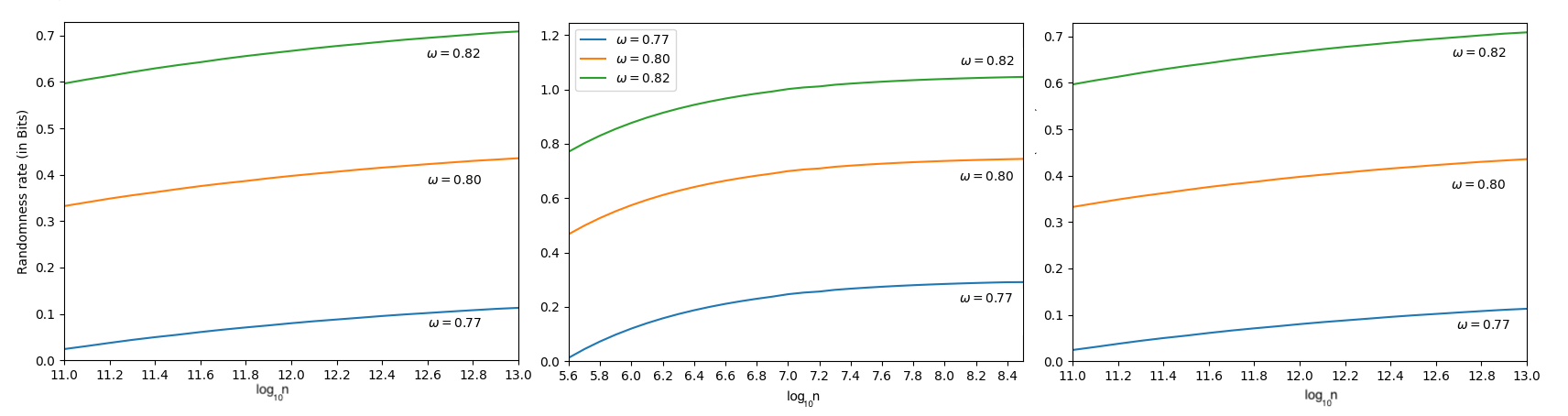}\\
\phantom{m}(a)\hspace{16em}(b)\hspace{16em}(c)
\caption{Graphs of the net rate of certifiable randomness according to the EAT for (a) the spot checking protocol (Protocol~\ref{prot:spotcheck}), (b) the protocol with recycled input randomness (Protocol~\ref{prot:nonspotcheck}), and (c) the protocol with biased local random number generators (Protocol~\ref{prot:biased}), showing the variation with the number of rounds for three different scores, $\score$.  The error parameters used were $\epsilon_S=3.09\times10^{-12}$ and $\epsilon_C=10^{-6}$. For each point on the curve (a) an optimization over $\gamma$ was performed to maximize the randomness; similarly, the values of $\zeta^A=\zeta^B$ were optimized over to generate the curves in (c).}
\label{fig:EAT_n}
\end{figure}

\begin{figure}
\includegraphics[width=\textwidth]{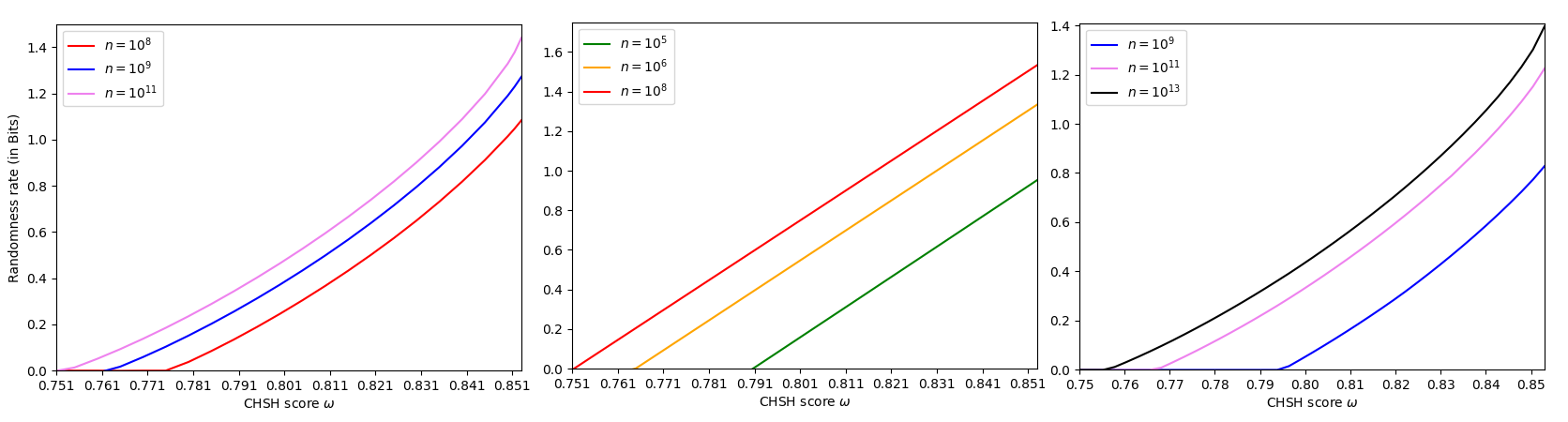} \\
\phantom{m}(a)\hspace{16em}(b)\hspace{16em}(c)
\caption{Graphs of the net rate of certifiable randomness according to the EAT for (a) the spot checking protocol (Protocol~\ref{prot:spotcheck}), (b) the protocol with recycled input randomness (Protocol~\ref{prot:nonspotcheck}), and (c) the protocol with biased local random number generators (Protocol~\ref{prot:biased}), showing the variation with the CHSH score $\score$. The round numbers, $n$, are indicated in the legend. The error parameters used were $\epsilon_S=3.09\times10^{-12}$ and $\epsilon_C=10^{-6}$. As in Figure~\ref{fig:EAT_n}, the values of $\gamma$ (for (a)) and $\zeta^A=\zeta^B$ (for (c)) were optimized over for each point.}
\label{fig:EAT_score}
\end{figure}

In each of the protocols, the parameter $\delta$ should be chosen depending on the desired completeness error. For the spot-checking protocol, the relation between the two is discussed in~\cite[Supplementary Information I~D]{LLR&}.  The analysis there can be applied to the protocol with recycled input randomness by setting $\gamma=1$ and the protocol with biased local random number generators is discussed in Appendix~\ref{app:comp}. The chosen soundness error affects the length of the extractor output, and, if chosen too small the output length becomes zero (see Appendix~\ref{app:EAT} for further discussion).

Figures~\ref{fig:EAT_n} and~\ref{fig:EAT_score} show how the amount of certifiable randomness varies with the score, $\score$, and round number, $n$. Note that in the cases where the rate curves of Figure~\ref{fig:rates} have linear sections, they are linear for most of their ranges.  Extending the linear part to the full range of quantum scores makes it easier to use the EAT while only resulting in a small drop in rate for scores close to the maximum quantum value. Note that, as mentioned above, strictly the numerical curves we gave for the von Neumann entropy are upper bounds; the curves in Figures~\ref{fig:EAT_n} and~\ref{fig:EAT_score} are generated under the assumption that these upper bounds are tight. [In the case of the spot-checking protocol (Protocol~\ref{prot:spotcheck}) we could use our lower bound instead. This would result in a small down-shifting of the curves, but mean that the bounds are provably reliable.]

To demonstrate the increased practicality of the two-sided curves, we use the parameters from a recent experiment~\cite{LLR&} with Protocol~\ref{prot:spotcheck}. There a score of just over $0.752$ was obtained, for which it would require about $9\times10^{10}$ rounds to achieve expansion using Protocol~\ref{prot:spotcheck} with $\gamma=3.383\times10^{-4}$, $\epsilon_S=3.09\times10^{-12}$, $\epsilon_C=10^{-6}$ and taking the one-sided randomness~\cite{LLR&}.  Using Protocol~\ref{prot:nonspotcheck} instead, and taking the two-sided randomness for the same score and error parameters (assuming Conjecture~\ref{conj:ABgXYE} holds) allows expansion for $n\gtrsim 8\times10^7$, significantly increasing the practicality. For instance, the main experiment of~\cite{LLR&} was based on a spot-checking protocol and took $19.2$ hours; the use of Protocol~\ref{prot:nonspotcheck} instead would allow the same amount of expansion in about $60$ seconds (this time holds under the assumption that the same repetition rate of the experiment can be met in the non-spot checking protocol\footnote{In some experiments, the rate at which we can switch between the two measurements is relatively slow, and hence when using Protocol~\ref{prot:nonspotcheck}, where switching is required on most rounds, the switching rate dominates, slightly increasing the time.}). Protocol~\ref{prot:biased}, however, produces lower randomness rates compared to the spot-checking protocol. This is partly because more input randomness is required, and also because the completeness error has a worse behaviour. Protocol~\ref{prot:biased} is hence useful when inputs are not recycled and when closing the locality loophole is desirable. 

When discussing randomness expansion we have considered the figure of merit to be the amount of expansion per entangled pair shared.  An alternative figure of merit is the ratio of the final randomness to the initial randomness, i.e., here we are considering how much randomness we can get from a given amount of initial randomness.  For the latter figure of merit, Protocol~\ref{prot:nonspotcheck} is no longer optimal, since the amount of expansion cannot exceed the amount of input randomness.  For the other two protocols the ratio of output randomness to input randomness can be made much higher by taking either $\gamma$ or $\zeta^A\zeta^B$ to be small.

\section{Discussion}
In this paper, we have given numerical bounds on various conditional von Neumann entropies that are relevant for CHSH-based device-independent protocols and discussed when each can be applied.  We have investigated their implications using explicit protocols, comparing the finite statistics rates using the EAT, showing use of two-sided randomness has the potential to make a big difference. We also looked at protocols beyond the usual spot checking type.  The first removes the spot checking to allow expansion while closing the locality loophole, and the second recycles the input randomness, so allowing expansion while performing a CHSH test on every round.

It remains an open question to find an analytic form for $F_{AB|X=0,Y=0,E}$, $F_{A|E}$ and $F_{AB|E}$.  Since the curves $F_{A|E}$ and $F_{AB|E}$ are linear for all but the very highest (experimentally least achievable) scores, in these cases not much is lost by extending the line to all scores forming a lower bound that tightly covers all of the experimentally relevant cases.  On the other hand $F_{AB|X=0,Y=0,E}$ is a convex curve throughout and hence a tight analytic form would be particularly useful in this case. Our initial analysis suggests that form of the parameters achieving the optima for these is sufficiently complicated that any analytic expression would not be compact. A reasonably tight analytic lower bound for $F_{AB|X=0,Y=0,E}$ could also be useful for theoretical analysis. Note also that the bound $F_{A|E}\geq F_{A|XYE}$ appears to be fairly tight (see Figure~\ref{fig:rates}(a)) so the analytic form for $F_{A|XYE}$ can be used to bound $F_{A|E}$ with little loss.  Another open problem is to find a concrete scenario in which $F_{AB|E}$ is directly useful.

The use of Jordan's lemma in this work prevents the techniques used being extended to general protocols, and finding improved ways to bound the conditional von Neumann entropy numerically in general cases remains of interest. For example, protocols that use three inputs for one party can allow up to 2 bits of randomness per entangled pair (see, e.g.~\cite{BRC}), and a way to tightly lower bound the von Neumann entropy in this case would further ease the experimental burden required to demonstrate DIQKD in the lab.\bigskip

\noindent{\bf Additional note:} While completing this work a new paper appeared~\cite{TSBSRSL}, which also found numerical curves for $H(A|XYE)$ and used them with some new protocols to generate improved rates for DIQKD.

\acknowledgements
We thank Erik Woodhead for useful comments and Peter Brown for sharing his data used in Figure~\ref{fig:rates_conj}. This work was supported by the Quantum Communications Hub of the UK Engineering and Physical Sciences Research Council (EPSRC) (grant nos.\ EP/M013472/1 and EP/T001011/1).

\appendix
\section{Additional graphs}\label{app:graphs}
Figure~\ref{fig:gXYE}(a) gives the graphs of $F_{AB|XYE}(\score)$ and $G_{AB|XYE}(\score)$, while Figure~\ref{fig:gXYE}(b) shows those for $F_{A|XYE}(\score)$ and $G_{A|XYE}(\score)$. In each case the $G$ graphs have a concave and convex part and the $F$ graphs are formed by taking the convex lower bound. For these cases the points at which the tangents are taken are $\score^*_{AB|XYE}\approx0.8440$ and $\score^*_{A|XYE}\approx0.8470$.

Figure~\ref{fig:gE}(a) gives the graphs of $F_{AB|E}(\score)$ and $G_{AB|E}(\score)$, while Figure~\ref{fig:gE}(b) shows those for $F_{A|XYE}(\score)$ and $G_{A|XYE}(\score)$. Again, in each case the $G$ graphs have a concave and convex part and the $F$ graphs are formed by taking the convex lower bound.  For these cases the points at which the tangents are taken are $\score^*_{A|E}\approx0.8505$ and $\score^*_{AB|E}\approx0.8523$.

\begin{figure}
\includegraphics[scale=0.50]{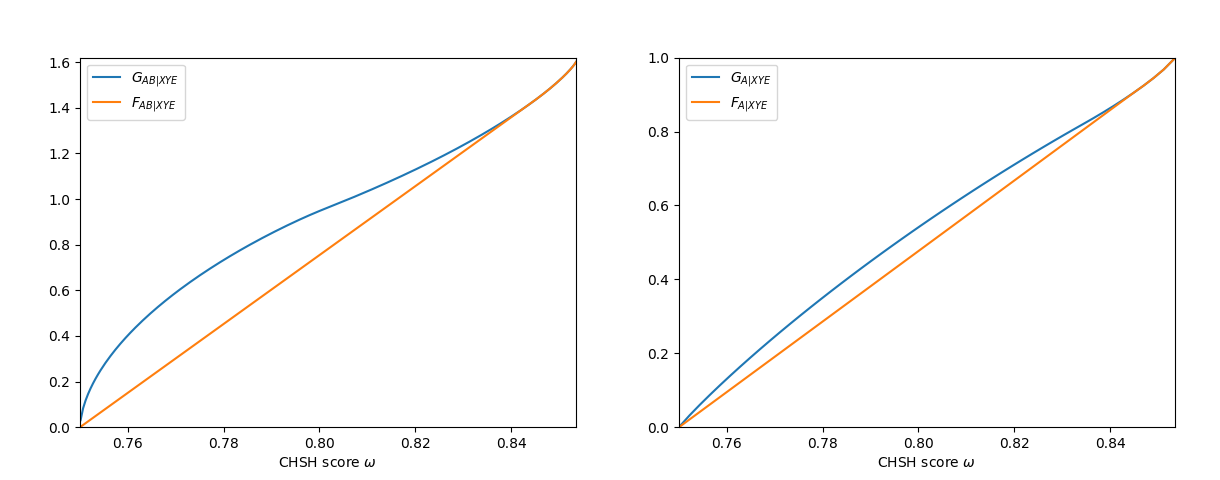} \\
(a)\hspace{24em}(b)
\caption{ (a) Two-sided and (b) one-sided entropy curves conditioned on $X$, $Y$ and $E$ with uniform input distribution.}
\label{fig:gXYE}
\end{figure}

\begin{figure}
\includegraphics[scale=0.50]{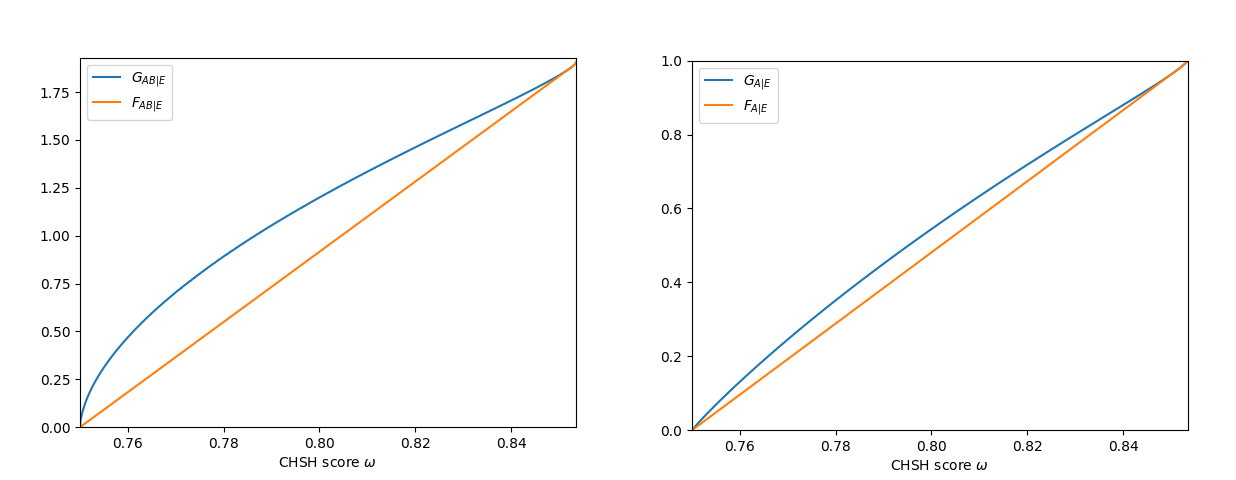} \\
(a)\hspace{24em}(b)
\caption{(a) Two-sided and (b) one-sided entropy curves conditioned on $E$ with uniform input distribution.}
\label{fig:gE}
\end{figure}

\section{Simplifying the strategy}\label{app:Proof1}
Given a Hilbert space $\cH$, we use $\cP(\cH)$ to be the set of positive semi-definite operators on $\cH$, and $\cS(\cH)$ to be the set of density operators, i.e., elements of $\cP$ with trace 1. The pure states on $\cH$ (elements of $\cS(\cH)$ with rank 1) will be denoted $\cS_P(\cH)$. A POVM on $\cH$ is a set of positive operators $\{E_i\}_i$ with $E_i\in\cP(\cH)$  for all $i$ and $\sum_i E_i=\id_\cH$, where $\id_\cH$ is the identity operator on $\cH$. A projective measurement on $\cH$ is a POVM on $\cH$ where $E_i^2=E_i$ for all $i$. We define the Bell states
\begin{eqnarray}
\ket{\Phi_0}&=&\frac{1}{\sqrt{2}}\left(\ket{00}+\ket{11}\right)\label{bell1}\\
\ket{\Phi_1}&=&\frac{1}{\sqrt{2}}\left(\ket{00}-\ket{11}\right)\\
\ket{\Phi_2}&=&\frac{1}{\sqrt{2}}\left(\ket{01}+\ket{10}\right)\\
\ket{\Phi_3}&=&\frac{1}{\sqrt{2}}\left(\ket{01}-\ket{10}\right)\,,\label{bell4}
\end{eqnarray}
and use $\sigma_1=\ketbra{1}{0}+\ketbra{0}{1}$, $\sigma_2=\ii\ketbra{1}{0}-\ii\ketbra{0}{1}$ and $\sigma_3=\proj{0}-\proj{1}$ as the three Pauli operators.

In this section we make a series of simplifications of the form of the optimization. The argument given broadly follows the logic of~\cite{PABGMS} (see also~\cite{WAP} for an alternative).

\begin{definition}
  A \emph{single-round CHSH measurement strategy} is a tuple $(\cH_{A'},\cH_{B'},\{M_{a|x}\}_{x,a}, \{N_{b|y}\}_{y,b})$, where $\cH_{A'}$ and $\cH_{B'}$ are Hilbert spaces, and $\{M_{a|x}\}_a$ is a POVM on $\cH_{A'}$ for each $x\in\{0,1\}$ and likewise $\{N_{b|y}\}_b$ is a POVM on $\cH_{B'}$ for each $y\in\{0,1\}$. In the case that all the POVMs are projective we will call this a \emph{single-round CHSH projective measurement strategy}.
\end{definition}

\begin{definition}
  A \emph{single-round CHSH strategy} is a single-round CHSH measurement strategy together with a state $\rho_{A'B'E}\in\cH_{A'}\ot\cH_{B'}\ot\cH_E$, where $\cH_E$ is an arbitrary Hilbert space.
\end{definition}
Note that in a device-independent scenario, such a strategy can be chosen by the adversary.

\begin{definition}
  Given a single-round CHSH measurement strategy and a distribution $p_{XY}$ over the settings $X$ and $Y$, the \emph{associated CHSH channel} is defined by
$$\cN:\cS(\cH_{A'}\ot\cH_{B'})\to\cS(\cH_A\ot\cH_B\ot\cH_X\ot\cH_Y): \sigma\mapsto\sum_{abxy}p_{XY}(x,y)\proj{a}\ot\proj{b}\ot\proj{x}\ot\proj{y}\tr\left((M_{a|x}\ot N_{b|y})\sigma\right)\,,$$
where $\cH_A$, $\cH_B$, $\cH_X$ and $\cH_Y$ are two dimensional Hilbert spaces.
The union of the sets of associated CHSH channels for all single-round CHSH measurement strategies for some fixed input distribution $p_{XY}$ is denoted $\cC(p_{XY})$. The union of the sets of associated CHSH channels for all single-round CHSH projective measurement strategies is denoted $\cC_{\Pi}(p_{XY})$.
\end{definition}
The output of the associated CHSH channel is classical, and $\cH_A\ot\cH_B\ot\cH_X\ot\cH_Y$ stores the outcomes and the chosen measurements.  We will usually apply this channel to the $AB$ part of a tripartite system, giving
\begin{equation}\label{eq:post-meas}
(\cN\ot\cI_E)(\rho_{A'B'E})=\sum_{abxy}p_{XY}(x,y)p_{AB|xy}(a,b)\proj{a}\ot\proj{b}\ot\proj{x}\ot\proj{y}\ot\tau^{a,b,x,y}_E\,,
\end{equation}
where $\tau^{a,b,x,y}_E\in\cS(\cH_E)$ for each $a,b,x,y$ (it is the normalization of $\tr_{A'B'}\left((M_{a|x}\ot N_{b|y}\ot \id_E)\rho_{A'B'E}\right)$).

Note that the CHSH score, which we denote $\scorefunction((\cN\ot\cI_E)(\rho_{A'B'E}))$ does not depend on the distribution $p_{XY}$ of input settings.

We will be interested in optimization problems of the form
\begin{align}
  F(\score,p_{XY})=&\inf_{\cR}  \bar{H}((\cN\ot\cI_E)(\rho_{A'B'E}))\text{  where  }\label{opt_1}\\\nonumber
  &\cR=\{(\cN,\rho_{A'B'E}):\cN\in\cC(p_{XY}),\,\scorefunction((\cN\ot\cI_E)(\rho_{A'B'E}))=\score\}
\end{align}
where $\cH_E$ is an arbitrary Hilbert space and the spaces $\cH_{A'}$ and $\cH_{B'}$ are those from the chosen element of $\cC(p_{XY})$, i.e., the set $\cR(\score)$ runs over all possible dimensions of these spaces, and $\omega$ is some fixed real number.  Here $\bar{H}$ can be any one of the  following entropic quantities defined on the state $(\cN\ot\cI_E)(\rho_{A'B'E})$: $H(AB|X=0,Y=0,E)$, $H(AB|XYE)$, $H(AB|E)$, $H(A|X=0,Y=0,E)$, $H(A|XYE)$ or $H(A|E)$. We consider the family of optimizations in this work and many of the arguments that follow are independent of this choice.

\subsection{Reduction to projective measurements}
In this section, we conclude that there is no loss in generality in assuming that the devices perform projective measurements. More precisely, we prove the following lemma.
\begin{lemma}\label{lem:proj}
  The sets
  \begin{eqnarray*}
    \cT_1&:=&\{(\cN\ot\cI_E)(\rho_{A'B'E}):\cN\in\cC(p_{XY}), \rho_{A'B'E}\in\cS(\cH_{A'}\ot\cH_{B'}\ot\cH_E)\}\text{  and}\\
    \cT_2&:=&\{(\cN\ot\cI_E)(\rho_{A'B'E}):\cN\in\cC_\Pi(p_{XY}), \rho_{A'B'E}\in\cS(\cH_{A'}\ot\cH_{B'}\ot\cH_E)\}
    \end{eqnarray*}
are identical.
\end{lemma}

This is a corollary of Naimark's theorem, which we state in the following way.
\begin{theorem}[Naimark's theorem]\label{thm:naimark}
  Let $\{E_i\}_i$ be a POVM on $\cH$. There exists a Hilbert space $\cH'$ and a projective measurement $\{\Pi_i\}_i$ on $\cH\ot\cH'$ such that for any $\rho\in\cS(\cH)$
  $$\sum_i\proj{i}\tr(\rho E_i)=\sum_i\proj{i}\tr(\Pi_i(\rho\ot\proj{0}))\,.$$
\end{theorem}
\begin{proof}
  We can directly construct this measurement as follows.  Consider the isometry $V:\cH\to\cH\ot\cH'$ given by $V=\sum_i\sqrt{E_i}\ot\ket{i}$, and let $U$ be the extension of $V$ to a unitary with the property that $U(\ket{\psi}\ot\ket{0})=\sum_i\sqrt{E_i}\ket{\psi}\ot\ket{i}$ for any $\ket{\psi}\in\cH$.  This construction ensures that the channels
  \begin{eqnarray*}
    \cE:\rho&\mapsto&\sum_i\proj{i}\tr(E_i\rho)\text{  and}\\
\cE':\rho&\mapsto&\sum_i\proj{i}\tr\left((\id\ot\proj{i})U(\rho\ot\proj{0})U^\dagger\right)
  \end{eqnarray*}
  are identical.  The second of these can be rewritten
  $$\rho\mapsto\sum_i\proj{i}\tr\left(\Pi_i(\rho\ot\proj{0})\right)\,,$$
where we take $\Pi_i=U^\dagger(\id\ot\proj{i})U$, as required.
\end{proof}

\begin{proof}[Proof of Lemma~\ref{lem:proj}]
  By definition $\cT_2\subseteq\cT_1$. For the other direction, consider a state $\rho_{A'B'E}$ and POVMs $\{M_{a|x}\}_{a,x}$ and $\{N_{b|y}\}_{b,y}$ forming a single-round CHSH strategy in $\cT_1$.   We use the construction in the proof of Theorem~\ref{thm:naimark} to generate the projectors $\Pi^A_{a|x}$ and $\Pi^B_{b|y}$ as Naimark extensions of the POVMs. Instead of creating the state $\rho_{A'B'E}$, the state $\rho_{A'B'E}\ot\proj{0}_{A''}\ot\proj{0}_{B''}$ is created instead, where the projectors $\Pi^A_{a|x}$ act on $A'A''$ and $\Pi^B_{b|y}$ act on $B'B''$. Since the latter is a strategy in $\cT_2$ leading to the same post-measurement state~\eqref{eq:post-meas}, we have $\cT_1\subseteq\cT_2$, which completes the proof.
\end{proof}

\subsection{Reduction to convex combinations of qubit strategies}\label{app: convex combinations of qubit strategies}
This is a consequence of Jordan's lemma~\cite{Jordan} and is a special feature that applies only because the Bell inequality has two inputs and two outputs for each party.

\begin{lemma}[Jordan's lemma]\label{lem:Jordan}
Let $A_1$  and $A_2$ be two Hermitian operators on $\cH$ with eigenvalues $\pm1$, then we can decompose $\cH=\bigoplus_{\alpha}\cH_{\alpha}$ such that $A_1$ and $A_2$ preserve the subspaces $\cH_\alpha$, and where each $\cH_\alpha$ has dimension at most 2.
\end{lemma}
\begin{corollary}\label{cor:proj}
Let $\Pi_1$ and $\Pi_2$ be two projections on $\cH$. We can decompose $\cH=\bigoplus_{\alpha}\cH_{\alpha}$ such that $\Pi_1$, $\id-\Pi_1$, $\Pi_2$ and $\id-\Pi_2$ preserve the subspaces $\cH_\alpha$, and where each $\cH_\alpha$ has dimension at most 2.
\end{corollary}
\begin{proof}
Apply Jordan's lemma to the Hermitian operators $A_1=2\Pi_1-\id$ and $A_2=2\Pi_2-\id$ with eigenvalues $\pm1$, and consider $\ket{\psi}\in\cH_\alpha$ for some $\alpha$.  By construction $A_1\ket{\psi}\in\cH_\alpha$ from which it follows that $\Pi_1\ket{\psi}\in\cH_\alpha$, and hence also $(\id-\Pi_1)\ket{\psi}\in\cH_\alpha$. Thus, $\Pi_1$ and $\id-\Pi_1$ preserve the subspace; likewise $\Pi_2$ and $\id-\Pi_2$.
\end{proof}

This implies the following
\begin{lemma}
  Let $\cC_{2\times2}(p_{XY})$ be the set of CHSH channels associated with the single-round CHSH projective measurement strategies where each of the four projectors $M_{a|x}$ is block diagonal with $2\times2$ blocks, and each of the four projectors $N_{b|y}$ is block diagonal with $2\times2$ blocks.
  The sets
  \begin{eqnarray*}
    \cT_2&:=&\{(\cN\ot\cI_E)(\rho_{A'B'E}):\cN\in\cC_\Pi(p_{XY}), \rho_{A'B'E}\in\cS(\cH_{A'}\ot\cH_{B'}\ot\cH_E)\}\text{  and}\\
    \cT_3&:=&\{(\cN\ot\cI_E)(\rho_{A'B'E}):\cN\in\cC_{2\times2}(p_{XY}), \rho_{A'B'E}\in\cS(\cH_{A'}\ot\cH_{B'}\ot\cH_E)\}
    \end{eqnarray*}
are identical.
\end{lemma}
\begin{proof}
This follows by applying Corollary~\ref{cor:proj} to the projectors $M_{0|0}$ and $M_{0|1}$ to get the blocks on $\cH_{A'}$ and to the projectors $N_{0|0}$ and $N_{0|1}$ to get the blocks on $\cH_{B'}$. Although some of the blocks may be $1\times1$, we can collect these together and treat them as a $2\times2$ block, or add an extra dimension to the space (on which the state has no support) to achieve all $2\times2$ blocks.
\end{proof}
We can also make the state only have support on the $2\times2$ blocks.
\begin{lemma}
  The sets
  \begin{eqnarray*}
    \cT_3&:=&\{(\cN\ot\cI_E)(\rho_{A'B'E}):\cN\in\cC_{2\times2}(p_{XY}), \rho_{A'B'E}\in\cS(\cH_{A'}\ot\cH_{B'}\ot\cH_E)\}\text{  and}\\
    \cT_4&:=&\{(\cN\ot\cI_E)(\rho_{A'B'E}):\cN\in\cC_{2\times2}(p_{XY}), \rho_{A'B'E}\in\cS_{2\times2}(\cH_{A'}\ot\cH_{B'}\ot\cH_E)\}
    \end{eqnarray*}
are identical. Here $\cS_{2\times2}(\cH_{A'}\ot\cH_{B'}\ot\cH_E)$ is the subset of $\cS(\cH_{A'}\ot\cH_{B'}\ot\cH_E)$ such that $\rho_{A'B'E}\in\cS_{2\times2}(\cH_{A'}\ot\cH_{B'}\ot\cH_E)$ implies
  $$\rho_{A'B'E}=\sum_{\alpha,\beta}(\Pi^\alpha_{A'}\ot \Pi^\beta_{B'}\ot\id_E)\rho_{A'B'E}(\Pi^\alpha_{A'}\ot \Pi^\beta_{B'}\ot\id_E)\,,$$
where $\{\Pi^{\alpha}\}_\alpha$ are projectors onto the $2\times2$ diagonal blocks.
\end{lemma}
\begin{proof}
Consider a state $\rho_{A'B'E}$ and sets of projectors $\{M_{a|x}\}_{a,x}$ and $\{N_{b|y}\}_{b,y}$ from the set $\cT_3$.
  For brevity, write $\Pi^{\alpha,\beta}=\Pi^\alpha_{A'}\ot \Pi^\beta_{B'}$. Then, since
  $$M_{a|x}\ot N_{b|y}=\sum_{\alpha,\beta}(\Pi^\alpha_{A'}\ot \Pi^\beta_{B'})(M_{a|x}\ot N_{b|y})(\Pi^\alpha_{A'}\ot \Pi^\beta_{B'})\,,$$
  we have
  \begin{eqnarray*}
    \tr_{A'B'}((M_{a|x}\ot N_{b|y}\ot\id)\rho_{A'B'E})&=&\tr_{A'B'}\left(\sum_{\alpha,\beta}(\Pi^{\alpha,\beta}\ot\id_E)(M_{a|x}\ot N_{b|y}\ot\id)(\Pi^{\alpha,\beta}\ot\id_E)\rho_{A'B'E}\right)\\
                                                     &=&\tr_{A'B'}\left(\sum_{\alpha,\beta}(M_{a|x}\ot N_{b|y}\ot\id)(\Pi^{\alpha,\beta}\ot\id_E)\rho_{A'B'E}(\Pi^{\alpha,\beta}\ot\id_E)\right)\\
    &=&\tr_{A'B'}((M_{a|x}\ot N_{b|y}\ot\id)\rho'_{A'B'E})\,,
  \end{eqnarray*}
where $\rho'_{A'B'E}=\sum_{\alpha,\beta}(\Pi^{\alpha,\beta}\ot\id_E)\rho_{A'B'E}(\Pi^{\alpha,\beta}\ot\id_E)$. Thus, if we replace $\rho_{A'B'E}$ by $\rho'_{A'B'E}$ we obtain the same post-measurement state~\eqref{eq:post-meas}. Hence $\cT_3\subseteq\cT_4$, and, since the other inclusion is trivial, $\cT_3=\cT_4$.
\end{proof}

\begin{lemma}
Let $\cN\in\cC_{2\times2}(p_{XY})$ and $\rho_{A'B'E}\in\cS_{2\times2}(\cH_{A'}\ot\cH_{B'}\ot\cH_E)$. The state $(\cN\ot\cI_E)(\rho_{A'B'E})$ can be formed as a convex combination of states $(\cN_\lambda\ot\cI_E)(\rho^\lambda_{A''B''E})$, where for each $\lambda$, the channel $\cN_\lambda$ is that associated with a single-round measurement strategy with two 2-dimensional Hilbert spaces and distribution $p_{XY}$.
\end{lemma}
\begin{proof}
Since $\rho_{A'B'E}\in\cS_{2\times2}(\cH_{A'}\ot\cH_{B'}\ot\cH_E)$ the $2\times2$ block structure means we can write
$\rho_{A'B'E}=\sum_{\alpha,\beta}p_{\alpha,\beta}\rho^{\alpha,\beta}_{A'B'E}$, where $p_{\alpha,\beta}\rho^{\alpha,\beta}_{A'B'E}=(\Pi^{\alpha,\beta}\ot\id_E)\rho_{A'B'E}(\Pi^{\alpha,\beta}\ot\id_E)$ and $\tr(\rho^{\alpha,\beta}_{A'B'E})=1$ for all $\alpha$ and $\beta$. Likewise, taking $M^\alpha_{a|x}=\Pi_{A'}^\alpha M_{a|x}\Pi_{A'}^\alpha$ and $N^\beta_{b|y}=\Pi^\beta_{B'} N_{b|y}\Pi^\beta_{B'}$ we can write $M_{a|x}=\sum_\alpha M^\alpha_{a|x}$ and $N_{b|y}=\sum_\beta N^\beta_{b|y}$. In terms of these we have
\begin{eqnarray*}
  \tr_{A'B'}((M_{a|x}\ot N_{b|y}\ot\id)\rho_{A'B'E})&=&\sum_{\alpha,\beta}p_{\alpha,\beta}\tr_{A'B'}((M^\alpha_{a|x}\ot N^\beta_{b|y}\ot\id)\rho^{\alpha,\beta}_{A'B'E})
\end{eqnarray*}
We can then associate a value of $\lambda$ with each pair $(\alpha,\beta)$, replace each $\rho^{\alpha,\beta}_{A'B'E}$ by a state on $A''B''E$ in which $A''$ and $B''$ are two-dimensional (the support of $\rho^{\alpha,\beta}_{A'B'}$ has dimension at most 4), and likewise replace the projectors by qubit projectors. In terms of these we have
\begin{equation*}
  (\cN\ot\cI_E)(\rho_{A'B'E})=\sum_\lambda p_\lambda(\cN_\lambda\ot\cI_E)(\rho^\lambda_{A''B''E})\,.\qedhere
  \end{equation*}
\end{proof}
In other words, any post-measurement state~\eqref{eq:post-meas} that can be generated in the general case, can also be generated if Eve sends a convex combination of two qubit states, and where the measurements used by the separated devices depend on the state sent. Eve could realise such a strategy in practice by using pre-shared randomness.  We can proceed to consider strategies in which qubits are shared between the two devices, and then consider the mixture of such strategies after doing so.

\subsection{Qubit strategies}
In this section we consider the single-round CHSH measurement strategies in which $\cH_{A'}$ and $\cH_{B'}$ are two-dimensional and the measurements are rank-1 projectors. Given a distribution $p_{XY}$ we use $\cC_{\Pi_1,2}(p_{XY})$ to denote the set of associated CHSH channels.  We restrict to rank-1 projectors because if one of the projectors is identity it is not possible to achieve a non-classical CHSH score, and the non-classical scores are the ones of interest.
\begin{lemma}
Consider a single-round CHSH measurement strategy for which one of the POVM elements is identity and let $\cN$ be the associated CHSH channel.  For any state $\rho_{A'B'}$ on which $\cN$ can act we have $\scorefunction(\cN(\rho_{A'B'}))\leq3/4$.
\end{lemma}
\begin{proof}
  Suppose the identity element corresponds to $M_{0|0}$ (the other cases follow symmetrically). The conditional distribution $p_{AB|XY}$ then takes the form\\
 \begin{tabular}{ |cc| c  c| c c| }
\hline
        && $X=0$ &  & $X=1$ &\\
        && $A=0$ & $A=1$  & $A=0$ & $A=1$\\
\hline
 $Y=0$&$B=0$   & $\mu$ &  $0$ & $\nu$  & $\mu-\nu$\\  
  &$B=1$ & $1-\mu$ & $0$ & $\zeta$ &  $1-\mu-\zeta$     \\ 
  \hline
  $Y=1$&$B=0$  & $\gamma$ &  $0$ & $\xi$ & $\gamma-\xi$ \\
  &$B=1$ &  $1-\gamma$  &  $0$ & $\nu+\zeta-\xi$ & $1+\xi-\gamma-\nu-\zeta$\\ 
  \hline
  \end{tabular}
where we have used the no-signalling conditions.
The associated score is
\begin{eqnarray*} \frac{1}{4}\left(\mu+\nu+(1-\mu-\zeta)+\gamma+(\gamma-\xi)+(\nu+\zeta-\xi)\right)&=&\frac{1}{4}\left(1+2\nu+2\gamma-2\xi\right)
 \end{eqnarray*}
 Since every element of the distribution must be between $0$ and $1$ we have  $1+\xi-\gamma-\nu-\zeta\geq0$, and hence $1+2\nu+2\gamma-2\xi\leq 3-2\zeta\leq3$, from which the claim follows.
\end{proof}

We will then consider an optimization of the form~\eqref{opt_1}, but restricting to $\cC_{\Pi_1,2}$, i.e., 
\begin{align}
  h(\score)&=\inf_{\cR(\score)}  \bar{H}((\cN\ot\cI_E)(\rho_{A'B'E}))\text{  where  }\label{opt_2}\\\nonumber
  \cR(\score)&=\{(\cN,\rho_{A'B'E}):\cN\in\cC_{\Pi_1,2}(p_{XY}),\,\scorefunction((\cN\ot\cI_E)(\rho_{A'B'E}))=\score\}
\end{align}

The next step is to show that without loss of generality we can reduce to states that are invariant under application of $\sigma_2\ot\sigma_2$ on $A'B'$.
\begin{lemma}\label{lem:marginals}
Let $p_{XY}$ be a distribution, $\cN\in\cC_{\Pi_1,2}(p_{XY})$ and $\rho_{A'B'E}\in\cS(\cH_{A'}\ot\cH_{B'}\ot\cH_E)$ be such that $\scorefunction((\cN\ot\cI_E)(\rho_{A'B'E}))=\score$. There exists a state $\tilde{\rho}_{A'B'EE'}\in\cS(\cH_{A'}\ot\cH_{B'}\ot\cH_E\ot\cH_{E'})$ such that $\scorefunction((\cN\ot\cI_{EE'})(\tilde{\rho}_{A'B'EE'}))=\score$, $\tilde{\rho}_{A'B'EE'}=(\sigma_2\ot\sigma_2\ot\id_{EE'})\tilde{\rho}_{A'B'EE'}(\sigma_2\ot\sigma_2\ot\id_{EE'})$ and $\bar{H}((\cN\ot\cI_{EE'})(\tilde{\rho}_{A'B'EE'}))=\bar{H}((\cN\ot\cI_E)(\rho_{A'B'E}))$ for all six of the entropic functions given earlier.
\end{lemma}
Note that this implies that $p_{A|X}$ and $p_{B|Y}$ can be taken to be uniform.

This is a consequence of the following lemmas.
\begin{lemma}\label{lem:real}
  Let $\{\Pi_{0|0},\Pi_{1|0}\}$ and $ \{\Pi_{0|1} ,\Pi_{1|1}\}$ be two rank-one projective measurements on a two dimensional Hilbert space $\cH$. There exists a basis $\{\ket{e_i}\}_{i=1}^{2}$ such that $\bra{e_l}\Pi_{i|j}\ket{e_k} \in \mathbb{R}$ for all $i,j,k,l$.
  \end{lemma}
\begin{proof}
Without loss of generality, we can take $\Pi_{0|0}=\proj{0}$ and $\Pi_{1|0}=\proj{1}$, and then write $\Pi_{0|1}=\proj{\alpha_{0|1}}$ and $\Pi_{1|1}=\proj{\alpha_{1|1}}$, where $\ket{\alpha_{0|1}}=\cos(\lambda) \ket{0} + \e^{\ii\chi} \sin(\lambda) \ket{1}$ and $\ket{\alpha_{1|1}}=\sin(\lambda) \ket{0} - \e^{\ii\chi} \cos(\lambda) \ket{1}$. Then, we can re-define $\ket{1} \rightarrow \e^{\ii\chi} \ket{1}$ so that $\ket{\alpha_{0|1}} = \cos(\lambda) \ket{0} + \sin(\lambda) \ket{1}$ and $\ket{\alpha_{1|1}} = \sin(\lambda) \ket{0} - \cos(\lambda) \ket{1}$, with $\lambda \in \mathbb{R}$.
\end{proof}

\begin{lemma}\label{lem:Uproj}
Let $\{\Pi_{0|0},\Pi_{1|0}\}$ and $ \{\Pi_{0|1} ,\Pi_{1|1}\}$ be two rank-one projective measurements on a two dimensional Hilbert space $\cH$, then, there exists a unitary transformation $U$ such that $U\Pi_{j|i}U^\dagger=\Pi_{j\oplus1|i}$ for all $i,j$. 
\end{lemma}
\begin{proof}
Let $\Pi_{j|i}=\proj{\alpha_{0|1}}$ for all $i,j$. From Lemma~\ref{lem:real}, we can change basis such that $\ket{\alpha_{0|0}}=\ket{0}$, $\ket{\alpha_{1|0}}=\ket{1}$, $\ket{\alpha_{0|1}}=\cos(\lambda)\ket{0}+\sin(\lambda)\ket{1}$ and $\ket{\alpha_{1|1}}=\sin(\lambda)\ket{0}-\cos(\lambda)\ket{1}$ for some $\lambda\in\mathbb{R}$. Any unitary of the form $U=\e^{\ii\phi}(\ketbra{0}{1}-\ketbra{1}{0})$, with $\phi\in\mathbb{R}$ then satisfies the desired relations. 
\end{proof}

The following lemma is well-known (it follows straightforwardly from e.g.,~\cite[Section~11.3.5]{Nielsen&Chuang})
\begin{lemma}\label{lem:N&C}
For $\rho_{CZEE'}=\sum_ip_i\rho^i_{CZE}\ot\proj{i}_{E'}$ we have $H(C|ZEE')_\rho=\sum_ip_iH(C|ZE)_{\rho^i}$.
\end{lemma}

We now prove Lemma~\ref{lem:marginals}.
\begin{proof}[Proof of Lemma~\ref{lem:marginals}]
  Let $U_A$ and $U_B$ be the unitaries formed by applying Lemma~\ref{lem:Uproj} to respective measurements of each device and using the choice of basis specified in the proof of Lemma~\ref{lem:Uproj} we can take $U_A=\sigma_2$ and $U_B=\sigma_2$.  Then define $\rho'_{A'B'E}=(\sigma_2\ot\sigma_2\ot\id_E)\rho_{A'B'E}(\sigma_2\ot\sigma_2\ot\id_E)$.  The states $(\cN\ot\cI_E)(\rho'_{A'B'E})$ and $(\cN\ot\cI_E)(\rho_{A'B'E})$ are related by
  $$(\cN\ot\cI_E)(\rho'_{A'B'E})=(\id_{XYE}\ot\sigma_1\ot\sigma_1)(\cN\ot\cI_E)(\rho_{A'B'E})(\id_{XYE}\ot\sigma_1\ot\sigma_1)\,.$$
In other words $(\cN\ot\cI_E)(\rho'_{A'B'E})$ is identical to $(\cN\ot\cI_E)(\rho_{A'B'E})$, except that the outcomes of each device have been relabelled.  It follows that $\scorefunction((\cN\ot\cI_E)(\rho'_{A'B'E}))=\score$ and $\bar{H}((\cN\ot\cI_E)(\rho'_{A'B'E}))=\bar{H}((\cN\ot\cI_E)(\rho_{A'B'E}))$.

Now consider the state $\tilde{\rho}_{A'B'EE'}=(\rho_{A'B'E}\ot\proj{0}_{E'}+\rho'_{A'B'E'}\ot\proj{1}_{E'})/2$.  We have $(\cN\ot\cI_{EE'})(\tilde{\rho}_{A'B'EE'})=((\cN\ot\cI_E)(\rho_{A'B'E})\ot\proj{0}_{E'}+(\cN\ot\cI_E)(\rho'_{A'B'E'})\ot\proj{1}_{E'})/2$. Since the CHSH score is linear, we have $\scorefunction((\cN\ot\cI_{EE'})(\tilde{\rho}_{A'B'EE'}))=\score$. By construction, $\tilde{\rho}_{A'B'E}=(\sigma_2\ot\sigma_2\ot\id_E)\tilde{\rho}_{A'B'E}(\sigma_2\ot\sigma_2\ot\id_E)$. Finally, as a consequence of Lemma~\ref{lem:N&C}, for any of the entropy functions $H$ we have $\bar{H}(\tilde{\rho}_{A'B'EE'})=\bar{H}(\rho_{A'B'E})$. 
\end{proof}

\begin{corollary}\label{cor:sigma2}
Any optimization of the form~\eqref{opt_2} is equivalent to an optimization of the same form but where each of the projectors are onto states of the form $\alpha\ket{0}+\beta\ket{1}$ with $\alpha,\beta\in\mathbb{R}$ and $\rho_{A'B'E}=(\sigma_2\ot\sigma_2\ot\id)\rho_{A'B'E}(\sigma_2\ot\sigma_2\ot\id)$.
\end{corollary}

Next we consider the form of the reduced state $\rho_{A'B'}$ in the Bell basis.
\begin{lemma}
  Let $\cN$ be the channel associated with a single-round CHSH strategy in which each POVM element is a projector of the form $\cos(\alpha)\ket{0}+\sin(\alpha)\ket{1}$ with $\alpha\in\mathbb{R}$. The state $\rho^P_{A'B'E}$ satisfies $(\cN\ot\cI_E)(\rho^P_{A'B'E})=(\cN\ot\cI_E)(\rho_{A'B'E})$, where $\rho^P_{A'B'E}$ is formed from $\rho_{A'B'E}$ by taking the partial transpose on $A'B'$ in the Bell basis.
\end{lemma}
\begin{proof}
By definition, the partial transpose generates the state $$\rho^P_{A'B'E}=\sum_{ij}(\ketbra{\Psi_i}{\Psi_j}\ot\id_E)\rho_{A'B'E}(\ketbra{\Psi_i}{\Psi_j}\ot\id_E)\,.$$
Writing the partial trace out in the Bell basis, for any two projectors $\Pi_1$ and $\Pi_2$ on $\cH_{A'}$ and $\cH_{B'}$ we have
\begin{eqnarray}
  \tr_{A'B'}((\Pi_1\ot\Pi_2\ot\id_E)\rho^P_{A'B'E})&=&\sum_i(\bra{\Psi_i}(\Pi_1\ot\Pi_2)\ot\id_E)\rho^P_{A'B'E}(\ket{\Psi_i}\ot\id_E)\nonumber\\
                                                   &=&\sum_{ijk}((\bra{\Psi_i}(\Pi_1\ot\Pi_2)\ketbra{\Psi_j}{\Psi_k})\ot\id_E)\rho_{A'B'E}(\ketbra{\Psi_j}{\Psi_k}\ket{\Psi_i}\ot\id_E)\nonumber\\
                                                   &=&\sum_{ij}\bra{\Psi_i}(\Pi_1\ot\Pi_2)\ket{\Psi_j}(\bra{\Psi_i}\ot\id_E)\rho_{A'B'E}(\ket{\Psi_j}\ot\id_E)\,.\label{eq:sw}
\end{eqnarray}
When $\Pi_1$ and $\Pi_2$ are each projectors onto states of the form $\cos(\alpha)\ket{0}+\sin(\alpha)\ket{1}$ a short calculation reveals $\bra{\Psi_i}(\Pi_1\ot\Pi_2)\ket{\Psi_j}=\bra{\Psi_j}(\Pi_1\ot\Pi_2)\ket{\Psi_i}$. Using this in~\eqref{eq:sw} we can conclude that
$$\tr_{A'B'}((\Pi_1\ot\Pi_2\ot\id_E)\rho^P_{A'B'E})=\tr_{A'B'}((\Pi_1\ot\Pi_2\ot\id_E)\rho_{A'B'E})\,,$$
from which it follows that $(\cN\ot\cI_E)(\rho^P_{A'B'E})=(\cN\ot\cI_E)(\rho_{A'B'E})$.
\end{proof}

\begin{corollary}\label{cor:trans}
Any optimization of the form~\eqref{opt_2} is equivalent to an optimization of the same form but where each of the projectors are onto states of the form $\alpha\ket{0}+\beta\ket{1}$ with $\alpha,\beta\in\mathbb{R}$, $\rho_{A'B'E}=(\sigma_2\ot\sigma_2\ot\id)\rho_{A'B'E}(\sigma_2\ot\sigma_2\ot\id)$ and $\rho_{A'B'E}=\rho^P_{A'B'E}$.
\end{corollary}
\begin{proof}
We established the invariance under $(\sigma_2\ot\sigma_2\ot\id)$ in Corollary~\ref{cor:sigma2}. Since $(\cN\ot\cI_E)(\rho^P_{A'B'E})=(\cN\ot\cI_E)(\rho_{A'B'E})$, if Eve uses the state $(\rho_{A'B'E}\ot\proj{0}_{E'}+\rho^P_{A'B'E}\ot\proj{1}_{E'})/2$, then, by the same argument used at the end of the proof of Lemma~\ref{lem:marginals}, the entropy and scores are unchanged while the state satisfies the required conditions.
\end{proof}

The next step is to show that the state on $A'B'$ can be taken to come from the set of density operators that are diagonal in the Bell basis.  We define
\begin{align}
\cS_B:=&\left\{\lambda_0\proj{\Phi_0}+\lambda_1\proj{\Phi_1}+\lambda_2\proj{\Phi_2}+\lambda_3\proj{\Phi_3}:1\geq\lambda_0\geq\lambda_3\geq0,\,1\geq\lambda_1\geq\lambda_2\geq0,\vphantom{\sum_i}\right.\nonumber\\
    &\ \ \left.\lambda_0-\lambda_3\geq\lambda_1-\lambda_2,\sum_i\lambda_i=1\right\}\,,\label{eq:bellset}
\end{align}
where the states $\{\ket{\phi_i}\}_i$ are defined by~\eqref{bell1}--\eqref{bell4}.

\begin{lemma}
Any optimization of the form~\eqref{opt_2} is equivalent to an optimization of the same form but where each of the projectors are onto states of the form $\cos(\alpha)\ket{0}+\sin(\alpha)\ket{1}$ with $\alpha\in\mathbb{R}$ and $\rho_{A'B'}\in\cS_B$.
\end{lemma}
\begin{proof}
  From Corollary~\ref{cor:trans}, we have that $\rho_{A'B'}$ can be taken to be invariant under $\sigma_2\ot\sigma_2$. Hence we can write
  \begin{equation}\label{eq:bell_form}
\rho_{A'B'} = \left(\begin{array}{cccc}\lambda'_0&0&0&r_1\\0&\lambda'_1&r_2&0\\0&r_2^*&\lambda'_2&0\\r_1^*&0&0&\lambda'_3\end{array}\right)\,
    \end{equation}
where the matrix is expressing the coefficients in the Bell basis. That $\rho_{A'B'}=\rho_{A'B'}^T$ then implies that $r_1$ and $r_2$ are real. Note that in order that $\rho_{A'B'}$ is a positive operator we require $r_1^2\leq\lambda'_0\lambda'_3$ and $r_2^2\leq\lambda'_1\lambda'_2$.

    Let $U_\theta=\cos(\theta/2)\proj{0}+\sin(\theta/2)\ketbra{0}{1}-\sin(\theta/2)\ketbra{1}{0}+\cos(\theta/2)\proj{1}$, so that $U_\theta$ preserves the set $\{\cos(\alpha)\ket{0}+\sin(\alpha)\ket{1}:\alpha\in\mathbb{R}\}$. We proceed to show that for any state of the form~\eqref{eq:bell_form} with $r_1$ and $r_2$ real, there exist values of $\theta_A$ and $\theta_B$ such that
  $$\rho'_{A'B'} = (U_{\theta_A}\ot U_{\theta_B})\rho_{A'B'}(U^\dagger_{\theta_A}\ot U^\dagger_{\theta_B})$$
  is diagonal in the Bell basis.  We can compute the form of $\rho'_{A'B'}$ in the Bell basis.  This has the same form as~\eqref{eq:bell_form}, but with $r_1$ replaced by $r_1\cos(\theta_A-\theta_B)+\frac{\lambda'_0-\lambda'_3}{2}\sin(\theta_A-\theta_B)$ and $r_2$ replaced by
  $r_2\cos(\theta_A+\theta_B)+\frac{\lambda'_2-\lambda'_1}{2}\sin(\theta_A+\theta_B)$. To make these zero we need to choose $\theta_A$ and $\theta_B$ such that $\cos^2(\theta_A-\theta_B)=\frac{(\lambda'_0-\lambda'_3)^2}{(\lambda'_0-\lambda'_3)^2+4r_1^2}$ and $\cos^2(\theta_A+\theta_B)=\frac{(\lambda'_1-\lambda'_2)^2}{(\lambda'_1-\lambda'_2)^2+4r_2^2}$. If we write
  \begin{align*}
    \phi_1=\cos^{-1}\left(\frac{\lambda'_0-\lambda'_3}{\sqrt{(\lambda'_0-\lambda'_3)^2+4r_1^2}}\right),\quad&\phi_2=\cos^{-1}\left(\frac{\lambda'_1-\lambda'_2}{\sqrt{(\lambda'_1-\lambda'_2)^2+4r_2^2}}\right),\\
    \zeta_A=\frac{\phi_1+\phi_2}{2}\quad\text{and}\quad&\zeta_B=\frac{\phi_1-\phi_2}{2}
\end{align*}    
then we can express the four solutions
$$(\theta_A,\theta_B)=(\zeta_A,\zeta_B),\,(\zeta_A+\pi/2,\zeta_B-\pi/2),\,(\zeta_A+\pi/2,\zeta_B+\pi/2),\,(\zeta_A+\pi,\zeta_B)\,.$$
Each of these brings the state into the form $\rho_{A'B'}=\lambda_0\proj{\Phi_0}+\lambda_1\proj{\Phi_1}+\lambda_2\proj{\Phi_2}+\lambda_3\proj{\Phi_3}$.
The difference between the first two of these is an exchange of $\lambda_0$ with $\lambda_3$, the difference between the first and the third is an exchange of $\lambda_1$ with $\lambda_2$ and the difference between the first and the fourth is an exchange of $\lambda_0$ with $\lambda_3$ and of $\lambda_1$ with $\lambda_2$. It follows that we can ensure $\lambda_0\geq\lambda_3$ and $\lambda_1\geq\lambda_2$.  Finally, if $\lambda_0-\lambda_3<\lambda_1-\lambda_2$ we can apply $\sigma_3\ot\id$ to the resulting state, which simultaneously switches $\lambda_0$ with $\lambda_1$ and $\lambda_2$ with $\lambda_3$, while again preserving the set $\{\cos(\alpha)\ket{0}+\sin(\alpha)\ket{1}:\alpha\in\mathbb{R}\}$.
\end{proof}

The culmination of this section is the following.
\begin{lemma}\label{lem:bellred}
  For given $p_{XY}$, let
  \begin{eqnarray*}
    \cR_1(\score)&:=&\{(\cN,\rho_{A'B'E}):\cN\in\cC_{\Pi_1,2}(p_{XY}),\, \rho_{A'B'E}\in\cS(\cH_{A'}\ot\cH_{B'}\ot\cH_E),
              \,\scorefunction((\cN\ot\cI_E)(\rho_{A'B'E}))=\score\}\quad\text{and}\\
    \cR_2(\score)&:=&\{(\cN,\rho_{A'B'E}):\cN\in\cC_{\Pi_1,2}(p_{XY}),\, \rho_{A'B'E}\in\cS(\cH_{A'}\ot\cH_{B'}\ot\cH_E),
\,\rho_{A'B'}\in\cS_B,\,\scorefunction((\cN\ot\cI_E)(\rho_{A'B'E}))=\score\}\,.
  \end{eqnarray*}
We have $\inf_{\cR_2(\score)}\bar{H}((\cN\ot\cI_E)(\rho_{A'B'E}))=\inf_{\cR_1(\score)}  \bar{H}((\cN\ot\cI_E)(\rho_{A'B'E}))$ for any of the six entropy functions $\bar{H}$.
\end{lemma}

\subsection{Reduction to pure states}
Here we show that it is sufficient to restrict any of the optimizations we are interested in to pure states.
\begin{lemma}\label{lem:pure}
  For given $p_{XY}$ let $\cR_2(\score)$ be as in Lemma~\ref{lem:bellred} and consider
  \begin{eqnarray*}
    \cR_3(\score)&:=&\{(\cN,\rho_{A'B'E}):\cN\in\cC_{\Pi_1,2}(p_{XY}),\, \rho_{A'B'E}\in\cS_P(\cH_{A'}\ot\cH_{B'}\ot\cH_E),\,\rho_{A'B'}\in\cS_B,
                      \,\scorefunction((\cN\ot\cI_E)(\rho_{A'B'E}))=\score\}\,.
    \end{eqnarray*}
We have $\inf_{\cR_3(\score)}  \bar{H}((\cN\ot\cI_E)(\rho_{A'B'E}))=\inf_{\cR_2(\score)}  \bar{H}((\cN\ot\cI_E)(\rho_{A'B'E}))$ for any of the six entropy functions $\bar{H}$.
\end{lemma}
\begin{proof}
Since $\cR_3\subset\cR_2$, we have $\inf_{\cR_3(\score)}  \bar{H}((\cN\ot\cI_E)(\rho_{A'B'E}))\geq\inf_{\cR_2(\score)}  \bar{H}((\cN\ot\cI_E)(\rho_{A'B'E}))$. For the other direction, consider a state $\rho_{A'B'E}$ from the set $\cR_2$, and let $\rho_{A'B'EE'}$ be its purification. Using the strong subadditivity of the von-Neumann entropy, $H(C|ZEE')\leq H(C|ZE)$, a new strategy in which the only change is that Eve holds a purification of $\rho_{A'B'E}$ cannot increase any of the entropic quantities of interest and makes no change to the score. Thus, $\inf_{\cR_3(\score)}  \bar{H}((\cN\ot\cI_E)(\rho_{A'B'E}))\leq\inf_{\cR_2(\score)}  \bar{H}((\cN\ot\cI_E)(\rho_{A'B'E}))$.
\end{proof}

Note that this also means that we can restrict $\cH_E$ to be 4 dimensional.

\section{Simplifications of qubit strategies for specific entropic quantities}\label{app:simp}
In this section, we compute expressions for each of the entropies of interest, based on the simplifications from the previous section. In other words, we are considering the optimizations
\begin{align}
G(\score,p_{XY}):=\min\ &\bar{H}((\cN\ot\cI_E)(\rho_{A'B'E}))\label{opt_simp}\\
&\cH_{A'}=\cH_{B'}=\mathbb{C}^2,\ \cH_E=\mathbb{C}^4,\ \rho_{A'B'E}\in\cS_P(\cH_{A'}\ot\cH_{B'}\ot\cH_E),\ \rho_{A'B'}\in\cS_B\nonumber\\
                                                   &\cN:\sigma_{A'B'}\mapsto\sum_{abxy}p_{XY}(x,y)\proj{a}\ot\proj{b}\ot\proj{x}\ot\proj{y}\tr\left(\left(\proj{\phi^A_{a|x}}\ot\proj{\phi^B_{b|y}}\right)\sigma_{A'B'}\right)\nonumber\\
                         &\ket{\phi^A_{a|x}}=\cos(\alpha_{a|x})\ket{0}+\sin(\alpha_{a|x})\ket{1}\quad\text{and}\quad\ket{\phi^B_{b|y}}=\cos(\beta_{b|y})\ket{0}+\sin(\beta_{b|y})\ket{1}\label{opt_simp2}\\
&\alpha_{1|x}=\pi/2+\alpha_{0|x}\quad\text{and}\quad\beta_{1|x}=\pi/2+\beta_{0|x}\\
                         &\scorefunction((\cN\ot\cI_E)(\rho_{A'B'E}))=\score.
\end{align}
For convenience we sometimes use $\alpha_x=\alpha_{0|x}$ and $\beta_x=\beta_{0|x}$.

Let
\begin{eqnarray}
  \tau_{ABXYE}&=&(\cN\ot\cI_E)(\rho_{A'B'E})\nonumber\\
&=&\sum_{abxy}p_{XY}(x,y)\proj{a}\ot\proj{b}\ot\proj{x}\ot\proj{y}\ot\tr_{A'B'}\left(\left(\proj{\phi^A_{a|x}}\ot\proj{\phi^B_{b|y}}\ot\id_E\right)\rho_{A'B'E}\right)\nonumber\\
&=&\sum_{abxy}p_{XY}(x,y)p_{AB|xy}(a,b)\proj{a}\ot\proj{b}\ot\proj{x}\ot\proj{y}\ot\tau_E^{abxy}\label{eq:tau}\,,
\end{eqnarray}
where $\{\tau_E^{abxy}\}$ are normalized.

We make a few initial observations.

Consider $p_{AB|xy}(a,b)\tau_E^{abxy}=\tr_{A'B'}\left(\left(\proj{\phi^A_{a|x}}\ot\proj{\phi^B_{b|y}}\ot\id_E\right)\rho_{A'B'E}\right)$. Since $\rho_{A'B'E}$ is pure, we can use the Schmidt decomposition to write $\rho_{A'B'E}=\proj{\Phi}_{A'B'E}$, where
$$\ket{\Phi}_{A'B'E}=\sum_i\sqrt{\lambda_i}\ket{\Psi_i}\ot\ket{i}\,$$
where $\{\ket{i}\}$ is an orthonormal basis for $\cH_E$. We have
\begin{eqnarray*}
p_{AB|xy}(a,b)\tau_E^{abxy}=\sum_{ij}\sqrt{\lambda_i\lambda_j}\left(\bra{\phi^A_{a|x}}\ot\bra{\phi^B_{b|y}}\right)\ket{\Psi_i}\bra{\Psi_j}\left(\ket{\phi^A_{a|x}}\ot\ket{\phi^B_{b|y}}\right)\ketbra{i}{j}=\proj{\zeta^{abxy}}\,,
\end{eqnarray*}
where
\begin{eqnarray}\label{eq:zeta}
\ket{\zeta^{abxy}}=\!\sum_i\sqrt{\lambda_i}(\bra{\phi^A_{a|x}}\ot\bra{\phi^B_{b|y}})\ket{\Psi_i}\ket{i}\!\quad\!\text{and}\!\quad\! p_{AB|xy}(a,b)=\!\sum_i\lambda_i\left|(\bra{\phi^A_{a|x}}\ot\bra{\phi^B_{b|y}})\ket{\Psi_i}\right|^2\,.
\end{eqnarray}
Hence $\tau_E^{abxy}$ is pure for each $a,b,x,y$. Note also that
\begin{eqnarray}
  \left(\bra{\phi^A_{a|x}}\ot\bra{\phi^B_{b|y}}\right)\ket{\Psi_0}&=&\frac{\cos(\beta_{b|y}-\alpha_{a|x})}{\sqrt{2}}\label{eq:innerp1}\\
  \left(\bra{\phi^A_{a|x}}\ot\bra{\phi^B_{b|y}}\right)\ket{\Psi_1}&=&\frac{\cos(\beta_{b|y}+\alpha_{a|x})}{\sqrt{2}}\\
  \left(\bra{\phi^A_{a|x}}\ot\bra{\phi^B_{b|y}}\right)\ket{\Psi_2}&=&\frac{\sin(\beta_{b|y}+\alpha_{a|x})}{\sqrt{2}}\\
  \left(\bra{\phi^A_{a|x}}\ot\bra{\phi^B_{b|y}}\right)\ket{\Psi_3}&=&\frac{\sin(\beta_{b|y}-\alpha_{a|x})}{\sqrt{2}}\label{eq:innerp4}
\end{eqnarray}

Because $\tau_{ABXYE}$ is formed from $\rho_{A'B'E}$ without acting on $E$, we have $H(E)_\tau=H(E)_\rho$, and because $\rho_{A'B'E}$ is pure, $H(E)_\rho=H(A'B')_\rho=H(\{\lambda_0,\lambda_1,\lambda_2,\lambda_3\})$.\footnote{We use $H$ for both the von Neumann and Shannon entropies; if a list of probabilities is given as the argument to $H$ it signifies the Shannon entropy.} For the same reason, $\sum_{ab}p_{AB}(a,b)\tau_E^{abxy}=\tau_E$ for all $x,y$.

\begin{lemma}\label{lem:ent_simp}
For $\sigma_{AXE}=\sum_{ax}p_{AX}(a,x)\proj{a}\ot\proj{x}\ot\sigma^{a,x}_E$, we have
$$H(A|XE)=H(A|X)+\sum_{ax}p_{AX}(a,x)H(\sigma^{a,x}_E)-\sum_xp_X(x)H\left(\sum_ap_{A|x}(a)\sigma^{a,x}_E\right)\,.$$
\end{lemma}
\begin{proof}
We have 
\begin{align*}
    H(A|XE)&=H(AXE)-H(XE)=H(AX)+\sum_{ax}p_{AX}(a,x)H(E|A=a,X=x)-H(X)-\sum_xp_X(x)H(E|X=x)\\
    &=H(A|X)+\sum_{ax}p_{AX}(a,x)\left(H(\sigma^{a,x}_E)-H\left(\sum_{a'}p_{A|x}(a')\sigma^{a',x}_E\right)\right)\,.\qedhere
\end{align*}
\end{proof}

We can parameterize the Bell diagonal state in the following way:
\begin{eqnarray}
  \lambda_0&=&\frac{1}{4}+\frac{R\cos(\theta)}{2}+\delta\label{param1}\\
  \lambda_1&=&\frac{1}{4}+\frac{R\sin(\theta)}{2}-\delta\label{param2}\\
  \lambda_2&=&\frac{1}{4}-\frac{R\sin(\theta)}{2}-\delta\\
  \lambda_3&=&\frac{1}{4}-\frac{R\cos(\theta)}{2}+\delta\label{param4}
\end{eqnarray}
where $0\leq R\leq1$, $0\leq\theta\leq\pi/4$ if $R\leq1/\sqrt{2}$, or $0\leq\theta\leq\pi/4-\cos^{-1}(1/(R\sqrt{2}))$ if $R>1/\sqrt{2}$ and $-1/4+R\cos(\theta)/2\leq\delta\leq1/4-R\sin(\theta)/2$.

\begin{lemma}\label{lem:delta_opt}
For $R>1/\sqrt{2}$, $\max_\delta H(\{\lambda_0,\lambda_1,\lambda_2,\lambda_3\})$ is achieved when $\delta=\delta^*:=\frac{R^2\cos(2\theta)}{4}$.
\end{lemma}
\begin{proof}
One can compute the derivative of $H(\{\lambda_0,\lambda_1,\lambda_2,\lambda_3\})$ with respect to $\delta$ to see that it is $0$ only for $\delta=\delta^*:=\frac{R^2\cos(2\theta)}{4}$. 

We next check that $\delta^*$ is in the valid range of $\delta$. The condition $\delta^*\leq1/4-R\sin(\theta)/2$ rearranges to
  $2R^2\sin^2(\theta)-2R\sin(\theta)+1-R^2\geq0$. The roots of the quadratic equation $2R^2x^2-2Rx+1-R^2$ are at $x=\frac{1}{2R}(1\pm\sqrt{2R^2-1})$. For $R\geq\frac{1}{\sqrt{2}}$ the roots are real\footnote{If $R<\frac{1}{\sqrt{2}}$ there are no real roots and the condition always holds.}. Our condition on $\theta$ implies that $0\leq\sin(\theta)\leq\frac{1}{2R}(1-\sqrt{2R^2-1})$, hence taking $x=\sin(\theta)$ we are always to the left of the first root and so $\delta^*\leq1/4-R\sin(\theta)/2$. A similar argument shows $\delta^*\geq-1/4+R\cos(\theta)/2$.

We can then compute the double derivative of $H(\{\lambda_0,\lambda_1,\lambda_2,\lambda_3\})$ with respect to $\delta$ and evaluate it at $\delta^*$.  This gives $-\frac{32}{\ln(2)(2-4R^2+R^4+R^4\cos(4\theta))}$, which can be shown to be negative for $R>1/\sqrt{2}$ and any valid $\theta$ using a similar argument to that above.
\end{proof}

Using this state and the specified measurements, the probability table for the observed distribution has the form given in the table below (whose entries correspond to $p_{AB|XY}$):\smallskip

\begin{center}
  \begin{tabular}{|c |c| c  c| c c| }
  \hline
  &&$X=0$&&$X=1$&\\
       & & $A=0$ & $A=1$  & $A=0$ & $A=1$\\
\hline
$Y=0$& $B=0$   & $\epsilon_{00}$ &  \vphantom{$\frac{\sum^a}{f}$} $\frac{1}{2} -\epsilon_{00} $ & $\epsilon_{10}$ &  $\frac{1}{2} -\epsilon_{10} $\\  
&  $B=1$ & \vphantom{$\frac{\sum^a}{f}$} $\frac{1}{2} - \epsilon_{00}$ &   $\epsilon_{00}$ & $\frac{1}{2} - \epsilon_{10}$ &   $\epsilon_{10}$     \\ 
  \hline
$Y=1$&  $B=0$  & $\epsilon_{01}$ &  \vphantom{$\frac{\sum^a}{f}$} $\frac{1}{2} -\epsilon_{01} $ & $\frac{1}{2} - \epsilon_{11}$ &   $\epsilon_{11}$ \\
&  $B=1$ &  \vphantom{$\frac{\sum^a}{f}$}$\frac{1}{2} - \epsilon_{01}$ &   $\epsilon_{01}$& $\epsilon_{11}$ &  $\frac{1}{2} -\epsilon_{11} $\\ 
  \hline
\end{tabular}\end{center}\smallskip
where
\begin{align*}
  \epsilon_{00}&=\frac{1}{4}\left(1+R\cos(\theta)\cos(2(\alpha_0-\beta_0))+R\sin(\theta)\cos(2(\alpha_0+\beta_0))\right)\\
  \epsilon_{01}&=\frac{1}{4}\left(1+R\cos(\theta)\cos(2(\alpha_0-\beta_1))+R\sin(\theta)\cos(2(\alpha_0+\beta_1))\right)\\
  \epsilon_{10}&=\frac{1}{4}\left(1+R\cos(\theta)\cos(2(\alpha_1-\beta_0))+R\sin(\theta)\cos(2(\alpha_1+\beta_0))\right)\\
  \epsilon_{11}&=\frac{1}{4}\left(1-R\cos(\theta)\cos(2(\alpha_1-\beta_1))-R\sin(\theta)\cos(2(\alpha_1+\beta_1))\right)\,.
\end{align*}

Note that
\begin{align}
\scorefunction(\tau_{ABXY})=&(\epsilon_{00}+\epsilon_{01}+\epsilon_{10}+\epsilon_{11})/2\nonumber\\
    =&\frac{1}{2}+\frac{R\cos(\theta)}{8}\bigl(\cos(2 (\alpha_0 - \beta_0))+ 
   \cos(2 (\alpha_0 - \beta_1))+ 
   \cos(2 (\alpha_1 - \beta_0))- 
    \cos(2 (\alpha_1 - \beta_1))\bigr)\nonumber\\
    &+\frac{R\sin(\theta)}{8}\bigl( 
   \cos(2 (\alpha_0 + \beta_0))+ 
   \cos(2 (\alpha_0 + \beta_1))+ 
   \cos(2 (\alpha_1 + \beta_0))- 
   \cos(2 (\alpha_1 + \beta_1))\bigr)\,,\label{eq:score}
\end{align}
which is independent of $\delta$.

\subsection{H(A|X=0,Y=0,E)}
This case was already covered in~\cite{PABGMS} where it was solved analytically (see also~\cite{WAP} for a slight generalization).
\begin{lemma}\label{lem:Ag00E}
For $3/4\leq\omega\leq\frac{1}{2}(1+\frac{1}{\sqrt{2}})$ the solution to the optimization problem~\eqref{opt_simp} when $\bar{H}=H(A|X=0,Y=0,E)$ is $1-H_\bin\left(\frac{1}{2}(1+\sqrt{16\score(\score-1)+3})\right)$.
\end{lemma}
For completeness we give a proof here as well. We first show that the maximum CHSH score for a Bell diagonal state depends only on $R$.
\begin{lemma}\label{lem:Ag00Escore}
Given a state $\rho_{A'B'E}$ with $\rho_{A'B'}$ parameterized as in~\eqref{param1}--\eqref{param4}, if $\cN$ satisfies the requirements of the optimization problem~\eqref{opt_simp}, then $\tau_{ABXYE}=(\cN\ot\cI_E)(\rho_{A'B'E})$ satisfies $\scorefunction(\tau_{ABXYE})\leq\frac{1}{2}+\frac{R}{2\sqrt{2}}$, and there exists a channel $\cN$ achieving equality.
\end{lemma}
\begin{proof}
Consider the score function~\eqref{eq:score}. Collecting all the terms involving $\alpha_0$ and $\alpha_1$ and some manipulation gives
\begin{align}
\scorefunction(\tau_{ABXYE})=&\frac{1}{2}+\frac{R}{2\sqrt{2}}\cos(\beta_0-\beta_1)\left[\sin(2\alpha_0)\sin(\beta_0+\beta_1)\cos(\frac{\pi}{4}+\theta)+\cos(2\alpha_0)\cos(\beta_0+\beta_1)\sin(\frac{\pi}{4}+\theta)\right]\nonumber\\
&+\frac{R}{2\sqrt{2}}\sin(\beta_0-\beta_1)\left[\sin(2\alpha_1)\cos(\beta_0+\beta_1)\cos(\frac{\pi}{4}+\theta)-\cos(2\alpha_1)\sin(\beta_0+\beta_1)\sin(\frac{\pi}{4}+\theta)\right],\label{eq:scoresimp}
\end{align}
where we have used $\cos(\theta)+\sin(\theta)=\sqrt{2}\sin(\pi/4+\theta)$ and $\cos(\theta)-\sin(\theta)=\sqrt{2}\cos(\pi/4+\theta)$.
  For brevity we write $\bar{\theta}=\frac{\pi}{4}+\theta$. We then use that for $r,t,\phi\in\mathbb{R}$ we have $r\cos(\phi)+t\sin(\phi)\leq\sqrt{r^2+t^2}$ with equality if $\phi$ is chosen such that $r\cos(\phi)+t\sin(\phi)\geq0$ and $r\sin(\phi)=t\cos(\phi)$.  This allows us to form the bound
  \begin{align}
\scorefunction(\tau_{ABXYE})\leq&\frac{1}{2}+\frac{R}{2\sqrt{2}}\left(|\cos(\beta_0-\beta_1)|\sqrt{\sin^2(\beta_0+\beta_1)\cos^2(\bar{\theta})+\cos^2(\beta_0+\beta_1)\sin^2(\bar{\theta})}\right.\nonumber\\
                           &\left.+|\sin(\beta_0-\beta_1)|\sqrt{\cos^2(\beta_0+\beta_1)\cos^2(\bar{\theta})+\sin^2(\beta_0+\beta_1)\sin^2(\bar{\theta})}\right)\label{eq:score1}\\
    \phantom{\scorefunction}\leq&\frac{1}{2}+\frac{R}{2\sqrt{2}}\,.
  \end{align}
Choosing $\alpha_0=0$, $\alpha_1=\pi/4$, $\beta_0=\frac{\pi}{8}-\frac{\theta}{2}$, $\beta_1=-\frac{\pi}{8}+\frac{\theta}{2}$ achieves equality (for instance).
\end{proof}
It follows that $\score>3/4$ is only possible if $R>1/\sqrt{2}$.

We now turn to the entropy. In this case we trace out $B$ from the state $\tau'$ in Section~\ref{app:AB00E} to give
$\tau'_{AE}=\sum_ap_{A|00}(a)\proj{a}\ot\sum_bp_{B|a00}(b)\tau_E^{ab00}$, so that $H(A|X=0,Y=0,E)_\tau=H(A|E)_{\tau'}$. Using Lemma~\ref{lem:ent_simp} we have
\begin{align*}
H(A|E)_{\tau'}&=H(A)_{\tau'}+\sum_ap_{A|00}(a)H\left(\sum_bp_{B|a00}(b)\tau_E^{ab00}\right)-H\left(\sum_{ab}p_{AB|00}(a,b)\tau_E^{ab00}\right)\\
&=1+\sum_a\frac{1}{2}H\left(\sum_bp_{B|a00}(b)\tau_E^{ab00}\right)-H(E)_{\rho}\\
              &=1+\sum_a\frac{1}{2}H\left(\sum_b2p_{AB|00}(a,b)\tau_E^{ab00}\right)-H(E)_{\rho}\,,
\end{align*}
where we have used the fact that $p_{A|00}(a)=1/2$ for $a=0,1$. The eigenvalues of $\sum_b2p_{AB|00}(a,b)\tau_E^{ab00}$ turn out to be
$$\frac{1}{2}\left(1\pm\sqrt{2(\lambda_0-\lambda_3)(\lambda_1-\lambda_2)\cos(4\alpha_0)+(\lambda_0-\lambda_3)^2+(\lambda_1-\lambda_2)^2}\right)\,,$$
independently of $a$.  Hence, we can write $H(A|E)_{\tau'}$ in terms of the Bell diagonal state using
\begin{align*}
\sum_a\frac{1}{2}H\left(\sum_b2p_{AB|00}(a,b)\tau_E^{ab00}\right)&=H_\bin\left(\frac{1}{2}\left(1+\sqrt{2(\lambda_0-\lambda_3)(\lambda_1-\lambda_2)\cos(4\alpha_0)+(\lambda_0-\lambda_3)^2+(\lambda_1-\lambda_2)^2}\right)\right)\\
  H(E)_\rho&=H(\{\lambda_0,\lambda_1,\lambda_2,\lambda_3\})
\end{align*}
Having established this, we show the following.
\begin{lemma}\label{lem:Ag00Eent}
Let $\rho_{A'B'E}$ be pure with $\cH_{A'}=\cH_{B'}=\mathbb{C}^2$, and let $\rho_{A'B'}$ be a Bell diagonal state parameterized by~\eqref{param1}--\eqref{param4} with $R>1/\sqrt{2}$. Let $\tau$ be the state defined by~\eqref{eq:tau}. Then $$H(A|X=0,Y=0,E)_\tau\geq1+H_\bin\left(\frac{1}{2}\left(1+\sqrt{2R^2-1}\right)\right),$$
where equality is achievable for $\alpha_0=0$.
\end{lemma}
\begin{proof}
We note that
  \begin{align*}
\sum_a\frac{1}{2}H\left(\sum_b2p_{AB|00}(a,b)\tau_E^{ab00}\right)&\geq H_\bin\left(\frac{1}{2}\left(1+\sqrt{2(\lambda_0-\lambda_3)(\lambda_1-\lambda_2)+(\lambda_0-\lambda_3)^2+(\lambda_1-\lambda_2)^2}\right)\right)\\
  &=H_\bin(\lambda_0+\lambda_1)\,,
\end{align*}
with equality for $\alpha_0=0$.
Hence
\begin{equation}\label{eq:ent45}
  H(A|E)_{\tau'}\geq 1+H_\bin(\lambda_0+\lambda_1)-H(\{\lambda_0,\lambda_1,\lambda_2,\lambda_3\})\,.
  \end{equation}

Using the parameterization of~\eqref{param1}--\eqref{param4} we have $\lambda_0+\lambda_1=1/2(1+R(\cos(\theta)+\sin(\theta)))$. Thus, the minimum of $H(A|E)_{\tau'}$ over $\delta$ is achieved for $\delta=\delta^*$ (as in the case $H(AB|X=0,Y=0,E)$). Taking $\delta=\delta^*$ and differentiating the resulting expression with respect to $\theta$ yields
$$\frac{R}{2}(\cos(\theta)+\sin(\theta))\log\left(\frac{1-R\cos(\theta)+R\sin(\theta)}{1+R\cos(\theta)-R \sin(\theta)}\right)\,.$$
Since $\cos(\theta)+\sin(\theta)=\sqrt{2}\sin(\pi/4+\theta)$, the leading factor is always positive over our range of $\theta$.  The logarithm term is always negative, except for $\theta=\pi/4$ where it reaches zero. Thus, the minimum over $\theta$ is always obtained at the largest possible $\theta$, i.e., $\theta=\pi/4-\cos^{-1}\left(1/(R\sqrt{2})\right)$.

With this substitution the right hand side of~\eqref{eq:ent45} reduces to
$$1+H_\bin\left(\frac{1}{2}\left(1+\sqrt{2R^2-1}\right)\right)\,,$$
establishing the claim.
\end{proof}

Lemma~\ref{lem:Ag00E} is then a corollary of Lemmas~\ref{lem:Ag00Escore} and~\ref{lem:Ag00Eent}.
\begin{proof}[Proof of Lemma~\ref{lem:Ag00E}]
  From Lemma~\ref{lem:Ag00Eent} we have
  $$H(A|X=0,Y=0,E)_\tau\geq1+H_\bin\left(\frac{1}{2}\left(1+\sqrt{2R^2-1}\right)\right)\,.$$
  However, Lemma~\ref{lem:Ag00Escore} then implies
  \begin{eqnarray*}
    H(A|X=0,Y=0,E)_\tau&\geq&1+H_\bin\left(\frac{1}{2}\left(1+\sqrt{4(2\score-1)^2-1}\right)\right)\\
    &=&1+H_\bin\left(\frac{1}{2}\left(1+\sqrt{16\score(\score-1)+3)}\right)\right)\,,
  \end{eqnarray*}
where we use the fact that $H_\bin(p)$ is decreasing and concave for $p\geq1/2$.
  Equality is achievable by taking $\alpha_0=0$, $\alpha_1=\pi/4$, $\beta_0=\frac{\pi}{8}-\frac{\theta}{2}$, $\beta_1=-\frac{\pi}{8}+\frac{\theta}{2}$.  
\end{proof}

We use this case to gain confidence in our numerics, since we can make a direct comparison to the analytic curve.

\subsection{H(A|XYE)}\label{app:AgXYE}
In this case Lemma~\ref{lem:ent_simp} gives
\begin{eqnarray*}
  H(A|XYE)&=&H(A|XY)+\sum_{axy}p_{AXY}(a,x,y)H\left(\sum_bp_{B|axy}(b)\tau_E^{abxy}\right)\\
  &&-\sum_{xy}p_{XY}(x,y)H\left(\sum_{ab}p_{AB|xy}(a,b)\tau_E^{abxy}\right)\\
        &=&1+\sum_{axy}p_{XY}(x,y)p_{A|xy}(a)H\left(\sum_b2p_{AB|xy}(a,b)\tau_E^{abxy}\right)-H(E)\,.
\end{eqnarray*}
The eigenvalues of $\sum_b2p_{AB|xy}(a,b)\tau_E^{abxy}$ can be computed to be
$$\frac{1}{2}\left(1\pm\sqrt{2(\lambda_0-\lambda_3)(\lambda_1-\lambda_2)\cos(4\alpha_x)+(\lambda_0-\lambda_3)^2+(\lambda_1-\lambda_2)^2}\right)\,,$$
independently of $a,y$. If we define $$g(\alpha):=\frac{1}{2}\left(1+\sqrt{2(\lambda_0-\lambda_3)(\lambda_1-\lambda_2)\cos(4\alpha)+(\lambda_0-\lambda_3)^2+(\lambda_1-\lambda_2)^2}\right)$$
then
\begin{align*}
H(A|XYE)&=1+\sum_xp_X(x)H_\bin(g(\alpha_x))-H(E)\,.
\end{align*}
Using the parameterization~\eqref{param1}--\eqref{param4} we have
$$g(\alpha)=\frac{1}{2}\left(1+R\sqrt{1+\sin(2\theta)\cos(4\alpha)}\right).$$
Since this is independent of $\delta$, we can again use $\delta=\delta^*$ (cf.\ Lemma~\ref{lem:delta_opt}) to remove one parameter when minimizing $H(A|XYE)$.

We now restrict to the case where $p_X(x)=1/2$ for $x=0,1$. Since $H(E)$ is independent of $\{\alpha_x\}$, we can consider the optimization
\begin{align}
  \min_{\alpha_0,\alpha_1,\beta_0,\beta_1} &\ H_\bin(g(\alpha_0))+H_\bin(g(\alpha_1)) \nonumber\\
\text{subject to } &\ \scorefunction(\tau_{ABXYE})=\score\label{eq:redopt}
\end{align}
for some fixed values of $\score$, $R$ and $\theta$.

We proceed to make a series of simplifications of this optimization.
\begin{lemma}
  The optimization~\eqref{eq:redopt} is equivalent to
\begin{align}
  \min_{u,v} &\ H_\bin(g((v+u)/4))+H_\bin(g((v-u)/4))\nonumber\\
  \mathrm{subject\ to\ }\ &\ S(u,v):=\frac{1}{2}+\frac{R}{4}\left(\cos(u/2)\sqrt{1+\cos(v)\sin(2\theta)}+\sin(u/2)\sqrt{1-\cos(v)\sin(2\theta)}\right)=\score\label{eq:redopt2}\\
  &\ 0\leq u\leq\pi\nonumber
\end{align}
\end{lemma}
\begin{proof}
Noting that the objective function in~\eqref{eq:redopt} is independent of Bob's angles ($\beta_0$ and $\beta_1$), analogously to the derivation of~\eqref{eq:score1} we can bound the score function using
\begin{align*}
S(\tau_{ABXYE})\leq&\ \frac{1}{2}+\frac{R}{2\sqrt{2}}\left(\left|\cos(\alpha_0-\alpha_1)\right|\sqrt{\sin^2(\alpha_0+\alpha_1)\cos^2(\bar{\theta})+\cos^2(\alpha_0+\alpha_1)\sin^2(\bar{\theta})}\right.\\                           &\left.+\left|\sin(\alpha_0-\alpha_1)\right|\sqrt{\cos^2(\alpha_0+\alpha_1)\cos^2(\bar{\theta})+\sin^2(\alpha_0+\alpha_1)\sin^2(\bar{\theta})}\right).
\end{align*}
We now substitute $v/2=\alpha_0+\alpha_1$ and $u/2=\alpha_0-\alpha_1$ and rearrange (recalling that $\bar{\theta}=\pi/4+\theta$) to give
\begin{align*}
S(\tau_{ABXYE})\leq&\frac{1}{2}+\frac{R}{4}\left(\left|\cos(u/2)\right|\sqrt{1+\cos(v)\sin(2\theta)}+\left|\sin(u/2)\right|\sqrt{1-\cos(v)\sin(2\theta)}\right).
\end{align*}
Since $G_{A|XYE}(\score)$ is monotonically increasing in $\score$ (cf.\ Appendix~\ref{app:mono}), it follows that we wish to choose the angles to achieve the largest possible score function.

We first note that if $\cos(u/2)<0$ we can make the substitution $\alpha_0\mapsto \pi/2+\alpha_1$ and $\alpha_1\mapsto\alpha_0-\pi/2$ which maintains the objective function, constraint, $v$ and $\sin(u/2)$, while changing the sign of $\cos(u/2)$. In addition, if $\sin(u/2)<0$, the substitution $\alpha_0\mapsto\alpha_1$ and $\alpha_1\mapsto\alpha_0$ maintains the objective function, constraint, $v$ and $\cos(u/2)$ while changing the sign of $\sin(u/2)$. It follows that the maximum of $H_\bin(g((v+u)/4))+H_\bin(g((v-u)/4))$ for fixed score is obtained when $$S(\tau_{ABXYE})=S(u,v):=\frac{1}{2}+\frac{R}{4}\left(\cos(u/2)\sqrt{1+\cos(v)\sin(2\theta)}+\sin(u/2)\sqrt{1-\cos(v)\sin(2\theta)}\right)$$
and when both $\cos(u/2)\geq0$ and $\sin(u/2)\geq0$, or, alternatively $0\leq u\leq\pi$.
\end{proof}
\begin{lemma}
  In the optimization~\eqref{eq:redopt2} we can restrict to $0\leq u\leq\pi/2$ and $0\leq v\leq\pi/2$ without affecting the result.
\end{lemma}
\begin{proof}
Consider a $u$ that satisfies $0\leq u\leq\pi$. If $\cos(u/2)\geq\sin(u/2)$ then $0\leq u\leq\pi/2$. Otherwise, consider $u\mapsto\pi-u$, $v\mapsto\pi-v$. This maintains the constraint and the value of the objective function and hence the optimal value can still be obtained, but now with $\cos(u/2)\geq\sin(u/2)$, so we can assume $0\leq u\leq\pi/2$.

For the restriction on $v$, first note that transforming $v\mapsto-v$ has no affect on either the objective function or the constraint so we can take $\sin(v)\geq0$, or $0\leq v\leq\pi$. If $v>\pi/2$, then $\cos(v+u)<0$. Furthermore, $\cos(v-u)+\cos(v+u)=2\cos(v)\cos(u)\leq0$ and hence $\cos(v-u)\leq|\cos(v+u)|$.  Let $\bar{v}=\pi-v$. We have $$S(u,v)-S(u,\bar{v})=\frac{R}{4}\left(\left(\sin(u/2)-\cos(u/2)\right)\left(\sqrt{1-\cos(v)\sin(2\theta)}-\sqrt{1+\cos(v)\sin(2\theta)}\right)\right)\leq0\,,$$
where the inequality follows from $\cos(v)<0$, $\sin(2\theta)\geq0$ and $\cos(u/2)\geq\sin(u/2)$. Hence, the mapping $v\mapsto\pi-v$ increases $S(u,v)$.

Consider now the effect on the objective function $$J(u,v):=H_\bin\left(\frac{1}{2}\left(1+R\sqrt{1+\sin(2\theta)\cos(v+u)}\right)\right)+H_\bin\left(\frac{1}{2}\left(1+R\sqrt{1+\sin(2\theta)\cos(v-u)}\right)\right).$$
Note that each binary entropy term decreases as its cosine term increases. We have
$$J(u,\bar{v}):=H_\bin\left(\frac{1}{2}\left(1+R\sqrt{1-\sin(2\theta)\cos(v+u)}\right)\right)+H_\bin\left(\frac{1}{2}\left(1+R\sqrt{1-\sin(2\theta)\cos(v-u)}\right)\right).$$

If $\cos(v-u)\leq0$ then $J(u,\bar{v})\leq J(u,v)$, so the transformation decreases the objective function.

On the other hand, if $\cos(v-u)\geq0$ we must have $\cos(v-u)\leq|\cos(v+u)|$ and hence $$\sqrt{1-\sin(2\theta)\cos(v+u)}\geq\sqrt{1+\sin(2\theta)\cos(v-u)}\geq\sqrt{1-\sin(2\theta)\cos(v-u)}\geq\sqrt{1+\sin(2\theta)\cos(v+u)}.$$
Thus, $J(u,\bar{v})\leq J(u,v)$ and the transformation decreases the objective function.

Thus, in both cases the transformation decreases the objective function while increasing the score. Using the monotonicity of $G_{A|XYE}(\score)$ with the score $\score$ (cf.\ Appendix~\ref{app:mono}), it follows that we can further reduce the objective function while bringing the score back to its original level.
\end{proof}

Let us turn to the constraint. We have
$$\score=\frac{1}{2}+\frac{R}{4}\left(\cos(u/2)\sqrt{1+\cos(v)\sin(2\theta)}+\sin(u/2)\sqrt{1-\cos(v)\sin(2\theta)}\right)=\frac{1}{2}+\frac{R}{2\sqrt{2}}\cos(u/2-\phi)\,,$$
where $\cos(\phi)=\sqrt{(1+\cos(v)\sin(2\theta))/2}$ and $\sin(\phi)=\sqrt{(1-\cos(v)\sin(2\theta))/2}$. We can rearrange this to $\cos(u/2-\phi)=\sqrt{2}(2\score-1)/R$ and hence there are two possibilities for $u$: \begin{equation}\label{eq:uforms}
u_{\pm}=2\cos^{-1}\sqrt{(1+\cos(v)\sin(2\theta))/2}\pm2\cos^{-1}(\sqrt{2}(2\score-1)/R).
\end{equation}
We can hence remove the constraint and consider the optimizations
\begin{align*}
  \min_v\ J(u_{\pm}(v),v)
\end{align*}

Summarizing the above analysis we have the following.
\begin{corollary}\label{cor:onesidelower}
Let $\score \in(3/4,(1+1/\sqrt{2})/2]$ and $K(R,\theta)=1-H(\{\lambda_0(R,\theta),\lambda_1(R,\theta),\lambda_2(R,\theta), \lambda_3(R,\theta)\})$, where $\{\lambda_i(R,\theta)\}_{i=0}^3$ are given by~\eqref{param1}--\eqref{param4} with $\delta=\delta^*$. Defining $\M{D}_\score=\{(R,\theta,v):R\in[\sqrt{2}(2\score-1),1],\ \theta\in[0,\frac{\pi}{4}-\cos^{-1}\left(1/(\sqrt{2}R)\right)],\ v\in[0,\frac{\pi}{2}]\}$, we have
\begin{eqnarray}\label{One-sided-final opt}
  G_{A|XYE}(\score) &=&\min_{\M{D}_\score , u\in\{u_+,u_-\}}J(u(v),v)/2+K(R,\theta)\,.
\end{eqnarray}
\end{corollary}

\subsubsection{Monotonicity of $K(R,\theta)$}\label{app: Monotonicty of K}
The following monotonicity properties of the function $K(R,\theta)$ will be useful later. 
\begin{lemma}\label{lem: monotonicity of K(R)}
For any $\score\in(3/4,(1+1/\sqrt{2})/2]$, and $(R,\theta)\in \M{D}_{\score}$ we have $\partial_{R} K(R,\theta)\geq0$.
\end{lemma}
\begin{proof}
Note that 
\begin{equation*}
    \frac{\lambda_0 \lambda_3}{\lambda_1 \lambda_2} = 1 
\end{equation*}
and $\lambda_0 > \lambda_3$ and $\lambda_1 > \lambda_2$. We differentiate $K(R, \theta)$ with respect to $R$ 
\begin{eqnarray}
\partial_{R} K(R , \theta) &=&  \sum_{i} \Big( \log(\lambda_i) + \frac{1}{\ln2} \Big) \frac{\partial \lambda_i}{\partial R} \nonumber \\ 
&=& \frac{1}{2}\left(\left(\cos(\theta)+R\cos(2\theta)\right)(\log\lambda_0+\frac{1}{\ln2})+\left(\sin(\theta)-R\cos(2\theta)\right)(\log\lambda_1+\frac{1}{\ln2})-\right.\nonumber\\
&&\left.\left(\sin(\theta)+R\cos(2\theta)\right)(\log\lambda_2+\frac{1}{\ln2})-\left(\cos(\theta)-R\cos(2\theta)\right)(\log\lambda_3+\frac{1}{\ln2})\right)   \nonumber \\ 
&=& \frac{1}{2}\left(\log\Big(\frac{\lambda_0}{\lambda_3} \Big) \cos(\theta) + \log\Big(\frac{\lambda_1}{\lambda_2} \Big)\sin(\theta)+ R\cos(2 \theta)\log\Big(\frac{\lambda_0 \lambda_3}{\lambda_1 \lambda_2} \Big)\right) \nonumber \\
&=& \log\Big(\frac{\lambda_0}{\lambda_3} \Big) \cos(\theta) + \log\Big(\frac{\lambda_1}{\lambda_2} \Big)\sin(\theta) \geq0 \nonumber 
\end{eqnarray}
as claimed.
\end{proof}
We derive a similar result for monotonicity of $K(R,\theta)$ with respect to $\theta$
\begin{lemma}\label{lem: monotonicity of K(th)}
For any $\score\in(3/4,(1+1/\sqrt{2})/2]$, and $(R,\theta)\in \M{D}_{\score}$ we have $\partial_{\theta} K(R,\theta)\geq0$.
\end{lemma}
\begin{proof}
We differentiate $K(R,\theta)$ with respect to $\theta$ 
\begin{eqnarray}
\partial_{\theta} K(R , \theta) &=&  \sum_{i} \Big( \log(\lambda_i) + \frac{1}{\ln2} \Big) \frac{\partial \lambda_i}{\partial \theta} \nonumber \\ 
&=&  \sum_{i} \Big( \log(\lambda_i) \frac{\partial \lambda_i}{\partial \theta}\Big) + \frac{1}{\ln2}\sum_{i} \frac{\partial \lambda_i}{\partial \theta}  \nonumber \\ 
&=& \sum_{i} \Big( \log(\lambda_i) \frac{\partial \lambda_i}{\partial \theta}\Big)  \nonumber \\ 
&=&\frac{1}{2}\left(-\log\Big(\frac{\lambda_0}{\lambda_3} \Big) R\sin(\theta) + \log\Big(\frac{\lambda_1}{\lambda_2} \Big)R\cos(\theta) - R^2 \sin(2 \theta)\log\Big(\frac{\lambda_0 \lambda_3}{\lambda_1 \lambda_2} \Big)\right) \nonumber \\
&=& -\frac{R\sin(\theta)}{2} \log\Big(\frac{\lambda_0}{\lambda_3} \Big) + \frac{R\cos(\theta)}{2} \log\Big(\frac{\lambda_1}{\lambda_2} \Big)   \nonumber
\end{eqnarray}
We now consider the function 
\begin{eqnarray}
f(R, \theta) = \frac{2\ln2}{R}\partial_\theta K(R,\theta)=- \sin(\theta)\ln\left(\frac{\lambda_0(R, \theta)}{\lambda_3(R, \theta)}\right)+\cos(\theta)\ln\left(\frac{\lambda_0(R, \theta)}{\lambda_3(R, \theta)}\right)
\end{eqnarray}
Note that $\score\in(3/4,(1+1/\sqrt{2})/2]$ implies $1/\sqrt{2}<R\leq1$ and $0\leq\theta\leq\pi/4-\cos^{-1}(1/(R\sqrt{2}))$, or $0\leq\theta\leq\pi/4$, $1/\sqrt{2}<R\leq1/(\cos(\theta)+\sin(\theta))$. We can extend the domain of $f(R,\theta)$ to $0\leq\theta\leq\pi/4$, $0 \leq R \leq 1/(\cos(\theta) + \sin(\theta))$. \\ 
Taking derivative of $f$ with respect to $R$ gives 
\begin{eqnarray}
\partial_{R}f(R, \theta) = \frac{2 R^2 \sin (4 \theta)}{(1 - R^2(\cos(\theta) + \sin(\theta))^2) (1 - R^2(\cos(\theta) - \sin(\theta))^2 )}.
\end{eqnarray}
Thus, $\partial_{R}f(R , \theta) > 0$ whenever $\theta \in [0 ,\frac{\pi}{4}]$. We can then infer that $\partial_{\theta}K(R, \theta) = \frac{R}{2\ln2}f(R, \theta) \geq\frac{R}{2\ln2} f(0 , \theta) = 0$. 
\end{proof}

\subsubsection{Lower bounding the objective function}\label{app: lower bound using grids}
In this section, we propose a method to compute a lower bound on the function $G_{.|.E}$ by partitioning the domain. We start by considering an abstract version of the problem, which has the form
\begin{eqnarray}
 \min_{\B{x} \in \M{D}}  Q(\B{x})
\end{eqnarray}
where $\M{D} \subset \mathbb{R}^{n}$ is a compact set and $Q: \M{D} \mapsto \mathbb{R}$ is bounded.  Furthermore, we assume we know an upper bound $M$ such that $Q(\B{x})\leq M$ for all $\B{x}\in\M{D}$.

We use the notation $\M{C}_{\B{a},\B{b}}=[a_1,b_1]\times[a_2,b_2]\times\cdots\times[a_n,b_n]$, i.e., $\M{C}_{\B{a},\B{b}}$ is a hyper-cuboid with $\B{a}$ and $\B{b}$ as two opposite vertices. Then let $\M{C}\supseteq\M{D}$ be any hyper-cuboid that completely contains $\M{D}$.
We say $\M{P}=\{\M{C}_{\B{a}^i,\B{b}^i}\}_i$ is a partition of $\M{C}$ if 
\begin{eqnarray}
\M{C} = \bigcup_{i} \M{C}_{\B{a}^{i}, \B{b}^{i}} 
\end{eqnarray}
where $\{\M{C}_{\B{a}^{i}, \B{b}^{i}}\}_i$ are cuboids whose intersection has zero volume.

The main idea behind our lower bounds is to find lower bounds on $Q(\B{x})$ that hold on each cuboid and then to take the minimum of all the lower bounds.  In some cases these bounds are formed by starting from a corner and using bounds on the derivatives of $Q(\B{x})$ on the cuboid to form a bound that holds across the cuboid.  In other cases, we use monotonicity arguments to imply that evaluation at one of the corners lower bounds the whole cuboid. Some of our cuboids lie entirely outside the original domain $\M{D}$. To save calculation we assign the known upper bound on the function as the upper bound on cuboids in our partition that lie outside of $\M{D}$.

\subsubsection{Obtaining a lower bound on $G_{A|XYE}$}
We now return to our optimization problem~\eqref{One-sided-final opt}. It is convenient to switch parameterization to use $\eta := \cos^{-1}(\sqrt{2}(2 \score - 1)/ R)$ instead of $R$. Taking $\B{x}=(\eta,\theta,v)$ we rewrite~\eqref{One-sided-final opt} as \begin{eqnarray}\label{eqn: Optimization problem on grids}
G_{A|XYE}(\score) = \min_{\mathbf{x} \in \M{D}_\score,u\in\{u_+,u_-\}} F_1(\eta , \theta, v) + F_2(\eta , \theta , v) + K(R(\eta), \theta) 
\end{eqnarray}
where
\begin{eqnarray}
F_{1}(\eta, \theta , v) &=& \frac{1}{2}H_{\bin} \Big(\frac{1}{2} + \frac{R(\eta)}{2} \sqrt{1 + \cos(u(v) + v) \sin(2 \theta)}\Big) \nonumber\\ 
F_{2}(\eta, \theta , v) &=& \frac{1}{2}H_{\bin} \Big(\frac{1}{2} + \frac{R(\eta)}{2} \sqrt{1 + \cos(u(v) - v) \sin(2 \theta)}\Big) \nonumber\\ 
R(\eta)                 &=& \frac{\sqrt{2}(2 \score - 1)}{ \cos(\eta)}\,. \nonumber
\end{eqnarray}
Here the domain $\M{D}_\score$ is the set 
\begin{eqnarray}\label{eqn: Domain for one sided rate}
\M{D}_\score = \left\{ (\eta , \theta , v) : \eta\in[0,\cos^{-1}\big(\sqrt{2}(2 \score - 1)\big)] , \theta\in[0,\frac{\pi}{4}-\cos^{-1}(\cos(\eta)/(4\score-2))], v\in[0,\frac{\pi}{2}]\right\}.
\end{eqnarray} 
Define a cuboid $\M{C} \supseteq \M{D}_{\score}$ as $[0 , \cos^{-1}\big(\sqrt{2}(2 \score - 1)\big)]\times[0 , \frac{\pi}{4} - \cos^{-1}\big(1/(4\score-2)\big)]\times[0,\frac{\pi}{2}]$. We then partition $\M{C}$ as follows.  We take $\{\eta_{i}\}_{i=0}^{N+1}$ to be such that $0 = \eta_0<\eta_1<\eta_2 \cdots <\eta_{N+1} = \cos^{-1} \big( 
\sqrt{2}(2\score-1)\big)$. Similarly define $\{\theta^{(i)}_j\}_{j=0}^{M(i)+1}$ be such that $0=\theta^{(i)}_{0}<\theta^{(i)}_{1}<\cdots<\theta^{(i)}_{M(i)+1}=\frac{\pi}{4} - \cos^{-1}\big(1/(4\score-2)\big)$ and $\{v_{k}^{(i,j)}\}_{k = 0}^{P(i,j)+1} $ be such that $0 = v^{(i,j)}_{0}<v^{(i,j)}_{1}< v^{(i,j)}_{2} \cdots < v^{(i,j)}_{P(i,j)+ 1} = \frac{\pi}{2}$.  Thus $\M{C} = \bigcup_{i,j,k}\M{C}_{i,j,k}$, where, to streamline the notation, we have used $\M{C}_{i,j,k}:=\M{C}_{\B{x}_{i, j , k}, \B{x}_{i +1, j +1, k+1}}$ with $\B{x}_{i,j,k} := (\eta_i, \theta^{(i)}_j, v^{(i,j)}_k)$.

From~\eqref{eq:uforms} there are two possible functional forms of $u_{\pm}$. Taking derivatives we find
\begin{eqnarray}
\partial_{\eta} u_{\pm} &=&  \pm 2\nonumber \\ 
\partial_{\theta} u_{\pm} &=& -2 \frac{\cos (2 \theta ) \cos (v)}{\sqrt{1-\sin ^2(2 \theta ) \cos ^2(v)}}  \in [-2 , 0] \label{eq:uderivs} \\ 
\partial_{v}u_{\pm} &=& \frac{\sin (2 \theta ) \sin (v)}{ \sqrt{1-\sin ^2(2 \theta ) \cos ^2(v)}} \in [0 ,1]. \nonumber
\end{eqnarray}

We return to the problem of deriving an upper bound on the functions $F_{1}$ and $F_{2}$. To do so, we first need bounds on the functions $\cos(u_{\pm} \pm v) \sin(2 \theta)$. Our bounds use Taylor's theorem, which we first state for convenience.
\begin{theorem}[Taylor]\label{thm:Taylor}
Let $\mathcal{D}\subseteq\mathbb{R}^n$ be compact and $f:\mathcal{D}\to\mathbb{R}$ be differentiable on $\mathcal{D}$, then for all ${\bf a,x}\in\mathcal{D}$ there exists ${\bf x'}\in\mathcal{D}$ such that
\[f({\bf x})=f({\bf a})+\nabla f\big|_{{\bf x'}}\cdot({\bf x}-{\bf a}).
\]
\end{theorem}
Since $\sin(2 \theta)$ is always positive and increasing in $\theta$ in our domain, we get the following result for any $(\eta , \theta , v) \in \M{C}_{i,j,k}$.
\begin{eqnarray}\label{eqn: Taylors theorem for H(A|XYE)}
 \cos(u(\B{x}) \pm v) \sin(2 \theta) \leq \begin{cases}\max_{\B{x} \in \M{C}_{i,j,k}} \Big (\cos(u(\B{x}) \pm v) \Big)  \sin(2 \theta_{i+1}) & \text{if} \max_{\B{x} \in \M{C}_{i,j,k}} \Big( \cos(u(\B{x}_{i,j,k}) \pm v) \Big) > 0 \\ 
    \max_{\B{x} \in \M{C}_{i,j,k}} \Big(\cos(u(\B{x}) \pm v) \Big)  \sin(2 \theta_{i})    &  \text{if} \max_{\B{x} \in \M{C}_{i,j,k}} \Big( \cos(u(\B{x}_{i,j,k}) \pm v) \Big) < 0.
 \end{cases}
\end{eqnarray}
Let $\B{x}=(\eta,\theta,v) \in \M{C}_{i,j,k}$ and, for brevity, write $g_{\pm,y}({\bf x})=u_{\pm}(\B{x})+(-1)^yv$ with $y\in\{0,1\}$.  Then, by Taylor's theorem (cf.\ Theorem~\ref{thm:Taylor}), there exists ${\bf x'}\in\M{C}_{i,j,k}$ such that
\begin{equation}\label{eq:cosg}
  \cos(g_{\pm,y}(\B{x}))=\cos(g_{\pm,y}(\B{x}_{i,j,k})) + \partial_\eta\cos(g_{\pm,y}(\B{x}))\big|_{{\bf x'}}(\eta-\eta_i)+
\partial_\theta\cos(g_{\pm,y}(\B{x}))\big|_{{\bf x'}}(\theta-\theta_j^{(i)}) + \partial_v\cos(g_{\pm,y}(\B{x}))\big|_{{\bf x'}}(v-v_k^{(i,j)}).
\end{equation}
We upper bound this by upper bounding each of the partial derivatives on $\M{C}_{i,j,k}$:
\begin{eqnarray*}
  \partial_\eta\cos(g_{\pm,y}(\B{x}))&=&-\sin(g_{\pm,y}(\B{x}))\partial_\eta u_{\pm}\\
  \partial_\theta\cos(g_{\pm,y}(\B{x}))&=&-\sin(g_{\pm,y}(\B{x}))\partial_\theta u_{\pm}\\
  \partial_v\cos(g_{\pm,y}(\B{x}))&=&-\sin(g_{\pm,y}(\B{x}))(\partial_v u_{\pm}+(-1)^y).
\end{eqnarray*}
We have bounded the derivatives of $u$ on $\M{C}_{i,j,k}$ in~\eqref{eq:uderivs}.

We now consider the different cases.  Firstly, suppose $\max_{\B{x}\in\M{C}_{i,j,k}}\left[-\sin(g_{\pm,y}(\B{x}))\right]\geq0$.

Consider the terms in~\eqref{eq:cosg}. Using the bounds in~\eqref{eq:uderivs}, we have
\begin{align*}
  \partial_\eta\cos(g_{+,y}(\B{x}))\big|_{{\bf x'}}(\eta-\eta_i)&\leq2\max_{\B{x}\in\M{C}_{i,j,k}}\left[-\sin(g_{+,y}(\B{x}))\right](\eta_{i+1}-\eta_i)\\                                    \partial_\theta\cos(g_{+,y}(\B{x}))\big|_{{\bf x'}}(\theta-\theta_j^{(i)})&\leq0\\
  \partial_v\cos(g_{+,y}(\B{x}))\big|_{{\bf x'}}(v-v_k^{(i,j)})&\leq\begin{cases}2\max_{\B{x}\in\M{C}_{i,j,k}}\left[-\sin(g_{+,0}(\B{x}))\right](v_{k+1}^{(i,j)}-v_k^{(i,j)})&y=0\\0&y=1\end{cases}.
\end{align*}
Similarly,
\begin{align*}
  \partial_\eta\cos(g_{-,y}(\B{x}))\big|_{{\bf x'}}(\eta-\eta_i)&\leq0\\                                    \partial_\theta\cos(g_{-,y}(\B{x}))\big|_{{\bf x'}}(\theta-\theta_j^{(i)})&\leq0\\
  \partial_v\cos(g_{-,y}(\B{x}))\big|_{{\bf x'}}(v-v_k^{(i,j)})&\leq\begin{cases}2\max_{\B{x}\in\M{C}_{i,j,k}}\left[-\sin(g_{-,0}(\B{x}))\right](v_{k+1}^{(i,j)}-v_k^{(i,j)})&y=0\\0&y=1\end{cases}.
\end{align*}
Combining all of these, and bounding the $-\sin(g_{\pm,y}(\B{x}))$ terms by $1$, we find
\begin{align}
\cos(g_{+,0}(\B{x}))&\leq\cos(g_{+,0}(\B{x}_{i,j,k}))+2(\eta_{i+1}-\eta_i)+2(v_{k+1}^{(i,j)}-v_k^{(i,j)})\nonumber\\
  \cos(g_{+,1}(\B{x}))&\leq\cos(g_{+,1}(\B{x}_{i,j,k}))+2(\eta_{i+1}-\eta_i)\label{eq:gplus1}\\
  \cos(g_{-,0}(\B{x}))&\leq\cos(g_{-,0}(\B{x}_{i,j,k}))+2(v_{k+1}^{(i,j)}-v_k^{(i,j)})\nonumber\\
  \cos(g_{-,1}(\B{x}))&\leq\cos(g_{-,1}(\B{x}_{i,j,k}))\nonumber
\end{align}

Secondly, in the case $\max_{\B{x}\in\M{C}_{i,j,k}}\left[-\sin(g_{\pm,y}(\B{x}))\right]\leq0$, we have
\begin{align*}
  \partial_\eta\cos(g_{+,y}(\B{x}))\big|_{{\bf x'}}(\eta-\eta_i)&\leq0\\                                    \partial_\theta\cos(g_{+,y}(\B{x}))\big|_{{\bf x'}}(\theta-\theta_j^{(i)})&\leq-2\min_{\B{x}\in\M{C}_{i,j,k}}\left[-\sin(g_{\pm,y}(\B{x}))\right](\theta^{(i)}_{j+1}-\theta^{(i)}_j)\\
  \partial_v\cos(g_{+,y}(\B{x}))\big|_{{\bf x'}}(v-v_k^{(i,j)})&\leq\begin{cases}0&y=0\\-\min_{\B{x}\in\M{C}_{i,j,k}}\left[-\sin(g_{+,0}(\B{x}))\right](v_{k+1}^{(i,j)}-v_k^{(i,j)})&y=1\end{cases}.
\end{align*}
and
\begin{align*}
  \partial_\eta\cos(g_{-,y}(\B{x}))\big|_{{\bf x'}}(\eta-\eta_i)&\leq-2\min_{\B{x}\in\M{C}_{i,j,k}}\left[-\sin(g_{\pm,y}(\B{x}))\right](\eta_{i+1}-\eta_i)\\                                    \partial_\theta\cos(g_{-,y}(\B{x}))\big|_{{\bf x'}}(\theta-\theta_j^{(i)})&\leq-2\min_{\B{x}\in\M{C}_{i,j,k}}\left[-\sin(g_{\pm,y}(\B{x}))\right](\theta^{(i)}_{j+1}-\theta^{(i)}_j)\\
  \partial_v\cos(g_{-,y}(\B{x}))\big|_{{\bf x'}}(v-v_k^{(i,j)})&\leq\begin{cases}0&y=0\\-\min_{\B{x}\in\M{C}_{i,j,k}}\left[-\sin(g_{+,0}(\B{x}))\right](v_{k+1}^{(i,j)}-v_k^{(i,j)})&y=1\end{cases}.
\end{align*}

Combining all of these, and bounding the $-\sin(g_{\pm,y}(\B{x}))$ terms by $-1$, we find
\begin{align}
\cos(g_{+,0}(\B{x}))&\leq\cos(g_{+,0}(\B{x}_{i,j,k}))+2(\theta^{(i)}_{j+1}-\theta^{(i)}_j)\nonumber\\
  \cos(g_{+,1}(\B{x}))&\leq\cos(g_{+,1}(\B{x}_{i,j,k}))+2(\theta^{(i)}_{j+1}-\theta^{(i)}_j)+(v_{k+1}^{(i,j)}-v_k^{(i,j)})\label{eq:gplus2}\\
  \cos(g_{-,0}(\B{x}))&\leq\cos(g_{-,0}(\B{x}_{i,j,k}))+2(\eta_{i+1}-\eta_i)+2(\theta^{(i)}_{j+1}-\theta^{(i)}_j)\nonumber\\
  \cos(g_{-,1}(\B{x}))&\leq\cos(g_{-,1}(\B{x}_{i,j,k}))+2(\eta_{i+1}-\eta_i)+2(\theta^{(i)}_{j+1}-\theta^{(i)}_j)+(v_{k+1}^{(i,j)}-v_k^{(i,j)}).\nonumber
\end{align}

Combining~\eqref{eq:gplus1} and~\eqref{eq:gplus2} we define
\begin{align*}
  \Delta_{+,0}&=\max(2(\theta^{(i)}_{j+1}-\theta^{(i)}_j),\ 2(\eta_{i+1}-\eta_i)+2(v_{k+1}^{(i,j)}-v_k^{(i,j)}))\\
\Delta_{+,1}&=\max(2(\eta_{i+1}-\eta_i),\ 2(\theta^{(i)}_{j+1}-\theta^{(i)}_j)+(v_{k+1}^{(i,j)}-v_k^{(i,j)})\\
\Delta_{-,0}&=\max(2(v_{k+1}^{(i,j)}-v_k^{(i,j)}),\ 2(\eta_{i+1}-\eta_i)+2(\theta^{(i)}_{j+1}-\theta^{(i)}_j))\\
  \Delta_{-,1}&=2(\eta_{i+1}-\eta_i)+2(\theta^{(i)}_{j+1}-\theta^{(i)}_j)+(v_{k+1}^{(i,j)}-v_k^{(i,j)}).
\end{align*}

We hence have the following bounds for any $\B{x}\in\M{C}_{i,j,k}$:
\begin{eqnarray}
\cos(g_{+,y}(\mathbf{x})) \sin(2 \theta) \leq \zeta^{i,j,k}_{+,y} := \begin{cases} 
\big(\cos(g_{+,y}(\B{x}_{i,j,k})) + \Delta_{+,y}  \big) \sin(2 \theta^{(i)}_{j}) 
& \text{ if } \big(\cos(g_{+,y}(\B{x}_{i,j,k})) + \Delta_{+,y}  \big) < 0 \\
\big(\cos(g_{+,y}(\B{x}_{i,j,k})) + \Delta_{+,y}\big) \sin(2 \theta^{(i)}_{j+1}) & 
 \text{ otherwise }
\end{cases} \nonumber 
\end{eqnarray}

\begin{eqnarray}
\cos(g_{-,y}(\mathbf{x})) \sin(2 \theta) \leq \zeta^{i,j,k}_{-,y} := \begin{cases}  
\big(\cos(g_{-,y}(\B{x}_{i,j,k})) + \Delta_{-,y}   \big) \sin(2 \theta^{(i)}_{j}) 
& \text{ if } \cos(g_{-,y}(\B{x}_{i,j,k})) + \Delta_{-,y} < 0 \\
\big(\cos(g_{-,y}(\B{x}_{i,j,k})) + \Delta_{-,y}   \big) \sin(2 \theta^{(i)}_{j+1}) & 
 \text{ otherwise }
\end{cases}.\nonumber  
\end{eqnarray}

With this established we return to the optimization problem~\eqref{eqn: Optimization problem on grids}. We define the objective function $Q(\eta,\theta,v):=F_1(\eta,\theta,v)+F_2(\eta,\theta,v)+K(R(\eta),\theta)$, which we want to optimize over $\M{D}_{\score}$ and $u\in\{u_+,u_-\}$.

\begin{lemma}
Let $\M{P} = \bigcup_{i,j,k} \M{C}_{i,j,k}$ be a partition of $\M{C}$ as specified above. Define $g_{i,j,k}$ and $h_{i,j,k}$ as follows
\begin{eqnarray}
g_{i,j,k} &:=&  \frac{1}{2}H_{\bin} \left(\frac{1}{2} + \frac{R(\eta_{i+1})}{2} \sqrt{1 + \zeta_{+,0}^{i,j,k}}\right) + \frac{1}{2} H_{\bin} \left(\frac{1}{2} + \frac{R(\eta_{i+1})}{2} \sqrt{1 + \zeta^{i,j,k}_{+,1} }\right) + K\left( R(\eta_{i}) , \theta^{(i)}_{j} \right) \nonumber\\ 
    h_{i,j,k} &:=&  \frac{1}{2}H_{\bin} \left(\frac{1}{2} + \frac{R(\eta_{i+1})}{2} \sqrt{1 + \zeta_{-,0}^{i,j,k} }\right) + \frac{1}{2} H_{\bin} \left(\frac{1}{2} + \frac{R(\eta_{i+1})}{2} \sqrt{1 + \zeta^{i,j,k}_{-,1} }\right)  + K\left(R(\eta_i) , \theta^{(i)}_{j} \right). \nonumber
\end{eqnarray}
Let $M \in \mathbb{R}$ be any upper bound on $Q$, i.e., $M \geq \max_{\B{x} \in \M{D}} Q(\B{x})$. Then
\begin{equation}\label{eq:def_f}
    Q(\B{x})\geq f_{i,j,k}:=\begin{cases}\min\{g_{i,j,k}, h_{i,j,k}\}&\text{ if }\B{x}\in\M{C}_{i,j,k}\text{ such that }\M{C}_{i,j,k} \cap \M{D} \ne \emptyset\\
 M & \text{ otherwise.}
\end{cases}
\end{equation}
\end{lemma}
\begin{proof}
From Lemmas~\ref{lem: monotonicity of K(R)} and~\ref{lem: monotonicity of K(th)} we know that $\partial_{R} K>0$ and $\partial_{\theta} K>0$. In addition, $\partial_{\eta} K(R(\eta), \theta)  = \frac{\sqrt{2}(2 \score - 1) \sin(\eta )}{ \cos^2(\eta)} \partial_{R} K(R , \theta)$. Positivity of $\partial_{\eta} K$ and $\partial_{\theta} K$, implies  $K( R(\eta) , \theta) \geq K( R(\eta_{i}) , \theta^{(i)}_{j})$ within $\M{C}_{i,j,k}$. Furthermore, $H_{\bin}(\frac{1}{2} + \frac{x}{2})$ is decreasing for $x \geq 0$. Since $R(\eta) \sqrt{1 + \cos(v \pm u) \sin(2 \theta)} >0$,
\begin{align*}
H_{\bin}\left(\frac{1}{2} + \frac{R(\eta)}{2} \sqrt{1 + \cos(u_{+} \pm v) \sin(2 \theta)}\right) &\geq H_{\bin}\left(\frac{1}{2} + \frac{R(\eta_{i+1})}{2} \sqrt{1 + \cos(u_{+} \pm v) \sin(2 \theta)}\right)\\
&=H_{\bin}\left(\frac{1}{2} + \frac{R(\eta_{i+1})}{2} \sqrt{1 + \cos(g_{+,(1\mp1)/2}) \sin(2 \theta)}\right)\\
&\geq H_{\bin}\left(\frac{1}{2} + \frac{R(\eta_{i+1})}{2} \sqrt{1 + \zeta^{i,j,k}_{+,(1\mp1)/2}}\right).
\end{align*}

Similarly,
$$H_{\bin}\left(\frac{1}{2} + \frac{R(\eta_{i+1})}{2} \sqrt{1 + \cos(u_{-} \pm v) \sin(2 \theta)}\right) \geq H_{\bin}\left(\frac{1}{2} + \frac{R(\eta_{i+1})}{2} \sqrt{1 + \zeta^{i,j,k}_{-,(1\mp1)/2}}\right),$$
which establishes the claim.
\end{proof}
Combining the results in this section we obtain the following corollary. 
\begin{corollary}
Let $\score \in (\frac{3}{4}, \frac{1}{2} + \frac{1}{2 \sqrt{2}}]$ be fixed. Let $\M{D}_{\score}$ be defined as in \eqref{eqn: Domain for one sided rate} and $\M{P} = \bigcup_{i,j,k} \M{C}_{i,j,k}$ be any partition of the cuboid $\M{C} = [0 , \cos^{-1}\big(\sqrt{2}(2 \score - 1)\big)]\times[0 , \frac{\pi}{4} - \cos^{-1}\big((4 \score -2)^{-1}\big)]\times[0 , \frac{\pi}{2}]$. Then 
\begin{eqnarray}
G_{A|XYE}(\score) \geq \min_{i,j,k} f_{i,j,k}.
\end{eqnarray}
where $f_{i,j,k}$ are defined in~\eqref{eq:def_f}.
\end{corollary}
This means that for fixed $\score$ we can lower bound the randomness by evaluating $f_{i,j,k}$ at all grid points in the relevant cuboid and taking the minimum.  This is how our numerical algorithm works (note that the lower bound gets tighter as the number of grid points is increased).

\subsection{Lower bounding $F_{A|XYE}$}
In the previous section, we derived a technique to lower bound the function $G_{A|XYE}(\score)$ for a fixed value of the score, $\score$. In Appendix~\ref{app: convex combinations of qubit strategies} we show that the asymptotic rate $F_{A|XYE}$ can be computed by taking the convex lower bound on $G_{A|XYE}$. In this section, we construct a lower bound on the function $F_{A|XYE}$ using a lower bound on $G_{A|XYE}$. 

We start with a general lemma.
\begin{lemma}\label{lem:FLB}
Let $a$ and $b$ be real numbers, $a<b$ and $\Tilde{G}:[a,b]\to\mathbb{R}$ be a lower bound on $G:[a,b]\to\mathbb{R}$. Let $\Tilde{F}[a,b]\to\mathbb{R}$ and $F[a,b]\to\mathbb{R}$ be convex lower bounds on $\Tilde{G}$ and $G$ respectively. Then $\Tilde{F}$ is a lower bound on $F$.
\end{lemma}
\begin{proof}
Let $M_{\score_{0}}$ be the set of probability measures on the interval $[a,b]$ satisfying $\int d \mu(\score) \score = \score_0$.
\begin{eqnarray}
F(\score_0) = \inf_{\mu \in M_{\score_0}} \int d \mu(\score) G(\score)  
\end{eqnarray}
Since $G(\score) \geq \Tilde{G}(\score)$ for every value of $\score\in[a,b]$, for every measure $\mu \in M_{\score}$ we must have that $\int d \mu(\score) G(\score) \geq \int d \mu(\score) \Tilde{G}(\score)$. Thus
\begin{equation*}
F(\score_0) := \inf_{\mu \in \score_{0}}\int d \mu(\score) G(\score)  \geq \inf_{\mu \in \score_0} \int d \mu(\score) \Tilde{G}(\score) \geq \Tilde{F}(\score_0).\qedhere
\end{equation*}
\end{proof}

Since we can only compute our lower bound $G^{\M{P}}_{A|XYE}$ on $G_{A|XYE}$ for a finite set of values of $\score$, to form a lower bound that holds for all values of $\score$, we construct a function $\tilde{G}_{A|XYE}$ as follows.  Let $\{\score_i\}_{i=1}^N$ be an ordered set of values in $[\frac{3}{4},\frac{1}{2}+\frac{1}{2 \sqrt{2}}]$ with $\score_1=3/4$ at which we have computed $G^{\M{P}}_{A|XYE}$. We define $\tilde{G}_{A|XYE}(\score)$ to be equal to $G^{\M{P}}_{A|XYE}(\score_i)$ for $\score\in[\score_i,\score_{i+1})$, and equal to $G^{\M{P}}_{A|XYE}(\score_N)$ for $\score\geq\score_N$.  Because $G_{A|XYE}$ is monotonically increasing in $\score$ (see Lemma~\ref{lemm: Monotonicity of one sided rates}), it follows that for $\score\in[\score_i,\score_{i+1})$, $G_{A|XYE}(\score)\geq G_{A|XYE}(\score_i) \geq G^{\M{P}}_{A|XYE}(\score_i)=\Tilde{G}_{A|XYE}(\score)$.

A lower bound $\Tilde{F}_{A|XYE}$ of $F_{A|XYE}$ can then be formed by taking the convex lower bound of $\Tilde{G}^{\M{P}}_{A|XYE}$ (cf.\ Lemma~\ref{lem:FLB}).

\subsection{H(A|E)}
In this case Lemma~\ref{lem:ent_simp} gives
\begin{align*}
H(A|E)&=H(A)+\sum_ap_A(a)H\left(\sum_{bxy}p_{BXY|a}(b,x,y)\tau_E^{abxy}\right)-H\left(\sum_{abxy}p_{AB}(a,b)p_{XY|ab}(x,y)\tau_E^{abxy}\right)\\
      &=1+\frac{1}{2}\sum_aH\left(\sum_{bxy}2p_{ABXY}(a,b,x,y)\tau_E^{abxy}\right)-H(E)\\
      &=1+\frac{1}{2}\sum_aH\left(\sum_{bxy}2p_{XY}(x,y)P_{AB|xy}(a,b)\tau_E^{abxy}\right)-H(E)\\
      &=1+\frac{1}{2}\sum_aH\left(\sum_x2p_X(x)\tr_{A'}\left(\left(\proj{\phi^A_{a|x}}\ot\id_E\right)\rho_{A'E}\right)\right)-H(E)\,.
\end{align*}
Like in the case $H(AB|E)$ the middle term cannot be removed and this term is not independent of $\delta$.

\subsection{H(AB|X=0,Y=0,E)}\label{app:AB00E}
For $H(AB|X=0,Y=0,E)$ we are interested in the state
$$\tau'_{ABE}=\sum_{ab}p_{AB|00}(a,b)\proj{a}\ot\proj{b}\ot\tau_E^{ab00}$$
since $H(AB|X=0,Y=0,E)_\tau=H(AB|E)_{\tau'}$. Note that, as above, $H(E)_{\tau'}=H(E)_\rho$. Using Lemma~\ref{lem:ent_simp} we have
\begin{align*}
H(AB|E)_{\tau'}&=H(AB)_{\tau'}+\sum_{ab}p_{AB}(a,b)H(\tau_E^{ab00})-H\left(\sum_{ab}p_{AB}(a,b)\tau_E^{ab00}\right)\\
  &=H(AB)_{\tau'}+\sum_{ab}p_{AB}(a,b)H(\tau_E^{ab00})-H(E)_{\tau'}\,.
\end{align*}
However, since $\tau_E^{ab00}$ is pure for each $a,b$, $H(\tau_E^{ab00})=0$ and we find
\begin{align*}
  H(AB|E)_{\tau'}&=H(AB)_{\tau'}-H(E)_{\rho}\\
&=H(\{\epsilon_{00},\epsilon_{00},1/2-\epsilon_{00},1/2-\epsilon_{00}\})-H(\{\lambda_0,\lambda_1,\lambda_2,\lambda_3\})\\
  &=1+H_\bin(2\epsilon_{00})-H(\{\lambda_0,\lambda_1,\lambda_2,\lambda_3\})\,.
\end{align*}
Lemma~\ref{lem:delta_opt} shows that $\max_\delta H(\{\lambda_0,\lambda_1,\lambda_2,\lambda_3\})$ is achieved for $\delta=\delta^*=\frac{R^2\cos(2\theta)}{4}$. Since the score is independent of $\delta$ we can take the state to satisfy $\delta=\delta^*$ and remove $\delta$ from the optimization.

\subsubsection{Reparameterizing the optimisation problem}
We introduce some notation for convenience. Let $\mathbf{x} = (R , \theta , \alpha_0 ,\alpha_1 , \beta_0 , \beta_1)$ and define  
\begin{eqnarray}
\hat{\epsilon}_{00}(\mathbf{x})  &:=&  \cos(\theta) \cos(2 \alpha_0 - 2 \beta_0) + \sin(\theta) \cos(2 \alpha_0 + 2 \beta_0) \\ 
\hat{\epsilon}_{10}(\mathbf{x})  &:=&  \cos(\theta) \cos(2 \alpha_1 - 2 \beta_0) + \sin(\theta) \cos(2 \alpha_1 + 2 \beta_0) \\ 
\hat{\epsilon}_{01}(\mathbf{x})  &:=&  \cos(\theta) \cos(2 \alpha_0 - 2 \beta_1) + \sin(\theta) \cos(2 \alpha_0 + 2 \beta_1) \\ 
\hat{\epsilon}_{11}(\mathbf{x})  &:=&  -\cos(\theta) \cos(2 \alpha_1 - 2 \beta_1) - \sin(\theta) \cos(2 \alpha_1 + 2 \beta_1) \\ 
K(\mathbf{x}) &:=& K(R, \theta)\,,
\end{eqnarray}
where $K(R,\theta)$ is given in Corollary~\ref{cor:onesidelower}.
In this notation, the equation for the constraint is 
\begin{eqnarray}\label{eq:epsconstr}
\sum_{ij} \hat{\epsilon}_{i,j} &=& \frac{4(2\score-1)}{R}\,,
\end{eqnarray}
and hence the optimization problem is
\begin{equation}\label{Two-sided optimization}
\begin{aligned}
G_{AB|X=0,Y=0,E}(\score) = \min_{\mathbf{x}\in\cD_\score} \quad & \Big(H_{\bin}\left(\frac{1}{2} + \frac{R}{2} \hat{\epsilon}_{00}(\mathbf{x})\right)  + K(\mathbf{x}) \Big)\\
\textrm{s.t.} \quad & \sum_{ij} \hat{\epsilon}_{ij}(\mathbf{x}) = \frac{4(2\score-1)}{R}\,,
\end{aligned}
\end{equation}
where $\cD_\score=\{R\in[\sqrt{2}(2\score-1),1],\theta\in[0,\pi/4-\cos^{-1}(1/(R\sqrt{2}))],(\alpha_0,\alpha_1,\beta_0,\beta_1)\in\mathbb{R}^4\}$ (see Lemma~\ref{lem:Ag00Escore} for the justification of the range of $R$).

For brevity we use $P(\mathbf{x})$ for the objective function. We call $\mathbf{x}\in\cD_\score$ a solution to the optimization problem~\eqref{Two-sided optimization} if $G_{AB|X=0,Y=0,E}(\score)=P(\mathbf{x})$ and $\mathbf{x}$ satisfies the constraint. For reasons that shall be clear later, we now define the following functions on the extended domain
\begin{eqnarray}
\hat{H}_{\bin}(x) = \begin{cases}
H_{\bin}(x)               & \text{ if } x \in [\frac{1}{2} , 1 ]   \\              1    & \text{ otherwise }
\end{cases}
\end{eqnarray}
and 
\begin{eqnarray}
\hat{K}(R , \theta) = \begin{cases}
K(R , \theta)            & \text{ if } \sqrt{2} (2 \score -1 ) \leq R \leq 1 \text{ and } 0 \leq \theta \leq \frac{\pi}{4} - \cos^{-1} \left(\frac{1}{\sqrt{2}R}\right)    \\               
1    & \text{ otherwise }  
\end{cases}
\end{eqnarray}
Here $\hat{K}(R , \theta)$ and $\hat{H}_{\bin}$ both take the value $1$ when the functions $K(R, \theta)$ and $H_{\bin}(x)$ are outside the stated range. These values are chosen such that upon extension of the domain, the resulting optimization problem still has the same minimum\footnote{That $H_{\bin}(x) \leq 1$ and $K(R, \theta) \leq 1$ whenever defined justifies the choice made for defining $\hat{H}_{\bin}(x)$ and $\hat{K}(R,\theta)$.}.
\begin{lemma}\label{Claim: positive espilon}
Let $\mathbf{X}_\score$ be the set of solutions of~\eqref{Two-sided optimization} for some $\score\in(\frac{3}{4}, \frac{1}{2} + \frac{1}{2 \sqrt{2}}]$. There exists $\mathbf{x} \in \mathbf{X}_\score$ such that 
 $\hat{\epsilon}_{00}(\mathbf{x}) > 0$ and $\displaystyle\hat
 {\epsilon}_{00}(\mathbf{x}) = \max_{i,j} |\epsilon_{ij}(\mathbf{x})|$.
\end{lemma}
\begin{proof}
We first prove that we can choose
\begin{eqnarray}
|\hat{\epsilon}_{00}(\B{x})| = \max_{i,j} |\hat{\epsilon}_{ij}(\B{x})|.
\end{eqnarray}
From the symmetry of the binary entropy, $H_{\bin}(\frac{1}{2} + \frac{y_1}{2}) < H_{\bin}(\frac{1}{2} + \frac{y_2}{2})$ for $|y_1| > |y_2|$. Now consider the following cases
\begin{itemize}
    \item Suppose $|\hat{\epsilon}_{00}(\B{x})| < |\hat{\epsilon}_{10}(\B{x})|$: Perform the transformation $\alpha_0  \leftrightarrow \alpha_1$ and $\beta_1 \rightarrow \beta_1 + \frac{\pi}{2}$. Under this transformation $\hat{\epsilon}_{00}(\B{x}) \leftrightarrow \hat{\epsilon}_{01}(\B{x})$ and $\hat{\epsilon}_{10}(\B{x}) \leftrightarrow \hat{\epsilon}_{11}(\B{x})$. The CHSH score is hence preserved. This transformation also decreases the objective function, so ${\bf x}$ cannot have been an solution to~\eqref{Two-sided optimization} prior to the transformation.
    \item Suppose $|\hat{\epsilon}_{00}(\mathbf{x})| < |\hat{\epsilon}_{01}(\mathbf{x})|$: Perform the transformation $\beta_0  \leftrightarrow \beta_1$ and $\alpha_1 \rightarrow \alpha_1 + \frac{\pi}{2}$. Under this transformation $\hat{\epsilon}_{00}(\B{x}) \leftrightarrow \hat{\epsilon}_{10}(\B{x})$ and $\hat{\epsilon}_{01}(\B{x}) \leftrightarrow \hat{\epsilon}_{11}(\B{x})$. Again, this preserves the CHSH score while reducing the objective function.
    \item Suppose $|\hat{\epsilon}_{00}(\mathbf{x})| < |\hat{\epsilon}_{11}(\mathbf{x})|$: Perform the transformation $\alpha_0 \rightarrow \alpha_1 + \frac{\pi}{2}$ , $\alpha_1 \rightarrow \alpha_0$, $\beta_0 \rightarrow \beta_1$ and $\beta_1 \rightarrow \beta_0 + \frac{\pi}{2}$. Under this transformation $\hat{\epsilon}_{00}(\B{x}) \leftrightarrow \hat{\epsilon}_{11}(\B{x})$ and $\hat{\epsilon}_{01}(\B{x}) \leftrightarrow \hat{\epsilon}_{10}(\B{x})$. Again, this preserves the CHSH score while reducing the objective function.
\end{itemize}
Finally, we can show that $\hat{\epsilon}_{00}(\B{x}) > 0$ by observing that for $\score\in(\frac{3}{4}, \frac{1}{2} + \frac{1}{2 \sqrt{2}}]$ we have
\begin{eqnarray}
   R \sum_{i,j} \hat{\epsilon}_{ij}(\B{x}) &=& 4 (2 \score - 1) > 2.\label{eq:bound1}
   \end{eqnarray}
   In addition, for all $i,j$,
   \begin{eqnarray}
   R \hat{\epsilon}_{ij}(\B{x}) &\leq&R(\cos(\theta)+\sin(\theta))\leq1,\label{eq:bound2}
\end{eqnarray}
where the last inequality follows from~\eqref{param1} and~\eqref{param2} whose sum can be at most $1$.

Now suppose that $|\hat{\epsilon}_{00}(\B{x})|= \displaystyle\max_{i,j} |\hat{\epsilon}_{ij}|$ and $\hat{\epsilon}_{00}(\B{x})<0$. It follows that
\begin{eqnarray}
    R\sum_{i,j} \hat{\epsilon}_{ij} &=& R\big( \hat{\epsilon}_{00} + \hat{\epsilon}_{01} \big) + R\big( \hat{\epsilon}_{10} + \hat{\epsilon}_{11}\big) \nonumber \\ 
    &\leq& R\big( \hat{\epsilon}_{00} + \hat{\epsilon}_{01} \big) + 2 \nonumber \\
    &\leq& 2, \nonumber
\end{eqnarray}
where the first inequality uses~\eqref{eq:bound2}. This is in contradiction with~\eqref{eq:bound1}.
\end{proof}
\begin{lemma}\label{lemm: H(A|X=0,Y=0,E) on grids}
Let $\hat{P}$ be the objective function with extended domain, i.e., $\hat{P}(\mathbf{x}) := \hat{H}_\bin\left( \frac{1}{2} + \frac{R \epsilon_{00}(\B{x})}{2}\right) + \hat{K}(\mathbf{x})$, $\score\in(\frac{3}{4}, \frac{1}{2} + \frac{1}{2 \sqrt{2}}]$, and let $\M{X}$ be a set such that $\M{D}_\score\subseteq\M{X}\subseteq\mathbb{R}^6$. Then,
\begin{eqnarray}\label{Revised two sided}
\begin{aligned}
G_{AB|X=0,Y=0 ,E}(\score ) = \min_{\mathbf{x}\in\M{X}} \quad & \hat{P}(\mathbf{x}) \\
\mathrm{s.t.} \quad & \sum_{ij} \hat{\epsilon}_{ij}(\mathbf{x}) = \frac{4(2 \score - 1 )}{R},
\end{aligned}   
\end{eqnarray}
i.e., optimizing over $\hat{P}$ on an extended domain $\M{X}$ gives the same solution as the original optimization~\eqref{Two-sided optimization}.
Furthermore $\exists \mathbf{x} \in \M{D}_\score$ that is a solution to both optimization problems.
    \end{lemma}
\begin{proof}
Let $\B{x'}\in\M{D}_\score$ achieve the optimal value of $P$ and have $\hat{\epsilon}_{00}(\B{x'})>0$. [From Lemma~\ref{Claim: positive espilon} such an $\B{x'}$ exists.]  Since $\hat{P}(\B{x})=P(\B{x})\leq2$ for all $\B{x}\in\M{D}_\score$, and $\hat{P}(\B{x})=2$ for $\B{x}\in \mathbb{R}^{6}\setminus \M{D}_\score$, $\B{x'}$ must also achieve the optimal value of $\hat{P}$, where it takes the same value.
\end{proof}

\subsubsection{Some simplifications}
\begin{lemma}\label{lem:signs}
Let $\mathbf{X}_\score$ the set of solutions to the optimization problem~\eqref{Two-sided optimization} for some $\score\in(\frac{3}{4}, \frac{1}{2} + \frac{1}{2 \sqrt{2}}]$. There exists $\mathbf{x} = (R ,\theta , \alpha_0 , \alpha_1 , \beta_0 , \beta_1) \in \mathbf{X}_\score$ such that the following hold
\begin{itemize}
    \item $\sin(\beta_0+\beta_1) \geq 0$
    \item $\sin(\beta_0-\beta_1) \leq 0$
\end{itemize}
\end{lemma}
\begin{proof}
The expression for the CHSH score satisfies (cf.~\eqref{eq:scoresimp})
\begin{align*}
\sqrt{2} (2 \score - 1)=&R\cos(\beta_0 - \beta_1)\left[\sin(2\alpha_0)\sin(\beta_0 + \beta_1)\cos(\frac{\pi}{4}+\theta)+\cos(2\alpha_0)\cos(\beta_0 + \beta_1)\sin(\frac{\pi}{4}+\theta)\right]\\
&+R\sin(\beta_0 - \beta_1)\left[\sin(2\alpha_1)\cos(\beta_0+ \beta_1)\cos(\frac{\pi}{4}+\theta)-\cos(2\alpha_1)\sin(\beta_0 + \beta_1)\sin(\frac{\pi}{4}+\theta)\right].
  \end{align*}
Let $\alpha_0,\alpha_1, \beta_0, \beta_1 $ be optimal parameters. Consider performing the following steps sequentially.
\begin{enumerate}
    \item If $\sin(\beta_0 + \beta_1 ) <  0$, then perform the transformations $\beta_i \rightarrow -\beta_i$ and $\alpha_i \rightarrow -\alpha_i$. We get $\sin(\beta_0 + \beta_1) \geq 0$. 
    \item If $\sin(\beta_0 - \beta_1) > 0$ then perform the transformations $\beta_0 \rightarrow \beta_0 + \frac{\pi}{2}$, $\beta_1 \rightarrow \beta_1 - \frac{\pi}{2}$,   $\alpha_i \rightarrow \alpha_i + \frac{\pi}{2}$. This step does not affect $\sin(\beta_0 + \beta_1)$. Thus we ensure that  $\sin(\beta_0 - \beta_1) \leq 0$ and $\sin(\beta_0 + \beta_1) \geq 0$. 
    \end{enumerate}
In each step, the values of $\epsilon_{ij}$ for all $i,j$ remain the same, hence the CHSH score and the objective function remains invariant throughout. Thus, the transformations maintain optimal parameters.
\end{proof}

\subsubsection{Reduction in parameters}\label{app: two sided proof}
To rewrite the optimization in a way that removes the constraint we introduce the following functions 
\begin{eqnarray}
\hat{\alpha}_{0}(\lambda , v ,\theta) &:=&    -2\tan^{-1} \left( \frac{1}{\tan(\lambda) \tan(\frac{\pi}{4} + \theta)}\right) +  \tan^{-1} \left( \frac{1}{\tan(v) \tan(\frac{\pi}{4} + \theta)}\right) \\ 
\tilde{\epsilon}(\lambda, v, \theta) &:=& \cos(\theta) \cos(\hat{\alpha}_0 - 2 v + \lambda  ) + \sin(\theta) \cos(\hat{\alpha}_0 + 2 v - \lambda) \label{two sided function 1} \\ 
\hat{R}(\lambda , v , \theta) &:=& \frac {\sqrt{2} (2 \score - 1) }{\cos(\lambda\!-\!v)\!\left[\sin(\hat{\alpha}_0)\sin(v)\cos(\frac{\pi}{4}\!+\!\theta)\!+\!\cos(\hat{\alpha}_0)\cos(v)\sin(\frac{\pi}{4}\!+\!\theta)\right]\!+\!\frac{\sin(\lambda-v)}{\sqrt{2}} \sqrt{1\!-\!\cos(2 v)\sin(2 \theta)} }. \label{two sided function 2}
\end{eqnarray}

We also state the following small lemma for convenience.
\begin{lemma}\label{lem:arcsin}
Let $a,b\in\mathbb{R}$ with $a\neq0$. The values of $\gamma\in\mathbb{R}$ that form extrema of $a\cos(\gamma)+b\sin(\gamma)$ are
\begin{equation}
\gamma=\tan^{-1}(b/a)+n\pi
\end{equation}
for any $n\in\mathbb{Z}$.
If $a>0$ the maxima occur when $n$ is even and the minima when $n$ is odd, and vice-versa if $a<0$.
\end{lemma}
\begin{proof}
The problem is equivalent to maximizing
\begin{equation*}
\frac{a}{\sqrt{a^2+b^2}}\cos(\gamma)+\frac{b}{\sqrt{a^2+b^2}}\sin(\gamma).
\end{equation*}
Let $\phi$ satisfy $\cos(\phi)=\left(\frac{a}{\sqrt{a^2+b^2}}\right)$ and $\sin(\phi)=\left(\frac{b}{\sqrt{a^2+b^2}}\right)$. Thus, the expression is equivalent to $\cos(\gamma-\phi)$ which has maxima for $\gamma=\phi+2n\pi$ and minima for $\gamma=\phi+\pi+2n\pi$ for $n\in\mathbb{Z}$.

If $a>0$ then this gives maxima for $\gamma=\tan^{-1}(b/a)+2n\pi$ and minima for $\gamma=\tan^{-1}(b/a)+(2n+1)\pi$.

Alternatively, if $a<0$ then this gives maxima for $\gamma=\tan^{-1}(b/a)+(2n+1)\pi$ and minima for $\gamma=\tan^{-1}(b/a)+2n\pi$.
\end{proof}

\begin{lemma} \label{lem: twosided lower bound}
Let $\score\in(\frac{3}{4}, \frac{1}{2} + \frac{1}{2 \sqrt{2}}]$ and $\M{D}_\score' = \left\{(\lambda , v , \theta) \in \mathbb{R}^{3} : \lambda \in [0 , \pi] , v \in [0 , \pi] , \theta \in [0 , \frac{\pi}{4} - \cos^{-1} \left(1/(4 \score -2) \right) ] \right\}$, then
\begin{eqnarray}\label{eq:newopt}
G_{AB|X=0, Y=0, E}(\score) =  \inf_{\M{D}_\score'} \left( \hat{H}_{\bin}\left(\frac{1}{2} +\frac{\hat{R}(\lambda,v,\theta)\tilde{\epsilon}(\lambda,v,\theta)}{2}\right) + \hat{K}(\hat{R} , \theta) \right) 
\end{eqnarray}
\end{lemma}
\begin{proof}
Start from the form of $G$ in Lemma~\ref{lemm: H(A|X=0,Y=0,E) on grids}. The objective function $\hat{P}$ is independent of the parameters $\alpha_1$ and $\beta_1$, and, as shown in Lemma~\ref{lemm: H(A|X=0,Y=0,E) on grids}, the optimum is achieved for some ${\bf x}\in\cD_\score$. Because the function $G_{AB|X=0,Y=0,E}(\score)$ is increasing in $\score$ (see Lemma~\ref{lem:monot:w}), the optimal values of the parameters $\alpha_1$ and $\beta_1$ must maximize the CHSH score. Recall that the score can be related to $\alpha_i$ and $\beta_i$ by (cf.~\eqref{eq:scoresimp})
\begin{align}
\sqrt{2} (2 \score - 1) =&R\cos(\beta_0 - \beta_1)\left[\sin(2\alpha_0)\sin(\beta_0 + \beta_1)\cos(\frac{\pi}{4}+\theta)+\cos(2\alpha_0)\cos(\beta_0 + \beta_0)\sin(\frac{\pi}{4}+\theta)\right]\nonumber\\
&+R\sin(\beta_0 - \beta_1)\left[\sin(2\alpha_1)\cos(\beta_0+ \beta_1)\cos(\frac{\pi}{4}+\theta)-\cos(2\alpha_1)\sin(\beta_0 + \beta_1)\sin(\frac{\pi}{4}+\theta)\right],\label{eq:secn}
  \end{align}
Consider maximizing this over $\alpha_1$. From Lemma~\ref{lem:signs} we can assume $\sin(\beta_0-\beta_1)\leq0$, so we want to minimize the second term in square brackets in~\eqref{eq:secn}. This has the form of the expression in Lemma~\ref{lem:arcsin}.  Since the sine and cosine of $\pi/4+\theta$ are both positive, and from Lemma~\ref{lem:signs} we can assume $\sin(\beta_0+\beta_1)\geq0$, the minima of the square bracket (and hence maxima overall) occur for 
\begin{equation}\label{eq:alpha1}
2\alpha_1=-\tan^{-1}\left(\cot(\beta_0+\beta_1)\cot(\frac{\pi}{4}+\theta)\right)+2n\pi.
\end{equation}

The CHSH score is symmetric in the parameters for Alice and Bob, so we can re-write it as 
\begin{align*}
\sqrt{2} (2 \score - 1) =&  R\cos(\alpha_0 - \alpha_1)\left[\sin(2\beta_0)\sin(\alpha_0 + \alpha_1)\cos(\frac{\pi}{4}+\theta)+\cos(2\beta_0)\cos(\alpha_0 + \alpha_0)\sin(\frac{\pi}{4}+\theta)\right]\\
&+R\sin(\alpha_0 - \alpha_1)\left[\sin(2\beta_1)\cos(\alpha_0+ \alpha_1)\cos(\frac{\pi}{4}+\theta)-\cos(2\beta_1)\sin(\alpha_0 + \alpha_1)\sin(\frac{\pi}{4}+\theta)\right].
  \end{align*}
If we now maximize over $\beta_1$, from Lemma~\ref{lem:signs} the solutions either satisfy 
\begin{align*}
2\beta_1&=-\tan^{-1}\left(\cot(\alpha_0+\alpha_1)\cot(\frac{\pi}{4}+\theta)\right)+2n\pi\quad\text{or}\\
2\beta_1&=-\tan^{-1}\left(\cot(\alpha_0+\alpha_1)\cot(\frac{\pi}{4}+\theta)\right)+(2n+1)\pi
\end{align*}
for $n\in\mathbb{Z}$. (Which one holds depends on the signs of $\sin(\alpha_0-\alpha_1)$ and $\sin(\alpha_0+\alpha_1)$.)  In both cases, $\tan(2\beta_1)=\cot(\alpha_0+\alpha_1)\cot(\frac{\pi}{4}+\theta)$.

By symmetry (and because we can take $\sin(\alpha_0-\alpha_1)\leq0$ and $\sin(\alpha_0+\alpha_1)\geq0$ from Lemma~\ref{lem:signs}) the maxima of this over $\beta_1$ occur for
\begin{equation}
2\beta_1=-\tan^{-1}\left(\cot(\alpha_0+\alpha_1)\cot(\frac{\pi}{4}+\theta)\right)+2n\pi.
\end{equation}
Rearranging gives
\begin{equation}
\tan(\alpha_0+\alpha_1)=-\cot(2\beta_1)\cot(\frac{\pi}{4}+\theta)\,,
\end{equation}
and hence
\begin{eqnarray*}
\alpha_0 &=& -\alpha_1 - \tan^{-1}\left(\cot(2\beta_1) \cot(\frac{\pi}{4}+\theta) \right) + n\pi
\end{eqnarray*}
for $n \in \mathbb{Z}$. Using~\eqref{eq:alpha1} we find
\begin{equation}\label{eq:alpha0}
2\alpha_0=\tan^{-1}\left(\cot(\beta_0+\beta_1)\cot(\frac{\pi}{4}+\theta)\right) - 2\tan^{-1}\left(\cot(2\beta_1) \cot(\frac{\pi}{4}+\theta) \right) + 2n\pi.
\end{equation}

The proof proceeds as follows. We use~\eqref{eq:alpha1} to eliminate $\alpha_1$ from the constraint, noting that the value of $n$ in~\eqref{eq:alpha1} does not change the value so we can take $n=0$. We then use~\eqref{eq:alpha0} to reparameterize the objective function in terms of $\beta_1$ instead of $\alpha_0$ (again the value of $n$ in~\eqref{eq:alpha0} makes no difference and we take $n=0$). The parameters that remain are hence $\beta_0$, $\beta_1$, $R$ and $\theta$. We then reparameterize using $v=\beta_0+\beta_1$ and $\lambda=2\beta_1$, so that both the constraint and objective function are written in terms of $v$, $\lambda$, $R$ and $\theta$. We then use the constraint to write $R$ in terms of the other parameters, reducing the objective function to an unconstrained optimization over $v$, $\lambda$ and $\theta$. 

At this stage $v$ and $\lambda$ range over all reals, which can readily be restricted to $[0,2\pi]$. In fact, we can restrict both to $[0,\pi]$ by noting that Lemma~\ref{lem:signs} shows that it suffices to take $\sin(v)=\sin(\beta_0+\beta_1)\geq0$ hence $v\in[0,\pi]$. We then consider the transformation $(\lambda\mapsto2\pi-\lambda,v\mapsto\pi-v)$.  We find
\begin{eqnarray*}
    \hat{\alpha}_0(2\pi-\lambda,\pi-v,\theta)&=&-\hat{\alpha}_0(\lambda,v,\theta)\\
    \tilde{\epsilon}(2\pi-\lambda,\pi-v,\theta)&=&\tilde{\epsilon}(\lambda,v,\theta)\\
    \hat{R}(2\pi-\lambda,\pi-v,\theta)&=&\hat{R}(\lambda,v,\theta)\,,
\end{eqnarray*}
from which it follows that we can restrict both $\lambda$ and $v$ to the range $[0,\pi]$. Finally, the original range of $\theta$ is $[0,\pi/4-\cos^{-1}(1/(R\sqrt{2})]$, with $R\in[\sqrt{2}(2\score-1),1]$, hence the largest $\theta$ that needs to be considered for a given $\score$ is $\pi/4-\cos^{-1}(1/(4\score-2))$. Since we are using the functions with extended domain, it does not matter that we allow the range of $\theta$ to potentially be incompatible with the value of $\hat{R}$. This gives the optimization claimed in~\eqref{eq:newopt}.
\end{proof}

\subsubsection{Upper bounding the derivatives} 
For brevity in this section we often use $\bar{\theta}=\pi/4+\theta$. We upper-bound the derivatives for the functions $\hat{\alpha}_0, \tilde{\epsilon}, \hat{R}$. We first upper bound the derivatives for $\alpha$ as 
\begin{eqnarray}
\Big| \partial_{\lambda} \hat{\alpha}_0 \Big| &=&  \Big|\frac{2 \cot \left(\bar{\theta}\right) \csc ^2(\lambda )}{\cot ^2\left(\bar{\theta}\right) \cot
   ^2(\lambda )+1}  \Big|\\
\Big| \partial_v \hat{\alpha}_0 \Big| &=& \Big|\frac{\cot \left(\bar{\theta}\right) \csc ^2(v)}{\cot ^2\left(\bar{\theta}\right) \cot ^2(v)+1} \Big|\\ 
\Big| \partial_{\bar{\theta}} \hat{\alpha}_0 \Big| &=& \Big|\frac{ 2 \csc ^2\left(\bar{\theta}\right) \cot (\lambda )}{\cot ^2\left(\bar{\theta}\right) \cot ^2(\lambda )+1} -  \frac{\csc ^2\left(\bar{\theta}\right) \cot (v)}{\left(\cot ^2\left(\bar{\theta}\right) \cot ^2(v)+1\right)} \Big|.
\end{eqnarray}
Observe that for $x\in\mathbb{R}$
\begin{eqnarray}
    \frac{a \csc^2(x)}{a^2 \cot^{2}(x) + 1} &\leq& \max \{a, \frac{1}{a}\}.  
\end{eqnarray}
Noting that $\cot(\bar{\theta}) \leq 1$ for $\theta \in [0,\frac{\pi}{4}]$. This gives us 
\begin{eqnarray}
|\partial_{\lambda}\hat{\alpha}_0 |  &\leq&  2 \tan(\bar{\theta}) =: \alpha_{\lambda}\\
|\partial_v\hat{\alpha}_0 |          &\leq& \tan(\bar{\theta})  =: \alpha_v.
\end{eqnarray}
The identity $\frac{x}{1 + a^2 x^2}\leq \frac{1}{2|a|}$ can be used to get the following upper bound 
\begin{eqnarray*}
    \Big| \partial_{\bar{\theta}} \hat{\alpha}_0 \Big| &\leq& \Big|\frac{2 \csc ^2\left(\bar{\theta}\right) \cot (\lambda )}{\cot ^2\left(\bar{\theta}\right) \cot ^2(\lambda )+1} \Big|  + \Big|\frac{\csc ^2\left(\bar{\theta}\right) \cot (v)}{\left(\cot ^2\left(\bar{\theta}\right) \cot ^2(v)+1\right)}  \Big|\\ 
    &\leq&  2\frac{\csc^{2}(\bar{\theta}) \tan(\bar{\theta})}{2} + \frac{\csc^2(\bar{\theta})\tan(\bar{\theta})}{2}\\ 
    &=&\frac{3}{2\sin(2\bar{\theta})}:=\alpha_{\bar{\theta}}.
\end{eqnarray*}

Define $z(\lambda, v, \theta)$ to be the denominator in~\eqref{two sided function 2}, i.e., \begin{eqnarray}
z(\lambda, v, \theta)\!:=\!\cos(v\!-\!\lambda)\!\left[\sin(\hat{\alpha}_0)\sin(v)\cos(\frac{\pi}{4}\!+\!\theta)\!+\!\cos(\hat{\alpha}_0)\cos(v)\sin(\frac{\pi}{4}\!+\!\theta)\right]\!-\!\frac{\sin(v\!-\!\lambda)}{\sqrt{2}}\sqrt{1\!-\!\cos(2 v)\sin(2 \theta)}.
\end{eqnarray}
We now compute the derivatives of $z$.  For the derivative with respect to $\lambda$, we write $\partial_{\lambda} z = b_1 + b_2 \partial_{\lambda} \hat{\alpha}_0$, where 
\begin{eqnarray*}
b_1 &=& \sin(v-\lambda)\left(\sin(\hat{\alpha}_0)\sin(v)\cos(\bar{\theta})+\cos(\hat{\alpha}_0)\cos(v)\sin(\bar{\theta})\right)+\frac{\cos(v-\lambda)\sqrt{1-\cos(2v)\sin(2\theta)}}{\sqrt{2}}\\
b_2 &=& \cos(v-\lambda)\left(\cos(\hat{\alpha}_0)\sin(v)\cos(\bar{\theta})-\sin(\hat{\alpha}_0)\cos(v)\sin(\bar{\theta})\right).
\end{eqnarray*}
We can then bound these by $b_1 \leq \cos(\bar{\theta}) + \sin(\bar{\theta}) +1/\sqrt{2} \leq \sqrt{2} + 1/\sqrt{2}$ and $b_2 \leq \cos(\bar{\theta}) + \sin(\bar{\theta}) \leq \sqrt{2} $, so that 
\begin{eqnarray}
|\partial_{\lambda} z| \leq \sqrt{2}(3/2+\alpha_{\lambda}) =: z_{\lambda}.
\end{eqnarray}

Note that $\partial_v[\cos(v-\lambda)\sin(v)]=\cos(\lambda-2v)$ and $\partial_v[\cos(v-\lambda)\cos(v)]=\sin(\lambda-2v)$. We can hence write the $v$ derivative as 
\begin{eqnarray}
\partial_vz = a_1 + a_2 + a_3 \partial_{v} \hat{\alpha}_0\text{,\quad where}
\end{eqnarray}
\begin{eqnarray*}
a_1&=&\cos(\lambda-2v)\sin(\hat{\alpha}_0)\cos(\bar{\theta})+\sin(\lambda-2v)\cos(\hat{\alpha}_0)\sin(\bar{\theta}) \leq \cos(\bar{\theta}) + \sin(\bar{\theta}) \leq \sqrt{2}\\
   a_2 &=& -\frac{\cos (v-\lambda ) \sqrt{1-\sin (2
   \theta ) \cos (2 v)}}{\sqrt{2}}-\frac{\sin (2 \theta ) \sin(2 v) \sin (v-\lambda )}{\sqrt{2} \sqrt{1-\sin (2 \theta )\cos (2 v)}} \leq |\cos(v-\lambda)|+|\sin(v-\lambda)|\leq \sqrt{2}\\
   a_3 &=& \cos(v-\lambda)\left(\cos(\bar{\theta})\sin(v)\cos(\hat{\alpha}_0)-\sin(\bar{\theta})\cos(v)\sin(\hat{\alpha}_0)\right) \leq \cos(\bar{\theta}) + \sin(\bar{\theta}) \leq \sqrt{2},
\end{eqnarray*} 
and where we obtained the bound on $a_2$ using $|\frac{\sin(2v)\sin(2\theta)}{\sqrt{1 - \sin(2\theta) \cos(2v))}}| \leq \sqrt{2}$. Hence, we can bound
\begin{eqnarray}
|\partial_vz|\leq\sqrt{2}(2+\alpha_v)=: z_v.
\end{eqnarray}
Finally we compute the $\bar{\theta}$ derivative 
\begin{eqnarray}
\partial_{\bar{\theta}} z &=& c_1 + c_2+c_3 \partial_{\bar{\theta}} \hat{\alpha}_0
\end{eqnarray}
where 
\begin{eqnarray*}
c_1 &=&\cos(v-\lambda)\left(\cos(\hat{\alpha}_0)\cos(v)\cos(\bar{\theta})-\sin (\hat{\alpha}_0)\sin(v)\sin(\bar{\theta})\right) \leq \cos(\bar{\theta}) + \sin(\bar{\theta}) \leq \sqrt{2}\\
c_2&=&\frac{\cos(2\theta)\cos(2v)\sin(v-\lambda)}{\sqrt{2}\sqrt{1-\sin(2\theta)\cos(2v)}}\leq\frac{\cos(2\theta)}{\sqrt{2}\sqrt{1-\sin(2\theta)}}=\sqrt{\frac{1+\sin(2\theta)}{2}}\leq1\\
c_3 &=&\cos(v-\lambda)\left(\cos(\hat{\alpha}_0)\sin(v)\cos(\bar{\theta})-\sin(\hat{\alpha}_0)\cos(v)\sin(\bar{\theta})\right)\leq \cos(\bar{\theta}) + \sin(\bar{\theta}) \leq \sqrt{2}
\end{eqnarray*}
We hence obtain
\begin{eqnarray}
|\partial_{\bar{\theta}} z| &\leq&  \sqrt{2} + 1 + \sqrt{2}\alpha_{\bar{\theta}} =: z_{\theta}.
\end{eqnarray}
We now compute the derivatives of $\tilde{\epsilon}$:
\begin{eqnarray*}
   \partial_{\lambda} \tilde{\epsilon} &=& -\partial_{\lambda}\hat{\alpha}_0 (\cos (\theta ) \sin
   (\hat{\alpha}_0+\lambda-2 v)+\sin (\theta )
   \sin (\hat{\alpha}_0 -\lambda +2 v))+\sin (\theta
   ) \sin (\hat{\alpha}_0-\lambda +2 v)-\cos
   (\theta ) \sin (\hat{\alpha}_0+\lambda-2 v)\\ 
\partial_{v}\tilde{\epsilon} &=& -\partial_{v}\hat{\alpha}_0 (\cos (\theta)\sin(\hat{\alpha}_0+\lambda-2v)+\sin(\theta)
   \sin (\hat{\alpha}_0-\lambda +2 v))-2 \sin
   (\theta ) \sin (\hat{\alpha}_0-\lambda +2 v)+2
   \cos (\theta ) \sin (\hat{\alpha}_0+\lambda-2
   v)\\
   \partial_{\theta}\tilde{\epsilon} &=&-\partial_{\theta} \hat{\alpha}_0 (\cos (\theta ) \sin(\hat{\alpha}_0+\lambda-2 v)+\sin (\theta )
   \sin (\hat{\alpha}_0 -\lambda +2 v))+\cos (\theta
   ) \cos (\hat{\alpha}_0 -\lambda +2 v)-\sin
   (\theta ) \cos (\hat{\alpha}_0+\lambda-2 v).
\end{eqnarray*}
Using the same techniques as above, we find the following bounds 
\begin{eqnarray*}
|\partial_{\lambda}\tilde{\epsilon}| &\leq& \alpha_{\lambda} (\cos(\theta) + \sin(\theta )) + \cos(\theta) + \sin(\theta)  \leq  \sqrt{2}(\alpha_{\lambda}+1)=: \epsilon_{\lambda}\\
|\partial_{v}\tilde{\epsilon}| &\leq& \alpha_{v} (\cos(\theta) + \sin(\theta )) + 2\cos(\theta) + 2\sin(\theta)  \leq  \sqrt{2}(\alpha_{v}+2)=:\epsilon_v\\
|\partial_{\theta}\tilde{\epsilon}| &\leq& \alpha_{\theta} (\cos(\theta) + \sin(\theta )) + \cos(\theta) + \sin(\theta)  \leq  \sqrt{2}(\alpha_{\theta}+1) =: \epsilon_{\theta}.
\end{eqnarray*}

\subsubsection{Lower bounding the function} 
\noindent Consider a partition $\M{P}$ of $\M{D}_\score'$. Let $\M{C}_{i,j,k}$ be a cuboid (with $i$ label corresponding to $\lambda$, $j$ label for $v$ and $k$ label for $\theta$). Let $\Delta z = z_{\lambda}(\lambda_{i+1} - \lambda_i) + z_{v}(v^{(i)}_{j+1} - v^{(i)}_j)+z_{\theta}(\theta^{(i,j)}_{k+1} - \theta^{(i,j)}_k)$, then in $\M{C}_{i,j,k}$ 
\begin{eqnarray}
R_{\min}^{i,j,k} := \frac{\sqrt{2}(2 \score -1 )}{z(\lambda_i, v_j^{(i)}, \theta_k^{(i,j)}) + \Delta z } \leq \hat{R}(\lambda , v , \theta) = \frac{\sqrt{2}(2 \score -1 )}{z(\lambda_i, v_j^{(i)}, \theta_k^{(i,j)})  } \leq \frac{\sqrt{2} (2 \score - 1)}{z(\lambda_i, v_j^{(i)}, \theta_k^{(i,j)}) - \Delta z} =: R_{\max}^{i,j,k}  
\end{eqnarray}
Also let $\Delta \epsilon := \epsilon_{\lambda}(\lambda_{i+1} - \lambda_i) + \epsilon_{v}(v_{j+1}^{(i)} - v_j^{(i)})+\epsilon_{\theta}(\theta_{k+1}^{(i,j)} - \theta_k^{(i,j)})$, then in $\M{C}_{i,j,k}$ we have
\begin{eqnarray}
\tilde{\epsilon}(\lambda,v,\theta)\leq\epsilon_{\max}^{i,j ,k} := \tilde{\epsilon}(\lambda_i, v_j, \theta_k) + \Delta \epsilon  
\end{eqnarray}
For each cuboid we define a continuous function $g_{i,j,k}: \M{C}_{i,j,k} \rightarrow \mathbb{R}$ such that $g_{i,j,k}(\B{x}) \leq \hat{P}(\B{x})$ for all $\B{x} \in \M{C}_{i,j,k}$. Then we lower bound $G_{AB|X=0,Y=0,E}$ by using the following.
\begin{lemma}\label{lemm: Lower bound om H(A|X=0,Y=0,E) on grids}
Let
\begin{eqnarray}\label{eqn: two sided lower bound}
g_{i,j,k} := \hat{H}_{\bin}\left(\frac{1}{2} + \frac{R^{i,j,k}_{\max} \epsilon_{\max}^{i,j,k}}{2}\right) + \hat{K}(R^{i,j,k}_{\min}, \theta_{k}^{(i,j)}).
\end{eqnarray}
Then $\hat{P}(\B{x})\geq g_{i,j,k}$ for all $\B{x}\in\M{C}_{i,j,k}$.
\end{lemma}
\begin{proof}
By definition, we have $R_{\max}^{i,j,k} \epsilon_{\max}^{i,j,k} \geq \hat{R}(\lambda, v, \theta) \tilde{\epsilon}(\lambda, v, \theta)$ for all $\B{x}\in\M{C}_{i,j,k}$. Using the monotonicity of the function $\hat{H}_{\bin}(\frac{1}{2} + \frac{x}{2})$, we obtain $\hat{H}_{\bin}(\frac{1}{2} + \frac{R_{\max}^{i,j,k} \epsilon_{\max}^{i,j,k}}{2}) \leq \hat{H}_{\bin}(\frac{1}{2} + \frac{\hat{R}(\lambda, v, \theta) \tilde{\epsilon}(\lambda, v, \theta)}{2})$. Similarly, the monotonicity of $\hat{K}(R,\theta)$ with respect to $R$ and $\theta$ (see Lemmas~\ref{lem: monotonicity of K(R)} and~\ref{lem: monotonicity of K(th)}) implies $\hat{K}(R(\lambda,v,\theta),\theta)\geq\hat{K}(R_{\min},\theta_{k}^{(i,j)})$ for all $\B{x}\in\M{C}_{i,j,k}$.  These imply the claim.
\end{proof} 
\noindent Combining all the results in this section, we have the following
\begin{corollary}
Let $\score \in (\frac{3}{4}, \frac{1}{2} + \frac{1}{2 \sqrt{2}}]$ be fixed. Let $\M{D}_\score' = \{(\lambda, v, \theta) \in \mathbb{R}^{3} : \lambda \in [0, \pi], v \in [0, \pi], \theta \in [0, \frac{\pi}{4} - \cos^{-1} \left(\frac{1}{2(2 \score -1 )} \right) ] \}$ and $\M{P} = \bigcup_{i,j,k} \M{C}_{i,j,k}$ be a partition of any cuboid $\M{C} \supseteq \M{D}'(\score)$ as specified above. Then 
\begin{eqnarray}
G_{AB|X=0,Y=0,E}(\score) \geq \min_{i,j,k} g_{i,j,k}
\end{eqnarray}
where $g_{i,j,k}$ are defined in~\eqref{eqn: two sided lower bound}.
\end{corollary}
\begin{proof}
This is a direct consequence of Lemmas~\ref{lemm: H(A|X=0,Y=0,E) on grids} and~\ref{lemm: Lower bound om H(A|X=0,Y=0,E) on grids}.
\end{proof}

\subsection{H(AB|XYE)}\label{app:ABgXYE}
In this case we again use Lemma~\ref{lem:ent_simp} to obtain
\begin{align*}
H(AB|XYE)&=H(AB|XY)+\sum_{abxy}p_{ABXY}(a,b,x,y)H(\tau_E^{abxy})-\sum_{xy}p_{XY}(x,y)H\left(\sum_{ab}p_{AB|xy}(a,b)\tau_E^{abxy}\right)\\
  &=H(AB|XY)-H(E)\,,
\end{align*}
where we again use that $H(\tau_E^{abxy})=0$, and note that $\sum_{ab}p_{AB|xy}(a,b)\tau_E^{abxy}=\rho_E$ for all $x,y$.  Note that
\begin{align*}
  H(AB|XY)&=\sum_{xy}p_{XY}(x,y)H(AB|X=x,Y=y)\\
          &=1+\sum_{xy}p_{XY}(x,y)H_\bin(2\epsilon_{xy})\,,
\end{align*}
and so we have
\begin{equation}\label{eq:HABgXYE}
    H(AB|XYE)=1+\sum_{xy}p_{XY}(x,y)H_\bin(2\epsilon_{xy})-H(\{\lambda_0,\lambda_1,\lambda_2,\lambda_3\})\,.
\end{equation}
Again the $\delta$ dependence is all in the last term, so, like in the case of $H(AB|X=0,Y=0,E)$ we can take $\delta=\delta^*$ and remove $\delta$ from the optimization.

\subsection{H(AB|E)}
We first trace out $XY$ to give
$\tau_{ABE}=\sum_{ab}p_{AB}\proj{a}\ot\proj{b}\ot\sum_{xy}p_{XY|ab}(x,y)\tau_E^{abxy}$. For this state, Lemma~\ref{lem:ent_simp} gives
\begin{align*}
H(AB|E)&=H(AB)+\sum_{ab}p_{AB}(a,b)H\left(\sum_{xy}p_{XY|ab}(x,y)\tau_E^{abxy}\right)-H\left(\sum_{abxy}p_{AB}(a,b)p_{XY|ab}(x,y)\tau_E^{abxy}\right)\\
       &=H(AB)+\sum_{ab}p_{AB}(a,b)H\left(\sum_{xy}p_{XY|ab}(x,y)\tau_E^{abxy}\right)-H(E)\\
       &=H(AB)+\sum_{ab}p_{AB}(a,b)H\left(\sum_{xy}\frac{1}{p_{AB}(a,b)}p_{XY}(x,y)p_{AB|xy}(a,b)\tau_E^{abxy}\right)-H(E)\,.
\end{align*}
In this case we cannot remove the middle term, and the middle term is not independent of $\delta$. The optimization in this case is hence significantly more complicated.  Note that
\begin{align*} H(AB)=1+H_\bin\left(2\left(p_{XY}(0,0)\epsilon_{00}+p_{XY}(0,1)\epsilon_{01}+p_{XY}(1,0)\epsilon_{10}+p_{XY}(1,1)\left(\frac{1}{2}-\epsilon_{11}\right)\right)\right)\,.
\end{align*}

\section{Monotonicity}\label{app:mono}
In this section we prove the monotonicity of the functions $G_{A|XYE}(\score)$, $G_{AB|00E}(\score)$ and $G_{AB|XYE}(\score)$. There is a common part to the proofs, which we first establish.
\begin{lemma}\label{Monotonicity claim 2}
Let $\lambda_0(R,\theta), \lambda_1(R, \theta) , \lambda_2(R , \theta)$ and $\lambda_3(R, \theta)$ be the eigenvalues of a Bell-diagonal states $\rho_{A'B'}$ as in~\eqref{param1}--\eqref{param4} in the case where $\delta=\frac{R^2}{4}\cos(2\theta)$. Then
\begin{eqnarray}
\frac{\partial}{\partial R} \left(H_{\bin}(\lambda_0 + \lambda_1) - H(\{\lambda_0,\lambda_1,\lambda_2,\lambda_3\}) \right) > 0. 
\end{eqnarray}
\end{lemma}
\begin{proof}
\begin{eqnarray}
\frac{\partial}{\partial R} \Big( H_{\bin}(\lambda_0 + \lambda_1) \Big) &=& -\log(\lambda_0 + \lambda_1) \frac{\partial}{\partial R} (\lambda_0 + \lambda_1)  -\log(\lambda_2 + \lambda_3) \frac{\partial}{\partial R} (\lambda_2 + \lambda_3) 
\end{eqnarray}
The equality above follows from the fact that $1 - \lambda_1 - \lambda_0 = \lambda_2 + \lambda_3$ and thus $H_{\bin}(\lambda_0+ \lambda_1) = -(\lambda_0 + \lambda_1) \log(\lambda_1 + \lambda_0) - (\lambda_2 + \lambda_3) \log(\lambda_2 + \lambda_3)$. We also have that
\begin{eqnarray}
\frac{\partial }{ \partial R} H(\{\lambda_0,\lambda_1,\lambda_2,\lambda_3\}) =- \sum_i\log{\lambda_i}\frac{\partial \lambda_i}{\partial R} .
\end{eqnarray}
Adding the derivatives, we have
\begin{eqnarray}
\frac{\partial}{\partial R} \Big( H_{\bin}(\lambda_0 + \lambda_1) - H(\{\lambda_0,\lambda_1,\lambda_2,\lambda_3\})\Big)  &=& \log\left(\frac{\lambda_0}{\lambda_0 + \lambda_1}\right) \frac{\partial \lambda_0}{\partial R} + \log\left(\frac{\lambda_1}{\lambda_0 + \lambda_1}\right) \frac{\partial \lambda_1}{\partial R}  \nonumber \\& & + \log\left(\frac{\lambda_2}{\lambda_3 + \lambda_2}\right) \frac{\partial \lambda_2}{\partial R} + \log\left(\frac{\lambda_3}{\lambda_2 + \lambda_3}\right) \frac{\partial \lambda_3}{\partial R}  \nonumber\\ 
&=& \log_{2}\Big( \frac{\lambda_0}{\lambda_0 + \lambda_1}\Big) \frac{\partial}{\partial R}(\lambda_0 + \lambda_2) + \log_{2}\Big( \frac{\lambda_1}{\lambda_0 + \lambda_1}\Big) \frac{\partial}{\partial R}(\lambda_1 + \lambda_3) \nonumber  \\  
&=& \log\left(\frac{\lambda_0}{\lambda_1}\right)\frac{\partial}{\partial R}(\lambda_0 + \lambda_2)=\log\left(\frac{\lambda_0}{\lambda_1}\right)\frac{\cos(\theta)-\sin(\theta)}{2}\nonumber\\
&\geq& 0 .
\end{eqnarray}
Where the second equality follows from the fact that for Bell-diagonal states parameterized by $\delta = \frac{R^2}{4}\cos(2\theta)$, the eigenvalues obey 
\begin{eqnarray*}
\frac{\lambda_0}{\lambda_0 + \lambda_1} = \frac{\lambda_2}{\lambda_2 + \lambda_3} \quad\text{and}\quad
\frac{\lambda_1}{\lambda_0 + \lambda_1} = \frac{\lambda_3}{\lambda_2 + \lambda_3}
\end{eqnarray*}
and the inequality comes from the parameterization.
\end{proof}

\begin{lemma}\label{lemm: Monotonicity of one sided rates}
For $\score\in(\frac{3}{4},\frac{1}{2}(1+\frac{1}{\sqrt{2}}))$ and any distribution $p_{XY}$, the function $G_{A|XYE}(\omega,p_{XY})$ is increasing in $\score$.
\end{lemma}

\begin{proof}
Let us fix the score $\omega$. From the analysis in Appendix~\ref{app:AgXYE} we know that the optimum value of $\delta$ is $\frac{R^2}{4}\cos(2\theta)$.  Throughout this proof we take $\delta=\frac{R^2}{4}\cos(2\theta)$ and consider $\rho_{A'B'}$ to depend on two parameters $R$ and $\theta$.  Let $(\M{N}^*, \rho^*) \equiv \rho(R^*, \theta^*)$ be the channel and state that that solves the optimization problem for $G_{A|XYE}(\score,p_{XY})$, i.e., such that $G_{A|XYE}(\score,p_{XY})=H(A|XYE)_{(\cN^*\ot\cI_E)(\rho_{A'B'E}(R^*,\theta^*))}$.
It suffices to show that there exists a curve $\sigma: [-1 , 0] \mapsto \M{S}(\M{H}_{A'}\ot\M{H}_{B'}\ot\M{H}_E)$, such that
\begin{enumerate}
    \item $\sigma(0) = \rho^*$ 
    \item $g(t) := H(A|XYE)_{(\M{N}^*\ot\cI_E)(\sigma(t))}$ is differentiable for all $t \in [-1,0]$.  
    \item $\left.\frac{\dd g(t)}{\dd t}\right|_{t=0} > 0$ 
    \item $\forall t: \frac{\dd}{\dd t}\scorefunction\left((\M{N}^*\ot\cI_E)(\sigma(t))\right)>0$.
\end{enumerate}
Then, if 1--4 hold, using the fact that $g(t)$ is continuous and has a positive derivative at $t = 0$, there exists $t_0 < 0$ such that for $t\in(t_0,0)$, $g(t) < g(0)$. Since the $\scorefunction\left((\M{N}^*\ot\cI_E)(\sigma(t))\right)$ is continuous function, we must have that for any $t\in(t_0,0)$
\begin{eqnarray}
H(A|X Y  E)_{(\M{N}^*\ot\cI_E)(\rho_{A'B'E}^*)} &>& H(A|X Y E)_{(\M{N}^*\ot\cI_E)(\sigma(t))} \\ 
&\geq& G_{A|XYE}\left(\scorefunction((\M{N}^*\ot\cI_E)(\sigma(t)),p_{XY}\right).
\end{eqnarray}
Since $\scorefunction((\M{N}^*\ot\cI_E)(\sigma(t))<\score$ this establishes the claim.

It remains to show that there exists a function $\sigma(t)$ such that 1--4 hold. Recall from Appendix~\ref{app:AgXYE} that we can write
\begin{eqnarray}\label{Just above}
H(A|XYE)=1+p_X(0)H_{\bin}(g(\theta,\alpha_0))+p_X(1)H_{\bin}(g(\theta,\alpha_1) )-H(\{\lambda_0,\lambda_1,\lambda_2,\lambda_3\})
\end{eqnarray}
where $g(\theta, \alpha) := \frac{1}{2}\left(1+R\sqrt{1+\sin(2\theta)\cos(4\alpha)}\right)$. We then set
\begin{eqnarray}
\sigma(t)=\rho(R^*+\kappa t,\theta^*)
\end{eqnarray}
for some positive number $\kappa$ such that $R^*-\kappa>3/4$. Thus, $\sigma(0)=\rho^*$ and differentiability of $g(t)$ can be shown using the form~\eqref{Just above}. We compute the $t$ derivative:
\begin{eqnarray}
\left.\frac{\dd g(t)}{\dd t}\right|_{t=0}=\left.\kappa\frac{\partial}{\partial R} \Big(H(A|X Y E)_{(\M{N^*}\ot\cI_E)(\rho_{A'B'E}(R,\theta))} \Big) \right|_{R=R^*,\theta=\theta^*}
\end{eqnarray}
Note that 
\begin{eqnarray*}
\frac{\partial}{\partial R} H_{\bin} (g(\theta, \alpha)) &=&  H_{\bin}'(g(\theta, \alpha)) \frac{\sqrt{1 + \sin(2 \theta) \cos(4\alpha)}}{2}    \\
&\geq& H_{\bin}'\Big(\frac{1}{2} + \frac{R}{2} (\cos(\theta) + \sin(\theta)) \Big) \frac{\cos(\theta) + \sin(\theta)}{2} \\
&=& H_{\bin}'(\lambda_0 + \lambda_1) \frac{\partial}{\partial R}(\lambda_0 + \lambda_1)\\
&=& \frac{\partial}{\partial R}H_{\bin}(\lambda_0+\lambda_1)\,,
\end{eqnarray*} 
where we have used that $H'_\bin(p)$ is decreasing in $p$ for $p>1/2$, so we take $\alpha=0$ to obtain a bound. It follows that 
\begin{eqnarray*}
\left.\frac{\dd g(t)}{\dd t}\right|_{t=0}=\kappa\frac{\partial}{\partial R} \Big(H_{\bin}(\lambda_0 + \lambda_1) - H(\{\lambda_0,\lambda_1,\lambda_2,\lambda_3\}) \Big) > 0\,,
\end{eqnarray*}
where the inequality is Lemma~\ref{Monotonicity claim 2}.

Finally, the function 
$$\scorefunction((\M{N}^*\ot\cI_E)(\sigma(t))=\frac{1}{2}\sum_{i,j} \epsilon_{ij}$$
increases linearly with $t$ (the score is linear in $R$).
\end{proof}

\begin{lemma}\label{lem:monot:w}
For $\score\in(\frac{3}{4},\frac{1}{2}(1+\frac{1}{\sqrt{2}}))$, the function $G_{AB|X=0,Y=0,E}(\score)$ is increasing in $\score$.
\end{lemma}
\begin{proof}
The proof follows the same lines as the previous lemma but with the entropy changed. From Appendix~\ref{app:AB00E} we have
$$H(AB|X=0,Y=0,E)=1+H_\bin(2\epsilon_{00})-H(\{\lambda_0,\lambda_1,\lambda_2,\lambda_3\}).$$
We have
\begin{eqnarray}
\frac{\partial}{\partial R}  H_{\bin} (2\epsilon_{00})&=&H'_{\bin}(2\epsilon_{00})\frac{\cos(\theta)\cos(2(\alpha_0-\beta_0))+\sin(\theta)\cos(2(\alpha_0+\beta_0))}{2}\nonumber\\
&\geq&H'_{\bin}\left(\frac{1}{2}+\frac{R}{2}(\cos(\theta)+\sin(\theta))\right)\frac{\cos(\theta)+\sin(\theta)}{2}\label{eq:eps00}
\end{eqnarray}
and the remainder of the argument matches the previous proof.
\end{proof}

\begin{lemma}
For $\score\in(\frac{3}{4},\frac{1}{2}(1+\frac{1}{\sqrt{2}}))$ and any distribution $p_{XY}$, the function $G_{AB|XYE}(\score,p_{XY})$ is increasing in $\score$.
\end{lemma}
\begin{proof}
The proof for this again follows those above, except in this case (see Appendix~\ref{app:ABgXYE})
\begin{eqnarray}
H(AB|XYE) &=& 1 + \sum_{xy} p_{XY}(x,y) H_{\bin}(2\epsilon_{xy}) - H(\{\lambda_0,\lambda_1,\lambda_2,\lambda_3\})\,.
\end{eqnarray}
The bound that holds for $\epsilon_{00}$ in~\eqref{eq:eps00} holds for all $\epsilon_{xy}$, and hence the rest of the argument goes through as before.
\end{proof}

\section{Entropy Accumulation Theorem}\label{app:EAT}
In this section we state the Entropy Accumulation Theorem (EAT). The theorem is phrased in terms of a set of channels $\{\cM_i\}_i$ called EAT channels, where $\cM_i:\cS(R_{i-1})\to\cS(C_iD_iU_iR_i)$.
\begin{definition}[EAT Channels]
Let $\{R_i\}_{i=0}^n$  be arbitrary quantum systems and $\{C_i\}_{i=1}^n$, $\{D_i\}_{i=1}^n$, and $\{U_i\}_{i=1}^n$ be finite dimensional classical systems. Suppose that $U_i$ is a deterministic function of $C_i$, $D_i$ and that $\{\cM_i\}_{i=1}^n$, $\cM_i:\cS(R_{i-1})\to\cS(C_iD_iU_iR_i)$ are a set of quantum channels.  These channels form a set of EAT channels if for all $\rho_{R_0E}\in\cS(R_0E)$ the state $\rho_{{\bf CDU}R_nE}=(\cM_n\circ\ldots\cM_1)(\rho_{R_0E})$ after applying the channels satisfies $I(C_1^{i-1}:D_i|D_1^{i-1}E)=0$, where $I$ is the mutual information, and $C_1^{i-1}$ is shorthand for $C_1C_2\ldots C_{i-1}$.
\end{definition} 
In the context of DI algorithms, the quantum register $R_0$ can be taken to represent the initial state of the devices, which may be entangled with the register $E$. This state updates to $R_1, R_2, \ldots$ as the protocol proceeds. At step $i$ the devices (together with the random number generators $\M{R}_{A}$ and $\M{R}_B$) perform a map $\M{M}_i$ to give the output classical random variables and the random choices generated by the random number generators. The mutual information condition encodes the property that the random number generators are independent of $E$ and the previously generated data. The register $U_i$ records the score in the Bell game for that round. Along with the classical inputs, the channel $\M{M}_i$ also outputs the updated state of the devices represented by the register $R_i$ which may be stored by the device and acted on by the next channel\footnote{In practice, the devices may be sent new states in each round, but there is no loss in generality in assuming that the devices pre-share all the entangled quantum resources they need for the protocol.}. Each EAT channel, therefore, for the DI protocols is a set of maps of the form
\begin{eqnarray}\label{EAT channel}
\cM_i(\rho) = \sum_{c,d} \proj{c}\ot\proj{d}\ot\proj{u(c,d)}\ot\cM_i^{c,d}(\rho)\,,
\end{eqnarray}
where $u(c,d)$ records the score in the Bell game, and each $\M{M}_i^{c,d}$ is a subnormalized quantum channel from $\cS(R_{i-1})$ to $\cS(R_i)$. The joint distribution of the classical variables $C_i$ and $D_i$ is
\begin{equation}
    p_{C_iD_i}(c,d):=\tr(\M{M}_i^{c,d}(\rho))\,.
\end{equation}
\begin{definition}[Frequency distribution function]
Let ${\bf U}=U_1U_2\ldots U_n$ be a string of variables. The associated frequency distribution is 
\begin{equation}
    \Freq_{{\bf U}}(u) := \frac{|\{i\in\{1,\ldots,n\}:U_i=u\}|}{n}\,.
\end{equation}
\end{definition}
\begin{definition}
Given a set of channels $\mathfrak{G}$ whose outputs have a register $U$, the set of achievable score distributions is
\begin{equation}
    \cQ_{\mathfrak{G}}:=\{p_U:\cM(\rho)_U=\sum_up_U(u)\proj{u}\text{ for some }\cM\in\mathfrak{G}\}.
\end{equation}
We also use 
\begin{equation}
    \cQ_{\mathfrak{G}}^\gamma:=\{p_U:p_U(\bot)=(1-\gamma)\text{ and }p_U(u)=\gamma\tilde{p}_U(u)\text{ with }\tilde{p}_U\in\cQ_{\mathfrak{G}}\}.
\end{equation}
\end{definition}

\begin{definition}[Rate function]
Let $\mathfrak{G}$ be a set of EAT channels. A \emph{rate function} $\rate:\cQ_\mathfrak{G}\to\mathbb{R}$ is any function that satisfies
\begin{eqnarray}
\rate(q) \leq \inf_{(\M{M},\rho_{RE})\in\Gamma_{\mathfrak{G}}(q)} H(C|DE)_{(\M{M}\ot\cI_E)(\rho_{RE})}\,,
\end{eqnarray}
where
\begin{equation}
\Gamma_{\mathfrak{G}}(q):=\{(\cM,\rho_{RE}):(\cM\ot\cI_E)(\rho_{RE})_U=\sum_uq(u)\proj{u}\text{ for some }\cM\in\mathfrak{G}\}
\end{equation}
is the set of states and channels that can achieve distribution $q$.
\end{definition}
\begin{definition}[Min-tradeoff function]
A function $f:\cQ_\mathfrak{G}\to\mathbb{R}$ is a \emph{min-tradeoff function} if $f$ is an affine rate function.  Since min-tradeoff functions are affine, we can naturally extend their domain to all probability distributions on $U$, denoted $\cP$.
\end{definition}
The entropy accumulation theorem then can be stated as follows (this is Theorem~2 of~\cite{LLR&}, which is a generalization of the results of~\cite{DF}).

\begin{theorem}\label{thm:EAT}
 Let $f$ be a min-tradeoff function for a set of EAT channels $\mathfrak{G}=\{\cM_i\}_i$ and $\rho_{{\bf CDU}E}$ be the output after applying these channels to initial state $\rho_{RE}$. In addition let $\epsilon_h\in(0,1)$, $\alpha\in(1,2)$ and $r\in\mathbb{R}$ and $\Omega$ be an event on ${\bf U}$ that implies $f(\Freq_{\bf U})\geq r$. We have
 \begin{align}
     H_{\min}^{\epsilon_h}({\bf C}|{\bf D}E)_{\rho_{{\bf CD}E|\Omega}}>&nr-\frac{\alpha}{\alpha-1}\log\left(\frac{1}{p_\Omega(1-\sqrt{1-\epsilon_h^2})}\right)+\nonumber\\
     &n\inf_{p\in\cQ_{\mathfrak{G}}}\left(\Delta(f,p)-(\alpha-1)V(f,p)-(\alpha-1)^2K_{\alpha}(f)\right)\,,
 \end{align}
where $\Delta(f,p)=\rate(p)-f(p)$, and
\begin{align*}
    V(f,p)&=\frac{\ln2}{2}\left(\log(1+2d_C)+\sqrt{2+\Var_p(f)}\right)^2\\
    K_\alpha(f)&=\frac{1}{6(2-\alpha)^3\ln 2}2^{(\alpha-1)(\log(d_C)+\Max(f)-\Min_{\cQ_{\mathfrak{G}}}(f))}\ln^3\left(2^{\log(d_C)+\Max(f)-\Min_{\cQ_{\mathfrak{G}}}(f)}+\e^2\right)\,,
\end{align*}
and we have also used
\begin{align*}
\Max(f)&=\max_{p\in\cP} f(p)\\
\Min_{\cQ_{\mathfrak{G}}}(f)&=\inf_{p\in\cQ_{\mathfrak{G}}}f(p)\\
\Var_p(f)&=\sum_up(u)\left(f(\delta_u)-\mathbb{E}(f(\delta_u))\right)^2\,,
\end{align*}
and $\delta_u$ is the deterministic distribution with outcome $u$.
\end{theorem}
To use this theorem we have to assign the variables $C_i$ and $D_i$ to the parameters in the protocol.

\subsection{Protocol with recycled input randomness (Protocol~\ref{prot:nonspotcheck})}
For this protocol we want to extract randomness from the inputs and outputs. We hence set $C_i=A_iB_iX_iY_i$ and take $D_i$ to be trivial.  When running a protocol, we do not generally know the set of EAT channels being used (these are set by the adversary), but instead only know that they have the no-signalling form, i.e., we have
$$\cM(\rho_{A'B'})=\sum_{abxy}\proj{a}\ot\proj{b}\ot\proj{x}\ot\proj{y}\ot\proj{u(a,b,x,y)}\ot\cM^{abxy}(\rho_{A'B'})\,,$$
where $\cM^{abxy}(\rho_{A'B'})=p_{XY}(x,y)(\cE^{x,a}\ot\cF^{y,b})(\rho_{A'B'})$ and $\{\cE^{x,a}\}_a$ and $\{\cF^{y,b}\}_b$ are instruments on $A'$ and $B'$ respectively (cf.~\eqref{EAT channel}). Henceforth, the set $\mathfrak{G}$ will refer to all channels of this type.

In the CHSH protocol without spot-checking the rate function should be a lower bound on $H(ABXY|E)=2+H(AB|XYE)$ and we can form our rate function via $\rate(\{1-s,s\})=2+F_{AB|XYE}(s)$ or $\rate(\{1-s,s\})=2+F_{A|XYE}(s)$, the former being preferred as it is larger. A min-tradeoff function can then be obtained by taking the tangent at some point. Since $F_{AB|XYE}(s)$ is linear for $3/4\leq s\leq\score_{AB|XYE}^*\approx0.847$, for experimentally relevant scores we can form the min-tradeoff function using the extension of this line to the domain $[0,1]$, i.e., we can take $f(\{1-s,s\})=2+G'_{AB|XYE}(\score^*)(s-3/4)$ in Theorem~\ref{thm:EAT} when applying to Protocol~\ref{prot:nonspotcheck}, and in this case $d_C=d_Ad_Bd_Xd_Y=16$ and we get a bound on $H_{\min}^{\epsilon_h}({\bf ABXY}|E)$. The theorem holds for all $\alpha\in(1,2)$ and we can optimize over $\alpha$ to increase the bound.

\subsection{Spot-checking CHSH protocol (Protocol~\ref{prot:spotcheck})}
To use the EAT in the spot-checking CHSH protocol (Protocol~\ref{prot:spotcheck}) we set $C_i=A_iB_i$ and $D_i=X_iY_i$ in Theorem~\ref{thm:EAT}. The channels again have the no-signalling form mentioned above, and we can use either $F_{AB|00E}$ or $F_{A|00E}$ as the basis of our rate function.  Since the two-sided version is larger, it is better to work with $F_{AB|00E}(s)$, and the related min-tradeoff function based on taking its tangent at some point.  Modification is required to account for the spot-checking structure.  If we let $g_t(\{1-s,s\})$ be the tangent of $F_{AB|00E}(s)$ taken at $t$ then we can form the spot-checking min-tradeoff functions
\begin{align*}
f_t(\delta_u)=\begin{cases}\frac{1}{\gamma}g_t(\delta_u)+(1-\frac{1}{\gamma})g_t(\delta_1)&u\in\{0,1\}\\
    g_t(\delta_1)&u=\bot\end{cases}\,.
\end{align*}
where $t$ can be chosen (see e.g.~\cite[Section~5]{DF} for the argument behind this). Using this construction the following theorem can be derived (this is an adaptation of Theorem~3 in~\cite{LLR&}).
\begin{theorem}[Entropy Accumulation Theorem for spot-checking CHSH protocol]\label{thm:EATspot} Let $\rho_{{\bf ABXYU}E}$ be a CQ state obtained using the spot-checking CHSH protocol (Protocol~\ref{prot:spotcheck}). Let $\Omega$ be the event $\left|\{i:U_i=0\}\right|\leq n\gamma(1-\score_{\exp}+\delta)$ with $p_\Omega$ being the probability of this event in $\rho_{{\bf ABXYU}E}$, and let $\rho_{{\bf ABXYU}E|\Omega}$ be the state conditioned on $\Omega$. Let $\epsilon_h\in(0,1)$ and $\alpha\in(1,2)$. Then for any $r$ such that $f_t(\Freq_{\bf U})\geq r$ for all events in $\Omega$ we have
\begin{eqnarray}
H^{\epsilon_h}_{\min}({\bf AB}|{\bf XY}E)_{\rho_{{\bf ABXY}E|\Omega}} > && n r - \frac{\alpha}{\alpha -1} \log\left(\frac{1}{p_\Omega (1- \sqrt{1 - \epsilon_h^2} )}\right) \\
&&\nonumber+ n\inf_{p \in\cQ^\gamma_{\mathfrak{G}}}(\Delta (f_t , p) - (\alpha - 1) V(f_t,p) - (\alpha -1)^2 K_{\alpha}(f_t))\,,
\end{eqnarray}
where 
\begin{eqnarray}
\Delta(f_t,p)&:=&F_{AB|XYE}(p(1)/\gamma)-f_{t}(p) \\ 
V(f_t, p) &=& \frac{\ln 2}{2} \left(\log(9)+\sqrt{\Var_p(f_t)+2}\right)^2 \\
K_{\alpha}(f_t)&=& \frac{1}{6 \log(2 - \alpha)^3 \ln2 }2^{(\alpha -1)(2 + \Max(f_t) -\Min_{\M{Q}^\gamma_{\mathfrak{G}}}(f_t))} \ln^3(2^{2 + \Max(f_t) - \Min_{\M{Q}^\gamma_{\mathfrak{G}}}(f_t)}+\e^2).
\end{eqnarray}
\end{theorem}
To use this theorem we can take $r=(F_{AB|00E}(t)+(\score_{\exp}-\delta-t)F'_{AB|00E}(t))$ (cf.\ the discussion in~\cite{LLR&}), and since the theorem holds for any $t$ and $\alpha$ these can be optimized over.

\subsection{Protocol with biased local random numbers (Protocol~\ref{prot:biased})}
To derive the randomness rates, we use Theorem~\ref{thm:EAT} with $C_i=A_iB_i$ and $D_i=X_iY_i$, as in the previous subsection. What remains is to derive the min-tradeoff function and error terms. In this section, we compute these quantities and derive the expression for the completeness error in terms of the biasing parameters $\zeta_A,\zeta_B$ and statistical error $\delta$. 

\subsubsection{Deriving the min-tradeoff function} 
We seek a min-tradeoff function suitable for using with Protocol~\ref{prot:biased}.  To construct it we write the EAT channel in a slightly different way that is explicit in the input distribution $p_{XY}$:
\begin{eqnarray}\label{EAT channel2}
\cM_{p_{XY}}(\rho)=\sum_{abxy}p_{XY}(x,y)\proj{a}_A\ot\proj{b}_B\ot\proj{x}_X\ot\proj{y}_Y\ot\proj{(x,y,w)}_U\ot\cM^{x,y}_{a,b}(\rho)\,,
\end{eqnarray}
where $\cM^{x,y}_{a,b}$ are subnormalized channels. We can also consider the analogous channel where the $U$ register only stores $w$ (we use $\tilde{\cM}$ to indicate this case).  Next consider the entropy $H(AB|X=0,Y=0E)$, this entropy is calculated for the normalization of the state
$$(\proj{0}_X\ot\proj{0}_Y\ot\id_{ABUE})(\cM_{p_{XY}}\ot\cI_E)(\rho_{RE})(\proj{0}_X\ot\proj{0}_Y\ot\id_{ABUE})\,.$$
For fixed $\{\cM^{x,y}_{a,b}\}$, this is independent of $p_{XY}$ (it is defined provided $p_{XY}(0,0)\neq0$).

We next note that for $q$ as the distribution on the score ($U$) register
\begin{eqnarray*}
(\cM_{p_{XY}}\ot\cI_E)(\rho_{RE})_U&=&\sum_{abxyw}p_{XY}\tr(\cM^{x,y}_{a,b}(\rho_R))\proj{(x,y,w)}=\sum_{xyw}q((x,y,w))\proj{(x,y,w)}\\
(\cM_{1/4}\ot\cI_E)(\rho_{RE})_U&=&\sum_{abxyw}\frac{1}{4}\tr(\cM^{x,y}_{a,b}(\rho_R))\proj{(x,y,w)}=\sum_{xyw}\frac{q((x,y,w))}{4p_{XY}}\proj{(x,y,w)}\,,
\end{eqnarray*}
and hence
$$(\tilde{\cM}_{1/4}\ot\cI_E)(\rho_{RE})_U=\sum_{xyw}\frac{q((x,y,w))}{4p_{XY}}\proj{w}\,.$$

It follows that
\begin{align*}
  &\biggl\{H(AB|X=0,Y=0,E)_{(\tilde{\cM}_{1/4}\ot\cI_E)(\rho_{RE})}:(\tilde{\cM}_{1/4}\ot\cI_E)(\rho_{RE})_U=(1-s)\proj{0}+s\proj{1},\  s=\sum_{xy}\frac{q((x,y,1))}{4p_{XY}}\biggr\}=\\
  &\biggl\{H(AB|X=0,Y=0,E)_{(\cM_{p_{XY}\ot\cI_E)}(\rho_{RE})}:(\cM_{p_{XY}}\ot\cI_E)(\rho_{RE})_U=\sum_{xyw}q((x,y,w))\proj{(x,y,w)}\biggr\}\,.
\end{align*}

Let $\mathfrak{G}_{\zeta^A,\zeta^B}$ be the set of channels for which $X$ and $Y$ are independent, $X$ is 1 with probability $\zeta^A$ and $Y$ is 1 with probability $\zeta^B$.

\begin{lemma}
  The function $F_{AB|00E}$ as defined in the main text can be used to define a rate function for $\mathfrak{G}_{\zeta^A,\zeta^B}$ by taking
  $\rate_{\zeta^A,\zeta^B}(q)=F_{AB|00E}(\score(q))$ for $q\in\cQ_{\mathfrak{G}_{\zeta^A,\zeta^B}}$
  where
  \begin{equation}\label{eq:sc}
    \score(q)=\frac{1}{4}\sum_{xy}\frac{1}{p_X(x)p_Y(y)} q((x,y,1))\,.
    \end{equation}
\end{lemma}
\begin{proof}
  We have
  \begin{align*}
F_{AB|00E}&(\score(q)):=\!\!\!\inf_{(\tilde{\cM},\rho_{RE})}\!\left\{H(AB|X=0,Y=0,E)_{(\tilde{\cM}_{1/4}\ot\cI_E)(\rho_{RE})}\!:\!(\tilde{\cM}_{1/4}\ot\cI_E)(\rho_{RE})_U=(1-\score(q))\proj{0}+\score(q)\proj{1}\right\}\\
              &=\!\!\!\inf_{(\cM,\rho_{RE})}\!\biggl\{\!H(AB|X=0,Y=0,E)_{(\cM_{p_Xp_Y}\ot\cI_E)(\rho_{RE})}\!:\!(\cM_{p_Xp_Y}\ot\cI_E)(\rho_{RE})_U\!=\!\sum_{xyw}q((x,y,w))\proj{(x,y,w)}\!\biggr\}\\
    &\leq\inf_{(\cM,\rho_{RE})}\left\{H(AB|XYE)_{(\cM_{p_Xp_Y}\ot\cI_E)(\rho_{RE})}:(\cM_{p_Xp_Y}\ot\cI_E)(\rho_{RE})_U=\sum_{xyw}q((x,y,w))\proj{(x,y,w)}\right\},
  \end{align*}
and hence $F_{AB|00E}(\score(q))$ is a rate function for $q\in\cQ_{\mathfrak{G}_{\zeta^A,\zeta^B}}$.
\end{proof}
We can hence form min-tradeoff functions suitable for using with Protocol~\ref{prot:biased} by taking affine lower bounds to $F_{AB|00E}$.  Taking the tangent to $F_{AB|00E}$ at $t$ we have min-tradeoff function
$$f_t(q):=F_{AB|00E}(t)+F'_{AB|00E}(t)\left(\frac{1}{4}\sum_{x,y}\frac{1}{p_X(x)p_Y(y)} q((x,y,1))-t\right),$$
or, in other words, considering deterministic distributions on $U=(x,y,w)$
$$f_t(\delta_{(x,y,w)})=\begin{cases}
    \frac{F'_{AB|00E}(t)}{4p_X(x)p_Y(y)}+F_{AB|00E}(t)-tF'_{AB|00E}(t) & \text{if}\ w=1\\
    F_{AB|00E}(t)-tF'_{AB|00E}(t)& \text{if}\ w=0
    \end{cases}$$

    We have
    \begin{align*}
      \Max(f_t)&=\frac{F'_{AB|00E}(t)}{4\zeta^A\zeta^B}+F_{AB|00E}(t)-tF'_{AB|00E}(t)\\
      \Min_{\cQ_{\mathfrak{G}_{\zeta^A,\zeta^B}}}(f_t)&=F_{AB|00E}(t)-F'_{AB|00E}(t)\left(t-\frac{1}{2}\left(1-\frac{1}{\sqrt{2}}\right)\right)
      \end{align*}
    We now find a bound on $\Var_q(f_t)$ using the Bhatia-Davis bound~\cite{BhatiaDavis}.
\begin{lemma}[Bhatia-Davis bound]
Let $X$ be a real-valued random variable with $\max(X) = M$, $\min(X) = m$ and $\mathbb{E}(X)=\mu$, then
\begin{eqnarray}
\Var_X\leq(M-\mu)(\mu-m)\,.
\end{eqnarray}
\end{lemma}
In our case, $M=\Max(f_t)$, $m=F_{AB|00E}(t)-tF'_{AB|00E}(t)$ and $\mu=\mathbb{E}_q(f_t)=F_{AB|00E}(t)+F'_{AB|00E}(t)\left(\score(q)-t\right)$, where $\score(q)$ is defined in~\eqref{eq:sc}. Thus,
\begin{eqnarray*}
  \Var_q(f_t) &\leq&(F'_{AB|00E}(t))^2\score(q) \left(\frac{1}{4\zeta^A\zeta^B}-\score(q)\right)\\
              &\leq&\begin{cases}(F'_{AB|00E}(t))^2\left(\frac{1}{4\zeta^A\zeta^B}-1\right)&\text{if}\ \zeta^A\zeta^B<1/8\\\left(\frac{F'_{AB|00E}(t)}{8\zeta^A\zeta^B}\right)^2&\text{if}\ \zeta^A\zeta^B\geq1/8\end{cases}
\end{eqnarray*}
where we have optimized over $\score(q)\in[0,1]$ for the second inequality.

\subsubsection{Completeness error}\label{app:comp}
We can form a bound on the completeness error using Hoeffding's inequality~\cite{Hoeffding}.
\begin{lemma}[Hoeffding's inequality] 
Let $X_i$  be $n$ i.i.d.\ random variables with $a \leq X_i \leq b$, $a, b \in \mathbb{R}$. If $S  = \sum_{i} X_{i}$ and $\mu=\mathbb{E}(S)$. Then for $t > 0$
\begin{eqnarray}
\mathbb{P}(S - \mu \geq t) \leq \e^{-\frac{2t^2}{n (b - a)^2 }}\,.
\end{eqnarray}
\end{lemma}
\begin{theorem}
Suppose Protocol~\ref{prot:biased} is run using honest devices that behave in an i.i.d.\ fashion and that have an expected CHSH score $\score_{\exp}$. The probability that the protocol aborts is no greater than
\begin{eqnarray}
 \e^{-32n(\delta\zeta^A\zeta^B)^2}\,.
\end{eqnarray}
\end{theorem}
\begin{proof}
Recall the abort condition in the protocol, which states that $\score<\score_{\exp}-\delta$ where
$$\score=\frac{1}{4}\sum_{x,y} \frac{|\{i:U_i=(x,y,1)\}|}{n p_{X}(x) p_Y(y)}\,.$$
We can write this as $\sum_iJ_i$, where 
\begin{eqnarray}
J_i(x,y,w) &=& \begin{cases} 0& \text{ if } w=0\\
1/(4np_X(x)p_Y(y)) & \text{ if } w=1
\end{cases}
\end{eqnarray}
This construction gives $\mathbb{E}\left[\sum_i J_i\right]=n\mathbb{E}[J_i]=\sum_{xy}\frac{1}{4p_X(x)p_Y(y)}\mathbb{P}(U=(x,y,1))$.
In an honest implementation of the protocol, the distribution on the register $U$ takes the form
\begin{eqnarray}
\mathbb{P}(U=(x,y,w)) = \begin{cases} p_X(x)p_Y(y) (1-\score_{xy}) & \text{ if } w=0 \\
p_X(x)p_Y(y) \score_{xy} & \text{ if } w=1
\end{cases}
\end{eqnarray}
where $\sum_{xy} \score_{xy} = 4 \score_{\exp}$, and hence $\mathbb{E}\left[\sum_i J_i\right]=\score_{\exp}$.
The abort condition can be expressed as $\score_{\exp}-\sum_i J_i>\delta$. We have
\begin{align*}
    \mathbb{P}(\score_{\exp}-\sum_i J_i>\delta)&=\mathbb{P}(\sum_i(-J_i)-(-\score_{\exp})>\delta)\\
    &\leq\e^{-32n(\delta\zeta^A\zeta^B)^2}\,,
\end{align*}
where we have used Hoeffding's inequality for the random variable $-J_i$ with $a=-1/(4n\zeta^A\zeta^B)$ and $b=0$.
\end{proof}

\subsection{Error parameters}

Both Theorems~\ref{thm:EAT} and~\ref{thm:EATspot} are stated in terms of the probability that the protocol does not abort, $p_\Omega$, which is unknown to the users of the protocol. However, if we replace $p_\Omega$ by $\epsilon_{\text{EAT}}$, then if $p_{\Omega}\geq\epsilon_{\text{EAT}}$ we have a correct bound on the entropy. On the other hand, if $p_{\Omega}<\epsilon_{\text{EAT}}$ then the protocol aborts with probability greater than $1-\epsilon_{\text{EAT}}$.  In other words, prior to running the protocol the probability that it will both not abort and that the entropy is not valid is at most $\epsilon_{\text{EAT}}$.  The soundness error of the protocol is $\epsilon_S=\max(\epsilon_{\text{EAT}},2\epsilon_h+\epsilon_{\text{EXT}})$, where $\epsilon_{\text{EXT}}$ is the extractor error (essentially the probability that the extraction fails). A summary of the aspects of extraction relevant to the present discussion and in the same notation as used here can be found in~\cite[Supplementary Information I~C]{LLR&}.

\subsection{Application to $H(AB|E)$ and $H(A|E)$}
Note that the EAT as stated in Theorem~\ref{thm:EAT} cannot be directly used in conjunction with $H(AB|E)$ and $H(A|E)$. The basic reason is that the event $\Omega$ should be an event on ${\bf U}$, which in turn should be a deterministic function of ${\bf C}$ and ${\bf D}$.  To use $H(AB|E)$ and $H(A|E)$ we need ${\bf D}$ to be empty and ${\bf C}$ to be ${\bf AB}$. This means the register ${\bf U}$ cannot depend on the inputs, ${\bf XY}$, but without a score that depends on the inputs we cannot certify non-classicality let alone randomness.

Since we do not have strong use cases for $H(AB|E)$ and $H(A|E)$ (cf.\ Section~\ref{sec:signif}), we do not consider possible extensions of the EAT in this work.

An alternative, which loses tightness, is to use an idea from~\cite[Appendix~B.3]{MDRHW}. Applying to the present case this would mean taking ${\bf D}$ to be empty and ${\bf C}$ to be either ${\bf ABV}$ or ${\bf AV}$, where $V_i$ records whether the CHSH game was won on the $i$th round, with $U_i=V_i$.  Then, proceeding with the former, because $H(ABV|E)\geq H(AB|E)$ we can base our min-tradeoff function on $H(AB|E)$, and we can use a chain rule to recover a bound on the smooth min entropy of ${\bf AB}$ given $E$ from that of ${\bf ABV}$ given $E$.  The bounds used in this approach are tightest when $V$ has low entropy, so we expect better performance with spot-checking protocols. 

\section{Discussion of composability}\label{app:compos}
Throughout this work we consider a composable security definition. These involve a distinguisher who tries to guess whether the real protocol or a hypothetical ideal protocol is being run. This distinguisher is allowed access to all the systems an eavesdropper has access to and is also assumed to learn whether or not the protocol was successful. The idea is that no matter what strategy the distinguisher uses, before the protocol is run the probability is at most $1/2+\epsilon_S$ that they can correctly guess whether the real protocol or the ideal is being run.

The main purpose of this appendix is to make a few remarks on composability for protocols that recycle the input randomness. In general, input randomness (the strings ${\bf X}$ and ${\bf Y}$) is not directly reusable without processing~\cite{CK2}. For instance, the devices could be set up such that the protocol aborts unless $X_1=0$ and so if the protocol passes it is known that $X_1=0$.  If ${\bf X}$ directly forms part of the output, then with probability $1/2$ one bit of the final output is known, which contradicts the security statement which implies that the {\it a priori} probability (i.e., the probability before the protocol is run) of being able to distinguish the protocol from an ideal one that either aborts or gives out perfect randomness is at most the soundness error. Hence, in order to recycle the input randomness, it also has to undergo extraction to remove possible information that may have leaked about it.

Because we are working with device-independent protocols, the ongoing security of any randomness generated can be compromised if the devices used for one instance of the protocol are subsequently reused~\cite{bckone}. Hence, our discussion of security assumes devices are not reused (possible modifications to protocols that aim to allow restricted reuse are also discussed in~\cite{bckone}).


\end{document}